\documentclass[aps,pra,10pt,showpacs,notitlepage,onecolumn,superscriptaddress,nofootinbib]{revtex4-1}
\usepackage[tmargin=1in, bmargin=1in, lmargin=1.5in, rmargin=1.5in]{geometry}
\usepackage{amsmath, amssymb, amsthm}
\usepackage{graphicx}
\usepackage[svgnames]{xcolor}
\usepackage{enumitem}
\usepackage{titlesec}
\usepackage{datetime}
\usepackage{hyperref}
\usepackage{import}
\usepackage{mathtools}
\usepackage{tikz}
\usepackage{epigraph}

\usepackage{makecell} % For line breaks within table.
\setcellgapes{2pt}

\usepackage{twemojis}

\definecolor{crimson}{RGB}{186,0,44}
\definecolor{moss}{RGB}{0, 186, 111}

\hypersetup{
    colorlinks,
    linkcolor={crimson},
    citecolor={crimson},
    urlcolor={crimson}
}

\theoremstyle{definition}
\newtheorem{definition}{Definition}[section]
\newtheorem{lemma}{Lemma}[section]
\newtheorem{theorem}{Theorem}[section]
\newtheorem{corollary}{Corollary}[theorem]
\newtheorem{remark}{Remark}[section]

\newtheorem{example}{Example}[section]

\newtheorem{problem}{Problem}[section]
\newtheorem{proposition}{Proposition}[section]

\newtheorem*{definition*}{Definition}
\newtheorem*{lemma*}{Lemma}
\newtheorem*{theorem*}{Theorem}
\newtheorem*{corollary*}{Corollary}

\begin{document}

%%%%%%%%%%%%%%%%%%%%%%%%%%%%%%%%%%%%%%%%%%%%%
\title{Quantum element-wise transforms}
\author{Zane M.~Rossi}
\email[]{zmr@g.ecc.u-tokyo.ac.jp}
\affiliation{Department of Physics, Graduate School of Science, The University of Tokyo, Hongo 7-3-1, Bunkyo-ku, Tokyo 113-0033, Japan}
\author{Rahul Sarkar}
\email[]{rsarkar@berkeley.edu}
\affiliation{Department of Mathematics, University of California, Berkeley, CA 94720, USA}

% %%%%%%%%%%%%%%%%%%%%%%%%%%%%%%%%%%%%%%%%%%%%%
\begin{abstract} 
    \noindent Quantum algorithms for basic numerical linear algebraic tasks have proven essential for translating diverse problems to a unified quantum computational context. 
    Many of these tasks---e.g., applying a function to the spectrum of a matrix embedded in a unitary process (a so-called \emph{block encoding}), or taking linear combinations of block encodings---are well-addressed by techniques like quantum singular value transformation or linear combination of unitaries.
    However, there exist useful matrix transforms whose realization is unclear or inefficient.
    In this work we construct improved quantum algorithms for some of these transforms, the simplest of which is a \emph{polynomial function applied element-wise}.
    We show the space required to compute \emph{quantum element-wise transforms} (QEWTs) can be reduced exponentially in the degree of the applied function compared to prior work, and rectify errors in previous constructions. 
    We prove novel space and query lower bounds for QEWTs, construct useful function families nearly saturating these bounds, and establish the hardness of QEWT dequantization. This represents the first sound theoretical treatment of QEWTs, and brings their resource requirements on par with other practical block encoding techniques.
    We present our algorithms alongside concrete applications to machine learning inference and signal processing.
\end{abstract}

%%%%%%%%%%%%%%%%%%%%%%%%%%%%%%%%%%%%%%%%%%%%%
\maketitle

\vspace*{-1.5\baselineskip} % For more pleasant spacing.

% Inclusion for clarity of review.
\tableofcontents

%%%%%%%%%%%%%%%%%%%%%%%%%%%%%%%%%%%%%%%%%%%%%
\section{Introduction}

% \begin{epigraphs}
%     \qitem{``Numerical linear algebra is really functional analysis, but with the emphasis always on practical algorithmic ideas rather than mathematical technicalities...''}{\emph{Numerical Linear Algebra}, Trefethen \& Bau \cite{tb_num_alg_22}}
% \end{epigraphs}

\noindent The field of classical algorithms---and the art of running these algorithms on real-world devices---reveals, wherever one stoops to flip over a rock, a rich sub-structure of numerical linear algebra \cite{tref_approx_19, tb_num_alg_22}. The modern classical numerical linear algebraic toolkit indeed encourages most researchers, as a first attempt, to linearize their problem and determine whether existing methods yield acceptable solutions or useful insights. The success of this scheme is self-evident; matrix methods pervade signal processing, differential equations, statistics, machine learning, and the discrete analogues of huge swaths of functional analysis.

The field of quantum algorithms has warmed to a similar scheme. One manifestation is the increasing use of the \emph{block encoding} (Def.~\ref{def:be_formal}) input model, which embeds a matrix as a sub-block of a larger unitary process run on a quantum computer. The permissible transformations of these block-encoded matrices (which represent non-linear subsystem dynamics) inform the methods and resource complexities for mimicking a diverse family of known quantum algorithms within a unified, numerically well-understood framework \cite{gslw_19, mrtc_unification_21}. Crucially, for certain common transformations, and for block-encoded matrices obeying certain constraints, these algorithms appear to be able to prepare states, sample from distributions, and solve discrete problems far more efficiently than their classical cousins (e.g., for certain $N$-dimensional problems, the required space or time for a quantum algorithm can scale as $\text{polylog}(N)$).

In this work we expand the quantum numerical linear algebraic toolkit to matrix transformations not easily captured by current quantum algorithms. Concretely, we provide improved circuits for computing non-standard matrix products (e.g., the \emph{element-wise product}), achieving exponential (in the degree of the product) space reductions compared with prior work. These non-standard products, while less studied in quantum computing, appear ubiquitously in classical signal processing, machine learning, and statistics (see Sec.~\ref{subsec:prior_work}). This provides a new tool, with resource requirements on par with other block encoding algorithms, for efficiently computing non-linear transformations of high-dimensional data.

We construct quantum algorithms for computing various matrix products under differing input assumptions, but use the rest of this section to define the simplest example (Sec.~\ref{subsec:notation} lays out our notation), highlighting how our matrix transformations differ from those usually considered for algorithms like quantum singular value transformation (QSVT) \cite{gslw_19}; a brief introduction to the suite of previously known block-encoding-based quantum algorithms (many of whose methods we repurpose) appears in Sec.~\ref{subsec:be_basics}.

\begin{definition}[Singular-value transformation of a matrix] \label{def:sv_func}
    Let $M \in \mathbb{C}^{n\times n}$ a matrix with $\lVert M \rVert \le 1$; then the \emph{singular-value transformation} (SVT) of $M$ according to a function (typically continuous) $f: [0,1] \rightarrow [0,1]$ is:
    \begin{equation}
        f^{\text{(SV)}}(M)
        \equiv
        Uf(\Sigma) V^{\dagger}
        =
        U
        \begin{bmatrix}
            f(\xi_{0}) & & \\
            & \ddots & \\
            & & f(\xi_{n-1})\\
        \end{bmatrix}
        V^\dagger, \quad \Sigma = \text{diag}(\xi_0, \xi_1, \dots, \xi_{n-1}),
    \end{equation}
    where $M = U\Sigma V^\dagger$ is the singular value decomposition (SVD), and where we have restricted to the square case (otherwise one can pad $\Sigma$ with additional rows or columns of zeros, and increase the dimensions of the unitaries $U$ and $V$ to keep the product conformable).
\end{definition}

QSVT, whose main results are introduced and employed in Sec.~\ref{sec:main_const}, provides a method to compute the SVT of exponentially large matrices using simple circuits whose depth scales as the degree of a polynomial approximating $f$ to a desired precision. In our work we will instead focus on a well-defined, commonly applied, but comparatively less well-studied matrix transformation wherein a function is applied \emph{element-wise} to $M$. Before defining this transformation we quickly disambiguate between two types of matrix product.

\begin{definition}[The matrix and element-wise products] \label{def:matrix_prod}
    Throughout this work we refer to a number of non-standard products between matrices (a table for our notation is given in Sec.~\ref{subsec:notation}). The most well-known of these is the standard \emph{matrix product}, which takes $A \in \mathbb{C}^{n\times m}$ and $B \in \mathbb{C}^{m\times p}$ and returns $C = AB \in \mathbb{C}^{n\times p}$ following $C_{ij} = A_{ik} B_{kj}$, where repeated indices are summed over. Less commonly, and here referred to as the \emph{element-wise product}, one can take two matrices $A$ and $B$ of identical shape and compute $C_{ij} = A_{ij} B_{ij}$, denoted\footnote{Here, and confusingly across literature, the repeated indices $i, j$ are \emph{not} summed over, but merely denote component-wise access. Where possible we'll try to avoid Einstein summation notation.} $C = A \circ B$.
\end{definition}

The element-wise product, like its matrix product cousin, is associative and distributes over matrix addition; moreover, it is also commutative and an inverse for $\circ$ exists for a matrix $A$ if and only if all of its elements are non-zero. These properties, among numerous related ones, are given treatment in Ch.~5 of the standard reference \cite{hj_topics_91}. Definitions or results we make heavy use of related to element-wise products are glossed in Appx.~\ref{appx:ewt_prop}. This product can be immediately extended, given its pleasant properties, to the notion of an \emph{element-wise function} applied to a matrix.

\begin{definition}[Element-wise function of a matrix] \label{def:elem_func}
    Let $M \in \mathbb{C}^{n\times m}$; then the element-wise application of a function $f: \mathbb{C} \rightarrow \mathbb{C}$ to $M$, often denoted $f^{\circ}(M)$, is simply:
        \begin{equation}
            \begin{bmatrix}
                a_{00} & a_{01} & \cdots & a_{0,m-1}\\
                a_{10} & a_{11} & & \vdots \\
                \vdots & & \ddots &\\
                a_{n-1,0} & \cdots & & a_{n-1,m-1}
            \end{bmatrix}
            \mapsto
            \begin{bmatrix}
                f(a_{00}) & f(a_{01}) & \cdots & f(a_{0,m-1})\\
                f(a_{10}) & f(a_{11}) & & \vdots \\
                \vdots & & \ddots &\\
                f(a_{n-1,0}) & \cdots & & f(a_{n-1,m-1})
            \end{bmatrix}.
        \end{equation}
    Equivalently, we could express this element-wise function using index notation: $a_{ij} \mapsto f(a_{ij})$.
\end{definition}

It is this transformation (and its constituent parts and extensions) that will be the main subject of this work. We will generically refer to quantum algorithms achieving this and related transforms of block-encoded matrices as \emph{quantum element-wise transforms} (QEWTs),\footnote{Pronounced like \emph{newt} \texttwemoji{lizard}.} joining an already-large family of acronyms and initialisms for block encoding algorithms (the most relevant of these appear in Sec.~\ref{subsec:prior_work}).

Unlike singular value transformations (and related eigenvalue or Jordan-block transformations), which have been the focus of previous quantum algorithms research, these element-wise transformations exhibit numerous unusual properties, including basis-dependence and non-obvious effects on the matrix norm (and rank, positive-definiteness, condition number, etc.). Correspondingly the space-, query-, and gate-complexities, as well as success probability upper- and lower-bounds for quantum algorithms computing element-wise functions of matrices remain incompletely characterized. In the following section we provide an overview of our main results (Sec.~\ref{subsec:main_results}), discussion of prior work on adjacent topics and related quantum algorithms (Sec.~\ref{subsec:prior_work}), and finally a quick-start guide to our notation, which for this topic especially is inhomogeneous across literature (Sec.~\ref{subsec:notation}).

%%%%%%%%%%%%%%%%%%%%%%%%%%%%%%%%%%%%%%%%%%%%%
\subsection{Main results} \label{subsec:main_results}

\noindent We construct a quantum circuit that uses $d$ copies of a unitary block encoding of a matrix and produces a block encoding of a desired degree-$d$ polynomial applied element-wise (in a chosen basis) to that matrix. The cost of this transformation (as in previous work) can be measured along multiple axes, including query complexity, gate complexity, auxiliary space use, block encoding subnormalization, and success probability. 

For brevity we present a table (Tab.~\ref{tab:resource_comp}) comparing our construction to the prior work in terms of query complexity and auxiliary space use, and leave detailed discussion of our construction's full attributes to the main text. We also note that while our work and prior work give constructions of a `sequential' and `parallel' flavor respectively, both have linear depth in $d$, and success probability that can decrease (at worst) exponentially in $d$. This is a generic feature of quantum algorithms that multiply $d$ (possibly ill-conditioned) block encodings, and we enumerate various sufficient assumptions (common in prior work considering similar subroutines \cite{lw_comp_gadget_19, flt_time_marching_23, vg_reducing_space_25}) to avoid worst-case behavior.

\begin{theorem}[Quantum element-wise transform; \emph{condensed}] 
\label{thm:qewt_informal}
    Let $f$ be a polynomial of degree $d$ with no constant term satisfying $|f| \leq 1$ on $[-1,1]$, with coefficients $c_{k}, k \in \{1, \dots, d\}$. Then, given access to controlled $U$, where $U$ is an $(\alpha, a, \varepsilon)$ block encoding of an $n$-qubit operator $A$, there exists a quantum circuit preparing a $(\beta, b, \delta)$ block encoding of $f^\circ (A/\alpha)$---\emph{the element-wise application of $f$ to $A/\alpha$}---which requires $\mathcal{O}(d)$ copies of controlled $U$ and single\slash two-qubit gates, where $(\beta, b, \delta)$ satisfy
        \begin{equation}
            \beta = \sum_{k = 1}^{d} |c_k|,
            \quad\;\;
            b = \mathcal{O}(a + n\log{d}),
            \quad\;\;
            \delta = \sum_{k = 1}^{d} |c_k|\,\Big[(1 + \varepsilon/\alpha)^k - 1\Big].
        \end{equation}
\end{theorem}

While the top-line result of this work is an improved method for quantum element-wise transformation of block encodings, our construction comprises numerous subroutines that may be of independent interest. Principally these are (1) a novel technique for introducing a small amount of additional space when processing block encodings with block structure to cheaply copy sub-blocks in fixed bases (which we call the `weaving lemma' for its multi-time use of known, catalytic quantum states that can be `woven' through a quantum circuit), (2) the modification of quantum circuits for element-wise products such that we can apply previously-known `compression gadgets' for reducing space requirements when computing matrix products of block encodings, and (3) the use of simple, LCU-based subroutines and variants of amplitude amplification to improve success probability and gate complexity simultaneously.
These techniques prove remarkably flexible, and we devote a section (Sec.~\ref{sec:alt_const}) to their modification and improvement under specific assumptions.

{
% For providing some breathing room in tables; applying only locally
\renewcommand{\arraystretch}{1.5}
\begin{table}
    \centering
    \begin{tabular}{l | l l l}
         Method & $\;\;$Query comp.\; & Aux.~space & Notes \\\hline
         Prior work \cite{gychanar_had_prod_25} & $\;\;\mathcal{O}(d)$ & $\tilde{\mathcal{O}}(ad + nd)$ & $\mathcal{O}(d)$ depth\\
         Main result (Thm.~\ref{thm:qewt}) & $\;\;\mathcal{O}(d)$ & $\mathcal{O}(a + n(\log{d}))$ & $\mathcal{O}(d)$ depth\\
         LCU variant (Thm.~\ref{thm:lcu_alt_ewt_func})$\;\;$ & $\;\;\mathcal{O}(d)$ & $\mathcal{O}(a + n(\log{d}))$ & Reduced space\\
         Approx. variant (Cor.~\ref{cor:approx_ewt_func})\;\; & $\;\;\mathcal{O}(d^{\log{(\log{(d/\varepsilon)})}})\;\;$ & $\mathcal{O}(a + n\log{(d^{\log{(\log{(d/\varepsilon)}})}))}\;\;$ & Input assump.
         \\
         Half-compressed (Thm.~\ref{thm:half_qewt})\;\; & $\;\;\mathcal{O}(d)$ & $\tilde{\mathcal{O}}(a + nd)\;\;$ & Simple circuit
    \end{tabular}
    \caption{High-level overview of the resource complexities of our main result, its variants, and prior work. The problem considered is the element-wise application of a polynomial of degree $d$ given an $n$-qubit matrix block-encoded using $a$ qubits of auxiliary space. The simplest construction we provide (Thm.~\ref{thm:qewt}) can be further space-compressed by an LCU-based variant, as well as cheaply amplified (with success probability greater than $1 - \varepsilon$, to accuracy $\varepsilon$) under certain assumptions on the block-encoded input. The referenced theorems and corollaries provide far more detailed non-asymptotic resource complexities, associated circuits, and precise error/success probability definitions.
    }
    \label{tab:resource_comp}
\end{table}
}

None of the individual techniques in this work is particularly complex, though their assemblage introduces a bit of index bookkeeping. To build intuition for our construction we provide an overview of salient block encoding manipulation techniques (Sec.~\ref{subsec:be_basics}), then cover a small-scale example introducing the intended action of each subroutine (without their realizing circuit) (Sec.~\ref{subsec:minimal_ex}), followed by a sequential construction of element-wise products followed by element-wise functions (Secs.~\ref{subsec:swap_copy_weaving} and \ref{subsec:full_const}). Finally we discuss the quantum element-wise transform in generality, introducing more involved variant constructions with improved properties in special cases (Sec.~\ref{sec:alt_const}).
Along the way we clarify subtleties under-emphasized in previous work, as well as errors in these works' application of LCU-based techniques, providing the first comprehensive and correct treatment of the theory of these \emph{quantum element-wise transforms}, as well as their application to existing classical and quantum problems (Sec.~\ref{sec:applications}).
Beyond our construction and its applications, we also provide a detailed discussion (Sec.~\ref{sec:discussion}), including a series of novel lower bounds on the space, query, depth, and block encoding subnormalization required for any quantum algorithm computing QEWTs of block encodings (Sec.~\ref{sec:lower_bounds}, summarized in Tab.~\ref{tab:lower_bounds}), as well as proofs that standard methods for dequantizing block encoding algorithms do not apply to QEWTs composed with spectral maps (Sec.~\ref{sec:dequant}). 
Together this establishes a solid theoretical foundation for the properties of QEWTs which, in light of their now modest resource requirements (brought on par with those of QSVT/LCU for spectral maps), represent a substantively distinct, useful, and near-term implementable block encoding-based quantum algorithmic primitive.
We close out our discussion by highlighting a few intriguing open problems (Sec.~\ref{subsec:open_dir}).

%%%%%%%%%%%%%%%%%%%%%%%%%%%%%%%%%%%%%%%%%%%%%
\subsection{Prior work} \label{subsec:prior_work}

    \noindent Below we introduce some work related to our construction along the following axes: (1) the theory and common uses of element-wise products and functions of matrices\slash tensors, (2) existing techniques for realizing element-wise transformations with quantum algorithms (and their application), and finally (3) parallel work in quantum algorithms for computing other non-standard matrix functions (or functions restricted to sub-classes of matrices allowing for significant resource improvements).
    
    An accessible treatment of the theory of element-wise products can be found in the well-known textbook of Horn and Johnson \cite{hj_matrix_analysis_85} (which uses the term \emph{Hadamard product}); this product is also given its own chapter (Ch.~5) in their topics textbook \cite{hj_topics_91}, which sketches its common uses, e.g., for computing convolutions among periodic functions, when studying integral equation kernels (with applications to Mercer's theorem), as related to the weak minimum principle in PDEs, and in the study of characteristic functions (related to Bochner's theorem). Horn and Johnson also gloss applications to combinatorics (e.g., association schemes), computing covariance matrices, control theory (e.g., the study of the relative gain array), and analyzing solutions to various generalized matrix equations (e.g., the Lyapunov equation). The study of fundamental analytic properties of the element-wise product (and its block-matrix variants \cite{johnson_had_mat_74, kr_prod_68, tk_tr_prod_app_89, liu_ts_kr_prod_99, hc_block_kronecker_89}) applied to structured families of matrices is active and ongoing, though these technical treatments are beyond the scope of this work. 

    The element-wise product appears commonly throughout classical signal processing, statistics \cite{mn_had_prod_textbook_99, styan_had_prod_73, kn_block_kron_vecb_91}, the analysis of certain matrix and integral equations \cite{kr_prod_68, ts_prod_72}, as well as ubiquitously in machine learning and tensor arithmetic \cite{tensor_decomp_09, learning_had_prod_25}. Much of these latter applications center on this product's ability to cheaply introduce nonlinearities into tensor manipulations, as well as compute useful approximations to certain tensors (e.g., the canonical polyadic decomposition: CPD). From this standpoint the element-wise product could be seen as an essential sub-routine for tensorial algorithmic analogues to singular value decomposition\slash transformation (noting, to temper this analogy, that many problems in the computation of useful tensorial properties are known to be extremely difficult \cite{hl_tensor_np_13}). In any event, at a practical level, the cheap computation of element-wise transforms in a dimension-efficient way by quantum computers offers numerous doors into the study of wide sub-fields of applied numerical linear algebra.

    The problem of computing element-wise products between representations of linear operators encoded in quantum processes has been considered previously \cite{zzrf_basic_q_lin_alg_21, gychanar_had_prod_25, dlx_prod_be_25, ltlf_matrix_prod_quantum_26} in varying access models, and has some connection to the related problem of applying non-linear functions to quantum state amplitudes \cite{gmkf_nonlinear_24, rr_nonlinear_23, hcss_nonlinear_23}. In general the work of applying element-wise functions to block encodings, as noted in \cite{rr_nonlinear_23} is incomparable to the state case, which relies on a crucial coincidence between various matrix products under diagonality assumptions.
    Simultaneous work has also considered numerous applications downstream of these transforms, e.g., the computation of convolutions \cite{rhgpr_acc_inf_25}. Nevertheless, the constructions given in these works (by-and-large un-optimized and treated only in passing) leave the study of element-wise transformations quite open and, as we show, subject to significant improvement and generalization. 

    More generally the theory of (and accompanying numerical tools for) block-encoding quantum algorithms has expanded rapidly over the past ten years. This includes quantum signal processing (QSP) \cite{lyc_optimal_pulses_14, lyc_equiangular_16, lc_ham_sim_17, lc_qubitization_19} and its lifted version quantum singular value transformation (QSVT) \cite{gslw_19}, which have been the main focuses of pedagogical introductions\footnote{For the basic results we use in this work, see the treatment of Sec.~\ref{subsec:be_basics}.} to the subfield \cite{mrtc_unification_21, cs_qsvt_tang_tian_23}, and whose numerical properties (e.g., results on the convergence of classical algorithms for computing the \emph{QSP phases}) have received extensive treatment\footnote{Including recent work tethering the theory of QSP\slash QSVT and their variants to beautiful methods in the comparatively niche subfield of \emph{non-linear Fourier analysis} \cite{amt_23, szego_nlfa_qsp_24}.} with accompanying public software \cite{haah_decomposition_19, cdghs_finding_qsp_angles_20, wdl_sym_qsp_energy_22, dmwl_efficient_phases_21, dlnw_robust_iter_24}.
    QSP and QSVT have themselves spawned a complex landscape of variants and extensions, including the multivariable setting (M-QSP)  \cite{rc_m_qsp_22}, the indefinite parity setting (G-QSP) \cite{mw_gqsp_24}, methods for eigenvalue transformations \cite{ls_qep_26, gls_qet_arb_26}, modular versions \cite{rcc_modular_qsp_23}, extensions to continuous variables \cite{rbmc_su_11_qsp, bw_analytical_phase_26}, and the practical use of randomization \cite{martyn_rall_halving_24} and parallelization \cite{mrclc_parallel_qsp_24}.
    Moreover, beyond QSP and QSVT, which transform block encoding spectra, other algorithms for block encoding arithmetic have found diverse success, including linear combination of unitaries (LCU) for Hamiltonian simulation \cite{cw_lcu_12, bccks_ham_lcu_14, bck_ham_lcu_15}, linear combination of Hamiltonian simulation (LCHS) for simulating general PDEs \cite{dll_lchs_state_cost_23, acl_lchs_near_opt_23, ls_opt_lchs_25}, and the combined use of QSVT and LCU for algorithms approximating contour integrals as applied to simulation and linear systems solving \cite{lw_mat_eq_25, sbs_lin_mat_diffeq_26}.

    The constructions above have also yielded independently interesting and widely applicable methods for reducing space use in block encoding algorithms \cite{lw_comp_gadget_19, vg_reducing_space_25}, increasing success probability, and leveraging known problem structure to improve resource requirements both above the block-encoding layer \cite{zs_low_energy_24, fast_spec_amp_25, yuan_cobble_25} and for compiling desired block encodings down to the gate-level \cite{nkl_dense_rank_be_22, cv_fable_be_22, scc_be_struct_24, clvy_q_circ_be_24}.
    Together these improvements compose a growing toolkit of flexible, fungible quantum algorithms for the manipulation of linear operators, with great potential (as we begin to implement these subroutines on real devices with specific architectural constraints) in hybridization, heuristic-led optimization, and approximate or `fuzzy' extensions.

\subsection{Notation and terminology} 
\label{subsec:notation}

\noindent We denote the set of positive integers as $\mathbb{N}$, the set of positive real numbers as $\mathbb{R}_{+}$, and $\log$ will always denote logarithm base $2$, unless stated otherwise. The complex Hilbert space associated with $n$ qubits has dimension $2^n$; any linear operator acting on this Hilbert space is called an \textit{$n$-qubit operator}, following the convention of \cite{gslw_19}, and is represented by a $2^n \times 2^n$ matrix over $\mathbb{C}$. All matrices and tensors, when their shape or dimension is given, are taken over $\mathbb{C}$ unless otherwise noted; we sometimes use the notation $M_n$ or $M_{n, m}$ for matrices of size $n\times n$ or $n\times m$ over $\mathbb{C}$. For a vector $v \in \mathbb{C}^m$ we index its components $v_0,\dots,v_{m-1}$, and denote the vector and matrix $p$-norms by $\lVert v \rVert_p$ and $\lVert A \rVert_p$ respectively for all $p \ge 1$. We often use the operator norm induced by the vector $2$-norm (the \textit{spectral norm}) which we denote $\lVert A \rVert$. Where specific notation is required, it will be introduced \emph{in situ}.

A focus of this work is the efficient computation of various matrix and vector products by quantum algorithms. The names and notation of these products are hopelessly inconsistent across prior work, both across and within disciplines. Below we provide a legend for the notational convention we follow in this work (all standard with the \LaTeX{} \texttt{amssymb} package):
\begin{alignat}{3}
    &\otimes &\quad &\text{\texttt{\textbackslash otimes}} &\quad &\text{(Kronecker/tensor product)},\\
    &\circ &\quad &\text{\texttt{\textbackslash circ}} &\quad &\text{(Element-wise product)},\\
    &\boxtimes &\quad &\text{\texttt{\textbackslash boxtimes}} &\quad &\text{(Tracy--Singh product; Def.~\ref{def:ts_prod})},\\
    &\odot &\quad &\text{\texttt{\textbackslash odot}} &\quad &\text{(Khatri--Rao product; Def.~\ref{def:kr_prod})},\\
    &\ast &\quad &\text{\texttt{\textbackslash ast}} &\quad &\text{(Convolution)},\\
    &\cdot/\times &\quad &\text{\texttt{\textbackslash cdot/\textbackslash times}} &\quad &\text{(Matrix product/block-encoded product)}.
\end{alignat}
For the usual matrix product we will usually use no symbol at all, and the non-standard use of $\times$ for the matrix product of \emph{block encoded} matrices is repeated explicitly where relevant to avoid confusion.

We exclusively employ the term \emph{element-wise product} to refer to the operation in Def.~\ref{def:matrix_prod}, which is elsewhere variously referred to as the \emph{Hadamard}, \emph{Schur}, \emph{entry-wise}, or \emph{point-wise} product, which we do not use given their collision with common constructions in quantum information. For the products mentioned above we will provide further in-line notes when defining related variants.

\begin{remark}[On matrix subscript conventions] \label{rem:mat_conv}
    We often make reference to unitary matrices that encode other matrices as sub-blocks, where these sub-blocks themselves have special block structure. This introduces a necessity for non-standard subscript uses, e.g., we may refer to a block encoding unitary (Def.~\ref{def:be_formal}) as $U_{A, ij, r_1 r_3}$.

    In this case we will always use upper-case Latin letters to denote the block encoded operator, doubled lower-case Latin indices to denote matrix elements, and $r_k$ subscripts to denote the qubit registers on which a linear operator acts. Sometimes we will also use single Latin indices to denote lists of operators, e.g., $U_{k}$ for $k \in \{0, 1, \dots, n\}$.

    As expected, the tensor product of two unitaries on two registers acts on their concatenation: $W_{r_1 r_2} \equiv U_{r_1} \otimes V_{r_2}$; for non-unitary matrices (e.g., a linear combination of unitaries all acting on a set of registers) we use the same subscript conventions. In all cases we follow the standard practice in quantum information where tensor products are considered equal up to the canonical isomorphism induced by re-ordering registers.
\end{remark}

We take liberties when denoting operations where the registers on which they act are obvious, e.g., following common conventions we sometimes denote the all-zeros state on $n$ qubits as $|0^n\rangle$, or even simply $|0\rangle$, and the $n$-qubit identity matrix as $I_n$. Additionally, when only a subset of the indices given in Rem.~\ref{rem:mat_conv} is provided, we will sometimes insert additional commas for clarity.

%%%%%%%%%%%%%%%%%%%%%%%%%%%%%%%%%%%%%%%%%%%%%
\section{Main construction} \label{sec:main_const}

\noindent The core contribution of our work involves the interaction of straightforward but carefully defined subroutines for transforming matrices that have been \emph{block-encoded} as sub-blocks of unitary processes. To smoothly introduce these subroutines (and the subtleties of their composition) we briefly review basic facts and methods related to block encodings (Sec.~\ref{subsec:be_basics}), followed by the aforementioned small-scale example of the \emph{action} of our construction (Sec.~\ref{subsec:minimal_ex}). We then motivate and explicate, in isolation, constructions enabling these \emph{actions} in the general case (Sec.~\ref{subsec:swap_copy_weaving}). These constructions comprise most of the novelty of this work, and are bootstrapped to our full \emph{quantum element-wise transform} (Sec.~\ref{subsec:full_const}) mainly through the application of well-known block encoding techniques. Throughout this section we discuss quirks of element-wise products (and the quantum algorithms realizing them), raising and rectifying gaps in previous work on the subject.

%%%%%%%%%%%%%%%%%%%%%%%%%%%%%%%%%%%%%%%%%%%%%
\subsection{Block encoding basics} \label{subsec:be_basics}

\noindent Before continuing to our main construction we briefly outline some common definitions and constructions for quantum algorithms manipulating block encodings; we focus on concision, covering only those techniques that appear in this work. For general, accessible references to this class of quantum algorithms we refer the reader to standard entry-points \cite{gslw_19, mrtc_unification_21, cs_qsvt_tang_tian_23}.

\begin{definition}[Block encoding; cf. \cite{gslw_19}] \label{def:be_formal}
    Let $A$ be an $n$-qubit operator, $\alpha,  \varepsilon \in \mathbb{R}_{+}$, and $a \in \mathbb{N}$. We say an $(n + a)$-qubit unitary $U$ is an $(\alpha, a, \varepsilon)$ \emph{block encoding} of $A$ if
        \begin{equation}
        \label{eq:be_formal}
            \lVert A - \alpha (\langle 0^a| \otimes I_{n}) U (| 0^a\rangle\otimes I_{n}) \rVert \leq \varepsilon.
        \end{equation}
    Note that as $\lVert U \rVert \leq 1$, we have $\lVert A \rVert \leq \alpha + \varepsilon$, by the triangle inequality: i.e., that the top-left block of $U$ contains an ($\varepsilon$-approximate) $\alpha$-subnormalized copy $\tilde{A}$ of $A$:
        \begin{equation}
            U =
            \begin{bmatrix}
                \tilde{A}/\alpha & \;\ast\;\\
                \ast & \ast
            \end{bmatrix}
            \quad\implies\quad
            \tilde{A} = \alpha (\langle 0^a| \otimes I_{n}) U (| 0^a\rangle\otimes I_{n}),
        \end{equation}
    with $\lVert \tilde{A} - A \rVert \le \varepsilon$. The quantity $\alpha$ in \eqref{eq:be_formal} is known as the \textit{subnormalization factor}. In the case that $A$ is non-square (or has dimensions not a power-of-two), this definition can be recovered by padding the rows and columns of $A$ with zeros.
\end{definition}

A natural question is whether the class of efficiently constructible block encodings is interesting, as generically specifying exponentially-many matrix elements is prohibitively expensive. The originating reference \cite{gslw_19} enumerates multiple settings (e.g., sparse oracles, state preparation unitaries, free Hamiltonian evolution) where block encoding access is natural and efficient.

Moreover, block encodings, by merit of encoding non-unitary behavior within a larger unitary process, turn out to be incredibly useful for subsuming diverse techniques in numerical linear algebra and quantum algorithms generally. In what follows we enumerate common operations and transformations of block encodings.

\begin{definition}[Block encoding multiplication]  \label{def:be_mult}
    Let $U$ and $V$ be $(\alpha, a, \varepsilon)$ and $(\beta, b, \delta)$ block encodings of $n$-qubit operators $A$ and $B$, respectively. Then $(I_{r_1}\otimes U_{r_2 r_3})(I_{r_2}\otimes V_{r_1 r_3})$ is a $(\alpha\beta, a + b, \alpha\delta + \beta\varepsilon + \varepsilon\delta)$ block encoding of the \emph{matrix product} $AB$. Here we have defined qubit registers $r_1, r_2, r_3$, which are respectively of size $b, a, n$. The additional space (register) used by one block encoding unitary (either $a$ or $b$ qubits, contained in $r_2$ or $r_1$) is left alone by the other.
\end{definition}

\begin{definition}[State preparation pair] \label{def:state_prep_pair}
Let $y \in \mathbb{C}^m$ and $\lVert y\rVert_1 \leq \beta$. For $b \ge \log{m}$, a pair of $b$-qubit unitaries $P_L, P_R$ is called a $(\beta, b, \varepsilon)$ \emph{state preparation pair} for $y$ if
        \begin{equation}
            P_L\, |0^b\rangle = \sum_{k = 0}^{2^b - 1} c_k\, |k\rangle, \quad
            P_R\, |0^b\rangle = \sum_{k = 0}^{2^b - 1} d_k\, |k\rangle, \quad\text{and}\quad
            \sum_{k = 0}^{m-1} \lvert \beta(c_k^{\ast} d_k) - y_k \rvert \leq \varepsilon,
        \end{equation}
    where for all $k \in \{m, \dots, 2^b -1\}$, $c_k^{\ast} d_k = 0$. 
\end{definition}

\begin{theorem}[Linear combination of unitaries; from \cite{gslw_19}]  \label{thm:be_lcu}
    Let $A = \sum_{k=0}^{m-1} y_k A_k$ be a weighted sum of $n$-qubit operators $A_k$, for some $y \in \mathbb{C}^m$ and $\lVert y\rVert_1 \leq \beta$.
    Then given a $(\beta, b, \varepsilon)$ state preparation pair (Def.~\ref{def:state_prep_pair}) $(P_L, P_R)$ for $y$, and access to $U_k$ for $k \in \{0, \dots, m-1\}$ respectively $(\alpha, a, \varepsilon')$ block encodings of the $A_k$, we can implement a $(\alpha\beta, a + b, \alpha\varepsilon + \beta\varepsilon')$ block encoding of $A$ with a single use of $P_L, P_R$, and a \textsc{select} unitary $W$ of the form
        \begin{equation} \label{eq:lcu_select_def}
            W =
            \sum_{k = 0}^{m - 1} |k\rangle\langle k|_{r_1}\otimes (U_k)_{r_2 r_3}
            +
            \left[\left(I - \sum_{k = 0}^{m - 1}|k\rangle\langle k|\right)_{r_1} \hspace*{-0.2cm}\otimes I_{r_2} \otimes I_{r_3}\right].
        \end{equation}
    Note here that we have notated by subscripts the registers on which our operators act, where $r_1, r_2, r_3$ have size $b = \lceil\log{m}\rceil, a, n$ respectively. The unitary product achieving this block encoding has the simple form $W' = ([P_L^\dagger]_{r_1}\otimes I_{r_2 r_3})W_{r_1 r_2 r_3}([P_R]_{r_1}\otimes I_{r_2 r_3})$.
\end{theorem}

It turns out that we can do even better than LCU for a certain class of matrix transformations described by spectral maps. The most prominent of these algorithms, quantum singular value transformation (QSVT), has proven remarkably useful in unifying, simplifying, and improving the performance of previously disparate quantum information processing tasks \cite{mrtc_unification_21}. We reproduce the main statement of QSVT below, towards presenting a few common applications that will appear as subroutines in our later constructions. This statement is adapted and simplified from various presentations (e.g., Thms.~17 and 18 in \cite{gslw_19}, Ch.~13 of \cite{lw_lecture_notes_26}), and we refer interested readers to pedagogically-oriented and streamlined introductions to this useful class of quantum algorithms \cite{mrtc_unification_21, cs_qsvt_tang_tian_23, lw_lecture_notes_26}.

\begin{definition}[Quantum singular value transformation]  \label{def:be_qsvt}
    Let $U$ an $(1, a, 0)$ block encoding of an $n$-qubit operator $A$ according to orthogonal projectors $\Pi, \tilde{\Pi}$:
    \begin{equation}
        \tilde{\Pi} U \Pi = A.
    \end{equation}
    Moreover, let $P \in \mathbb{R}[x]$ (a polynomial with real coefficients) and $\Phi \in \mathbb{R}^d$, where $P$ has parity $(d \bmod 2)$ and satisfies $|P(x)| \leq 1$ for $x \in [-1, 1]$. Then we can build a circuit that block encodes (parameterized by $\Phi$) the \emph{singular value transformation} of $A$ according to $P$ (taking $d$ odd for the moment):
    \begin{equation}
        P^{\text{(SV)}}(A) \equiv \sum_{k = 1}^{\mathclap{\text{rank}(A)}} P(\xi_k)\,|\ell_k\rangle\langle r_k| = \tilde{\Pi} U_{\Phi} \Pi,
    \end{equation}
    which uses $\mathcal{O}(d)$ copies of $U$, $U^\dagger$, $\Pi, \tilde{\Pi}$-controlled-NOT, and single- and two-qubit gates. Here $|\ell_k\rangle, |r_k\rangle$ are the left and right singular vectors of $A$, and their respective spans are the images of $\tilde{\Pi}$ and $\Pi$. This block encoding, referred to as $U_{\Phi}$, has the following form:
    \begin{equation}
        U_{\Phi} \equiv 
        e^{i\phi_1 (2\tilde{\Pi} - I)} U 
        \prod_{j = 1}^{(d-1)/2}
        \left[e^{i\phi_{2j} (2\Pi - I)} U^\dagger e^{i\phi_{2j + 1} (2\tilde{\Pi} - I)} U \right].
    \end{equation}
    Here the \emph{phases} $\Phi = \{\phi_1, \phi_2, \dots, \phi_d\}$ are precisely the same as those used by the quantum signal processing (QSP) protocol achieving the polynomial $P$. In the even-parity case the final application of $U$ and the $\tilde{\Pi}$-controlled single-qubit rotation is removed.
\end{definition}

The purpose of the exposition of the main QSVT definition (Def.~\ref{def:be_qsvt}) is not to provide deep insight into the mechanism and circuits of these transformations, but instead to communicate that applying spectral maps to block encoded linear operators is well-understood, associated with simple quantum circuits, and comes with tight resource complexity bounds. As above, we refer the interested reader to the numerous recent pedagogically-slanted treatments of the subject, which cleanly lay out competing notational conventions, definitional variants, and circuit shortcuts in special cases.

The workflow of applying QSVT is remarkably straightforward: once a block encoding has been built, the query and space complexity of the desired transformation relates simply to functional analytic properties of polynomial approximations, which are well understood. Moreover, through substantial functional analytic and numerical linear algebraic work, the classical algorithms used to compute the QSVT phases for a given transformation are known to be both efficient and stable, with associated public computational packages \cite{haah_decomposition_19, cdghs_finding_qsp_angles_20, dmwl_efficient_phases_21, wdl_sym_qsp_energy_22, dlnw_robust_iter_24}. In our main construction we will introduce some explicit functions (and their corresponding QSVT subroutines) useful for our purposes, including variants of amplitude amplification and certain thresholding protocols. Wherever possible we will indicate the reason for using a particular block encoding technique over another, highlighting the benefits, limitations, and various trade-offs among the most common block encoding methods.

%%%%%%%%%%%%%%%%%%%%%%%%%%%%%%%%%%%%%%%%%%%%%
\subsection{A minimal example} \label{subsec:minimal_ex}

\noindent We quickly run through an explicit, small example of our eventual construction in the computational basis. The circuits for realizing each subroutine, and other unique aspects of the general case, are the subject of later sections. It is intended that this acts as an extrapolatable template plainly laying out the logic of the full algorithm.

A distinguishing aspect of the QEWT is that it cannot be simply described as a matrix product, linear combination, or spectral map applied among input block encodings. This necessitates a few additional steps and notational amendments. As in prior work, our construction is rooted in the well-known observation that the element-wise product of two matrices is a principal submatrix of their Kronecker product.

\begin{proposition}[Relation between element-wise and Kronecker products] \label{prop:simple_had_prod}
    The element-wise (variously referred to as the \emph{Hadamard} or \emph{Schur} or \emph{pointwise}) product of two matrices $A$, $B$ is a principal submatrix\footnote{I.e., a submatrix obtained by deleting any $k$ rows and the correspondingly indexed $k$ columns.} of their Kronecker product $(A\otimes B)$. For $A$ and $B$ both $2\times 2$ we can see this concretely.\footnote{For later consistency we depict these matrices as block encoded inside larger unitaries, but the matrix products\slash Kronecker products are taking place among the \emph{block encoded} matrices.} Given block encodings of $(A\otimes I)$ and $(I\otimes B)$
    \begin{equation}
        U_{A\otimes I} =
        \left[
        \begin{array}{cccc|c}
            a_0 & & a_1 & & \\
            & a_0 & & a_1 &    \smash{\raisebox{-0.5\normalbaselineskip}{\;\;$\ast$\;}}\\
            a_2 & & a_3 & & \\
            & a_2 & & a_3 & \\\hline
            \multicolumn{4}{c|}{\vphantom{\rule{0pt}{\normalbaselineskip}}\ast} & \vphantom{\rule{0pt}{\normalbaselineskip}}\ast
        \end{array}
        \right],
        \quad\quad
        U_{I\otimes B} =
        \left[
        \begin{array}{cccc|c}
            b_0 & b_1 & & & \\
            b_2 & b_3 & & & \smash{\raisebox{-0.5\normalbaselineskip}{\;\;$\ast$\;}}\\
            & & b_0 & b_1 & \\
            & & b_2 & b_3 & \\\hline
            \multicolumn{4}{c|}{\vphantom{\rule{0pt}{\normalbaselineskip}}\ast} & \vphantom{\rule{0pt}{\normalbaselineskip}}\ast
        \end{array}
        \right],
    \end{equation}
    we can permute rows and columns of their matrix product\footnote{I.e., $(A\otimes I)(I\otimes B) = (A \otimes B)$.} to recover $A\circ B$ in the upper left corner:
        \begin{equation} \label{eq:tensor_had_iden}
        U_{A \otimes B} =
        \left[
        \begin{array}{cccc|c}
            a_{0}b_{0} & \ast & \ast & a_{1}b_{1} & \\
            \ast & \ast & \ast & \ast & \smash{\raisebox{-0.5\normalbaselineskip}{\;\;$\ast$\;}}\\
            \ast & \ast & \ast & \ast & \\
            a_{2}b_{2} & \ast & \ast & a_{3}b_{3} & \\\hline
            \multicolumn{4}{c|}{\vphantom{\rule{0pt}{\normalbaselineskip}}\ast} & \vphantom{\rule{0pt}{\normalbaselineskip}}\ast
        \end{array}
        \right]
        \quad\longmapsto\quad
        U_{P(A \otimes B)P^\dagger}
        =
        \left[
        \begin{array}{cccc|c}
            a_{0}b_{0} & a_{1}b_{1} & \ast & \ast & \\
            a_{2}b_{2} & a_{3}b_{3} & \ast & \ast & \smash{\raisebox{-0.5\normalbaselineskip}{\;\;$\ast$\;}}\\
            \ast & \ast & \ast & \ast & \\
            \ast & \ast & \ast & \ast & \\\hline
            \multicolumn{4}{c|}{\vphantom{\rule{0pt}{\normalbaselineskip}}\ast} & \vphantom{\rule{0pt}{\normalbaselineskip}}\ast
        \end{array}
        \right].
    \end{equation}
    This identity is well-known \cite{hj_matrix_analysis_85, gv_mat_computation_13} and generalizes simply to larger $A$ and $B$.
\end{proposition}

The problem we encounter after this product is that we cannot simply repeat this process, as we have not achieved a block encoding of $(A\circ B)\otimes I$, but rather a \emph{block encoding whose sub-block} is $(A\circ B)$. One way of remedying this (for which there will be multiple methods of varying performance discussed later) is to \emph{mask out} the undesired elements $\ast$ in the top-left block of (\ref{eq:tensor_had_iden}), followed by \emph{copying} the remaining desired sub-block $A\circ B$. The simplest way one could think to do this is again by matrix multiplication on the left and right:
    \begin{equation}
        V \equiv
        \begin{bmatrix}
            \;1 & & & \\
            & 1 & & \\
            & & 0 &\\
            & & & 0\;\;
        \end{bmatrix}
        \implies
        V [P(A\otimes B)P^\dagger] V
        =
        \begin{bmatrix}
            \;a_0 b_0 & a_1 b_1 &  \hphantom{\ast}& \hphantom{\ast} \\
            \;a_2 b_2 & a_3 b_3 &  \hphantom{\ast}& \hphantom{\ast} \\
            & & \hphantom{\ast} & \hphantom{\ast}\;\;\\
            & & \hphantom{\ast} & \hphantom{\ast}\;\;
        \end{bmatrix},
    \end{equation}
where we denote zeros in these matrices where otherwise obvious by blank space, and where we've switched to focusing only on the top-left block of the overall block encoding unitaries. It turns out, as we show in the next section, that \emph{copying blocks of block-diagonal block encodings}\footnote{It is important to distinguish that we will be constantly working with unitaries whose sub-blocks are matrices that themselves have block structure.} is comparatively simple, i.e., computing the following:
    \begin{equation}
        \begin{bmatrix}
            \;a_0 b_0 & a_1 b_1 &  \hphantom{\ast}& \hphantom{\ast} \\
            \;a_2 b_2 & a_3 b_3 &  \hphantom{\ast}& \hphantom{\ast} \\
            & & \ast & \ast\;\;\\
            & & \ast & \ast\;\;
        \end{bmatrix}
        \;\rightarrow\;
        \begin{bmatrix}
            \;a_0 b_0 & a_1 b_1 &  \hphantom{\ast}& \hphantom{\ast} \\
            \;a_2 b_2 & a_3 b_3 &  \hphantom{\ast}& \hphantom{\ast} \\
            & &  \;a_0 b_0 & a_1 b_1\;\;\\
            & & \;a_2 b_2 & a_3 b_3\;\;
        \end{bmatrix}
        =
        I\otimes (A\circ B).
    \end{equation}
If we can perform each of the above steps efficiently, then we have returned to the case that our block encoded matrix is indeed the desired\footnote{Up to a swap on the two registers required to receive the action $A\otimes B$.} $(A\circ B)\otimes I$. The main business of this work is showing that each of these steps, for general $A$ and $B$, is possible, efficient under certain restrictions on $A, B$, and uses less space than existing methods (i.e., logarithmic versus linear in the degree of the desired element-wise product).

Finally, we briefly distinguish between element-wise products between matrices and a derived transform: element-wise functions of a single matrix. The constructions in the following section discuss costs primarily for computing block encodings of $A^{\circ k}$ given a block encoding of $A$, after which linear combination of unitaries (LCU) methods\footnote{We note here that it is non-trivial to properly prepare the \textsc{select} unitary in our setting; see Lem.~\ref{lem:ew_select}.} can be used to compute arbitrary linear combinations of these element-wise monomials. It turns out that our methods extend to a much wider class of generalized matrix products and functions, though their core mechanism follows from techniques in Sec.~\ref{subsec:swap_copy_weaving}.

%%%%%%%%%%%%%%%%%%%%%%%%%%%%%%%%%%%%%%%%%%%%%
\subsection{Swap-copy and the weaving lemma} \label{subsec:swap_copy_weaving}

\noindent We now motivate and construct a few of the core mechanisms of this work, which will allow us to reduce the auxiliary space usage for element-wise functions of block encodings from linear to logarithmic in the degree of the applied function. This ability will stem from two compatible constructions: the first operation, named \emph{swap-copy} (Lem.~\ref{lem:swap_copy}), allows us to cheaply copy blocks of block-diagonal block encodings at the cost of a small amount of additional space, while the second, a technique we call the \emph{weaving lemma} (Lem.~\ref{lem:rec_swap_copy}) allows us to recursively apply the swap-copy operation in a space-efficient way. Together these techniques allow us to build quantum circuits for element-wise products that much more closely resemble circuits for matrix products, which in turn will enable us, in the following Sec.~\ref{subsec:full_const}, to apply existing (in a compatible way) space-compression techniques (i.e., the compression gadget presented again in Appx.~\ref{appx:comp_gadget}) for high-degree matrix products of block encodings.

\begin{definition}[Registers for block encodings with block structure] \label{def:be_registers}
    For most of this work we will consider block-encoded matrices which \emph{themselves} have block structure. For block encodings acting on qubits, we can ascribe this block structure (in a given basis) to sub-registers of qubits. A common definition of a block encoding explicitly indicates a `flag qubit' whose state identifies the location of the block-encoded matrix:
        \begin{equation}
            (\langle 0 |_{r_1}\otimes I_{r_2}) U (|0\rangle_{r_1}\otimes I_{r_2}) = A_{r_2}.
        \end{equation}
    Here the register $r_1$ labels the four blocks of $U$, the top-left of which is $A$, while $r_2$ is the register on which $A$ acts. Where obvious register subscripts may be omitted for brevity.

    We will commonly work with three registers $r_1, r_2, r_3$ (not necessarily single qubits), which define respectively (1) the block-structure of the block-encoded matrix, (2) the block structure of the overall unitary, and (3) the register on which the block-encoded matrix acts; i.e., (1) defines the \emph{inner blocks}, while (2) defines the \emph{outer blocks}. Moreover, registers of the same size as $r_1, r_2,$ or $r_3$ may be instantiated, and joint registers are notated simply (e.g., $r_1 r_2$). The input state for our transform will \emph{almost always be} loaded into the $r_3$ register.
\end{definition}

\begin{lemma}[Swap-copy] \label{lem:swap_copy}
    Let $U$ block encode $(A \oplus B)$, where $\dim{B} = (2^{n}-1) \dim{A}$. Then using one query to $U$, a register of $n$ qubits initialized to $|0^n\rangle$, and $2n$ swap gates, one can deterministically construct a unitary $V$ block encoding $A^{\oplus 2^n}$. Moreover, after applying this block encoding, the additional $n$-qubit register remains in the $|0^n\rangle$ state \emph{under the assumption that the block encoding is successfully applied}.
\end{lemma}

\begin{proof}
    The generic form of the block encoding unitary can be written (suppressing labels giving the sizes of the registers):
        \begin{equation}
            U_{A \oplus B}
            =
            |0\rangle\langle0|_{r_2} \otimes
            \left[
            |0\rangle\langle 0|_{r_1} \otimes A_{r_3}
            +
            \left(\sum_{j,k = 1}^{2^n - 1} |j\rangle\langle k|_{r_1} \otimes [B_{j,k}]_{r_3}\right)
            \right]
            +
            \cdots.
        \end{equation}
    Introducing a new register $r_0$ of $n$ qubits (the same size as $r_1$), the action of $U_{A\oplus B}$ under conjugation by the register-wise swap gate $S_{r_0, r_1}$ (under the assumption that the register $r_0$ is initialized to $|0\rangle_{r_0}$) is
        \begin{align}
            &S_{r_0, r_1}\,
            U_{A \oplus B}\,
            S_{r_0, r_1}\nonumber\\
            &=
            \left[\sum_{i = 0}^{2^n - 1} |0\rangle\langle i|_{r_0} \otimes |i\rangle\langle 0|_{r_1} + \dots \right]
            U_{A \oplus B}\left[\sum_{i = 0}^{2^n - 1} |i\rangle\langle0|_{r_0} \otimes |0\rangle\langle i|_{r_1} + \cdots \right]
            \\
            &=
            |0\rangle\langle0|_{r_2}
            \otimes
            \left[
            |0\rangle\langle 0|_{r_0}
            \otimes
            I_{r_1} \otimes A_{r_3}
            +
            \sum_{\substack{p, q = 1\\j, k = 0}}^{2^n-1} |p\rangle\langle q|_{r_0}\otimes |j\rangle\langle k|_{r_1} \otimes [\;\ast\;]_{r_3}\;
            \right]
            +
            \cdots.
        \end{align}
    We have dropped those terms that (1) appear outside the top left block of the block encoding unitary (labeled by $r_2$), or (2) appear if the register $r_0$ were initialized to some other state than $|0\rangle$, or (3) vanish given the block diagonal nature of the block encoded matrix. This is a block encoding of $A^{\oplus 2^n}$ on $r_1 r_3$ when the register $r_0$ is initialized to the $|0\rangle$ state, and the state in $r_0$ is used catalytically given the second sum over $p, q$. Recall that $r_2$ indexes the \emph{outer block} structure of our \emph{block encoding unitary}, while $r_1$ indexes the \emph{inner block} structure of our \emph{block encoded} matrix. A depiction of this circuit is given in Fig.~\ref{fig:swap_copy}.
\end{proof}

As an example of the above lemma, in the $n = 1$ case\footnote{Note the general case results in $2^n$ copies of $A$ along the diagonal, denoted $A^{\oplus 2^n}$.}, when $A$ and $B$ are both $2\times2$ matrices, this conversion can be written
\begin{equation} \label{eq:block_be}
    U =
    \left[
    \begin{array}{cc|c}
         A & & \dots \\
         & B & \\\hline
         \vdots & & \ddots
    \end{array}
    \right]
    \;\longmapsto\;
    V =
    \left[
    \begin{array}{cc|c}
         A & & \dots \\
         & A & \\\hline
         \vdots & & \ddots
    \end{array}
    \right],
\end{equation}
and this process catalytically uses a single qubit initialized to $|0\rangle$, one query to $U$, and two swap gates.

\begin{figure}
    \centering
    \includegraphics[width=\linewidth]{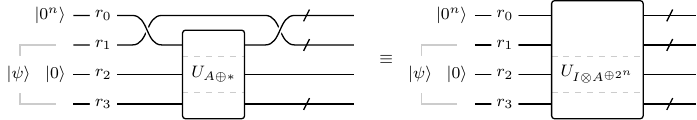}
    \caption{Circuit for swap-copy (Lem.~\ref{lem:swap_copy}). The catalytically used state $|0\rangle$ in $r_0$ (the \emph{weaving state}; see Lem.~\ref{lem:rec_swap_copy}) converts a block encoding of $(A \oplus\ast)_{r_1 r_3}$ to a block encoding of $I_{r_1}\otimes A_{r_3} = A^{\oplus 2^n}_{r_1 r_3}$, consuming only a single call to the block encoding, requiring two-qubit swap gates, and catalytically using a small amount of additional space \emph{if the block encoding succeeds}. On the left is the applied circuit; on the right is an equivalent form, i.e., a block encoding of $I_{r_0}\otimes I_{r_1}\otimes A_{r_3}$, \emph{under the assumption} $r_0$ begins in the all zeros state.}
    \label{fig:swap_copy}
\end{figure}

\begin{definition}[Submatrix permutation unitary] \label{def:submat_permut}
    We define the \emph{submatrix permutation} to be the unitary on $2n$ qubits applying CNOTs between respective pairs of qubits in each $n$-qubit half. This has form
        \begin{equation}
            P = \sum_{j, k = 0}^{2^n - 1} |j\rangle\langle j|\otimes |j \oplus k\rangle\langle k|,
        \end{equation}
    where if $j,k$ have binary representations $j_0 \dots j_{n-1}$ and $k_0 \dots k_{n-1}$ respectively, then the notation $j \oplus k$ defines the integer whose binary representation is $(j_0 \oplus k_0) \dots (j_{n-1} \oplus k_{n-1})$. As was observed in \cite{gychanar_had_prod_25}, this is precisely the permutation required to relocate the submatrix of $(A\otimes B)$ equal to the element-wise product $(A\circ B)$ to the top left corner, i.e.,
    \begin{equation}
        (I_n\otimes \langle 0^n|)P (A\otimes B) P^\dagger (I_n\otimes |0^n\rangle) = A \circ B.
    \end{equation}
    Here the CNOTs span the two $n$-qubit registers implicit in the tensor product $(A\otimes B)$.
\end{definition}

\begin{definition}[Mask unitary] \label{def:mask}
    We define the \emph{mask} as the unitary on registers $r_1$, $r_2$, and $r_3$ that block encodes the following operation on $r_1 r_3$:
    \begin{equation}
        M = \frac{1}{2}
        \left(I_{r_1} \otimes I_{r_3} + \Big[ 2|0\rangle\langle 0|_{r_1} - I_{r_1}\Big] \otimes I_{r_3}\right) = |0\rangle\langle0|_{r_1} \otimes I_{r_3}.
    \end{equation}
    Such a block encoding can be simply constructed through an LCU \cite{cw_lcu_12, bccks_ham_lcu_14, bck_ham_lcu_15} of state-dependent reflections. Here $|0\rangle$ indicates the all-zeros state on the labeled register.
\end{definition}

It is clear from the above discussion of the swap-copy operation (Lem.~\ref{lem:swap_copy}) that block diagonality of block encodings is a useful assumption. In the following definition and lemma we present a lifted and generalized version of this construction, concerning the resource requirements for the composition of operator-valued functions with block structure in certain bases, given certain arguments, over certain spaces of states. While the statement we give is a little abstract, in practice these operator-valued functions will be transformations of block encodings, and the ultimate purpose of the results of this section will be to achieve space efficient compositions of block encoding transformations which themselves have additional, special block structure.

\begin{definition}[Blocked operator] \label{def:blocked_op}
    Let $W(\ast_1, \ast_2)$ be a function that takes as input two operators\footnote{Bounded, finite dimensional, linear operators over Hilbert spaces.} and returns an operator. Moreover, let $S$ be a set of pairs of operators and $R, T$ sets of quantum states such that for all $(V_1, V_2) \in S$, all $|r\rangle \in R$ the following holds:
        \begin{align}
            % W(V_1, V_2)\,|r\rangle_{r_0} |t\rangle_{r_1} &= |r\rangle_{r_0} |\ast\rangle_{r_1},\quad\text{equivalently,}\label{eq:block_op_1}\\
            W(V_1, V_2) &= |r\rangle\langle r|_{r_{0}} \otimes U(V_1, V_2)_{r_1} + \sum_{r', r'' \neq r} |r'\rangle\langle r''|_{r_0} \otimes [\;\ast\;]. \label{eq:block_op_2}
        \end{align}
    Then we call $W$ an $(S, R)$-blocked operator function, or just a \emph{blocked operator} when context is clear. We say $W$ is \emph{sequential} if it can be written as the composition\footnote{Here we mean \emph{composition as linear operators}, and not the functional composition of linear operator-valued functions. Confusingly, in literature operator composition is sometimes notated $A \circ B$.} of possibly many linear operators, some of which may be copies of $V_1$ and $V_2$.
        
    In (\ref{eq:block_op_2}) $U(V_1, V_2)$ is the operator on $r_1$ induced by initializing the $r_0$ register to $|r\rangle$. Thus, $W$ returns operators that catalytically consume $|r\rangle$ (in the $r_0$ register), when fed operator pairs in $S$. Equivalently $W$, when given elements in $S$, is block-diagonal according to orthonormal bases containing elements $|r\rangle \in R$. Here $r', r''$ each run over labels to basis elements orthogonal to $|r\rangle$.
\end{definition}

We can consider the following question about \emph{blocked operators}: given a composition of many $U_k(\ast_1, \ast_2)$ acting on $r_1$, each of which induced by some $W_k$ sharing $(S, R)$, how many copies of $|r\rangle$ total, and across an auxiliary register of what size, are needed to implement this composition \emph{using only the $W_k$}? In other words we want to simulate the action of the composition of $U_k$ through the use of $W_k$ on some possibly larger Hilbert space. For simplicity, all $W_k$ share $(S, R)$ and we assume $V_1 = U_{k'}$, $V_2 = U_{k''}$ from any $W_{k'}$, $W_{k''}$ is a sufficient condition for $(V_1, V_2) \in S$.

We claim that the number of copies of $|r\rangle$ is equal to the \emph{depth of the tree representing the desired composition} of the $U_k$, and provide a brief statement and proof of this fact in the dedicated Lem.~\ref{lem:rec_swap_copy}. We will refer to this result as the \emph{weaving lemma}, as copies of $|r\rangle$ are re-used by `weaving' them through applications of the $W_k$ to induce $U_k$.

\begin{lemma}[Recursive composition of blocked operator functions; \emph{weaving lemma}] \label{lem:rec_swap_copy}
    Let $W_k$ for $k \in K$ all $(S, R)$-blocked operator functions (Def.~\ref{def:blocked_op}) acting on collections of qubits, and $F$ a tree associated to a composition of the $U_k$. Then if the $W_k$ are all \emph{sequential}, the space required to implement $F$ using the $W_k$ grows at most linearly in the depth of $F$.

    \begin{proof}
        Proof is by casework. If $U_k(\ast_1, \ast_2)$ is a leaf (i.e., applied only to pairs $(V_{1,k}, V_{2,k})$), then it can be implemented using $W_k$ consuming a single copy of $|r\rangle$ by definition. Alternatively, $U_k(\ast_1, \ast_2)$ might be a node:
            \begin{alignat}{2}
                U_{k}
                &=
                U_{k}(U_{k'}(\ast_1, \ast_2), U_{k''}(\ast_3, \ast_4)), \quad\quad&&\text{or}\\
                &= U_{k}(V_1, U_{k'}(\ast_1, \ast_2)), &&\text{or}\\
                &= U_{k}(U_{k'}(\ast_1, \ast_2), V_2).&&
            \end{alignat}
        In any case, as each of the $U_{k'}, U_{k''}$ is used sequentially by assumption, $U_k$ can be implemented using two copies total of $|r\rangle$, each in their own register, together with calls to $W_{k}$, $W_{k'}$, and possibly $W_{k''}$. This will be added to any copies required by the \emph{arguments} of $W_{k'}$ or $W_{k''}$ further down in the tree. I.e., an additional layer of recursion necessitates only a single fresh copy of the state $|r\rangle$ on its own register to convert each $W_k$ to its respective $U_k$ on $r_1$. The sequential assumption permits the catalytic use (allocation and freeing) of $|r\rangle$ across arguments of any $W_k$. This immediately yields space complexity that is linear in the depth of the tree. A depiction of this argument is given in Fig.~\ref{fig:weaving_lemma}.
    \end{proof}
\end{lemma}

\begin{figure}
    \centering
    \includegraphics[width=0.8\linewidth]{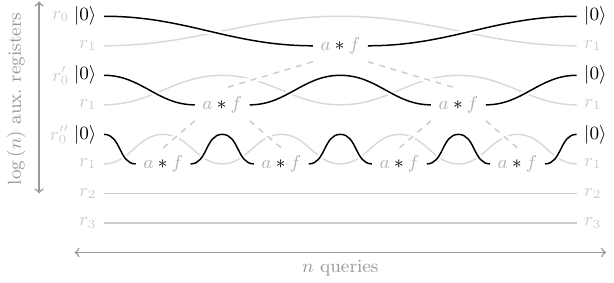}
    \caption{Simplified depiction of a specific instance of the weaving lemma (Lem.~\ref{lem:rec_swap_copy}). Here $(\ast)$ is a quantum process that catalytically uses copies of $|0\rangle$ (possibly on many qubits) in a register $r_0$ by swapping the contents of $r_0$ with this process' $r_1$ register (i.e., a blocked operator, Def.~\ref{def:blocked_op}). The use by $(\ast)$ of the contents of $r_0$ is indicated by an allocation $(a)$ followed by a freeing $(f)$. If this process itself can sequentially call copies of $(\ast)$, then the total required copies of register $r_0$, here denoted by $r_0', r_0'', \dots,$ grows with the depth of the tree induced by recursive function calls (seen by the re-use of $r_0$ across protocols at the same depth).}
    \label{fig:weaving_lemma}
\end{figure}

\begin{proposition}[Block encoding the element-wise product] \label{prop:ewt_simple_form}
    Given oracle access to unitaries $U_A$ and $U_B$, respectively $(1, a, 0)$ and $(1, b, 0)$ block encodings of $A$ and $B$ (square matrices of equal dimension), we can block encode their element-wise product
        \begin{align}
            (\langle 0|_{r_2} \otimes I_{r_1 r_3})(P')_{r_1 r_3} (I_{r_1}\otimes U_{A, r_3}) &(I_{r_3} \otimes U_{B, r_1})(P')_{r_1 r_3}^\dagger
            (|0\rangle_{r_2} \otimes I_{r_1 r_3})\\
            &=
            (\langle 0|_{r_2} \otimes I_{r_1 r_3})U_{P(A\otimes B)P^\dagger, r_1 r_3}(|0\rangle_{r_2} \otimes I_{r_1 r_3})&\\
            &=
            |0\rangle\langle 0|_{r_1} \otimes (A\circ B)_{r_3} +\, \cdots
        \end{align}
    using a single query to each of $U_A$ and $U_B$. Here $(P')$ (Def.~\ref{def:submat_permut}) is the unitary comprising $n$-CNOT gates permuting rows and columns of $(A\otimes B)$ to move $(A \circ B)$ to the top-left block. This unitary, which acts only on $r_1$ and $r_3$, can thus be taken into the block encoding unitary's subscript, denoted here $P(A\otimes B)P^\dagger$. For brevity from here on we will simply use $P$ to denote \emph{both} this qubit-wise CNOT gate and the induced permutation matrix.

    \begin{proof}
        This result is known and its proof reproduced in multiple prior works \cite{zzrf_basic_q_lin_alg_21, gychanar_had_prod_25}, essentially following from the known identity that the element-wise product is a submatrix of the tensor product \cite{hj_matrix_analysis_85}.
    \end{proof}
\end{proposition}

Important to note in Prop.~\ref{prop:ewt_simple_form} is that $r_1$ now labels \emph{blocks} of the \emph{block encoded matrix} $P(A\otimes B)P^\dagger$, and that the other blocks contain related (but undesirable) products of elements of $A$ and $B$. Consequently there is no easy way to feed this block encoding back into another element-wise product protocol without the instantiation of additional space---in the following Prop.~\ref{prop:clean_ewt_prod_form} we discuss one method of \emph{cleaning} this block encoding (i.e., removing these undesirable blocks outside of $(A\circ B)$ and performing a swap-copy operation; see Lem.~\ref{lem:swap_copy}) to make alternative methods for auxiliary space reduction applicable. This results in a block encoding which not only has $(A \circ B)$ as its top-left block, but which has the form $I_{r_1}\otimes(A\circ B)_{r_3}$, which is the same form as the input block encodings for our tensor-product to element-wise product gadget in Prop.~\ref{prop:ewt_simple_form}.

\begin{proposition}[Block encoding the clean element-wise product] \label{prop:clean_ewt_prod_form}
    Let $U_A$ and $U_B$ be block encoding unitaries for square matrices $A$, $B$ of equal dimension. Then the clean block encoding of their element-wise product $(A \circ B)$ can be written as a product\footnote{That is, this product, denoted where relevant by $\times$, is not between the block encoding unitaries but \emph{between the matrices they block encode}. Achieving this requires a block-encoding multiplication subroutine (e.g., Def.~\ref{def:be_mult} or Lem.~\ref{lem:comp_gadget}) and the use of additional space, discussed later and here denoted by the register label `mult.'} of three \emph{block-encoded matrices} ($U_M$, $U_{P (A\otimes B) P^\dagger}$, and $U_M^\dagger$) interspersed with unitaries. Here we denote this as an operator-valued function $W$ on $r_0, r_1, r_3$ of the block encodings of $A$ and $B$:
    \begin{align}
        W_{r_0 r_1 r_3}(U_A, U_B) \equiv& \label{eq:operator_func_def}\\
        (\langle 0|_{r_2, \text{mult}} \otimes I_{r_1 r_3}) S_{r_0 r_1} \big[ U_M \times U_{P (A\otimes B) P^\dagger} \times U_M^\dagger \big] S_{r_0 r_1} (| 0\rangle_{r_2, \text{mult}} \otimes I_{r_1 r_3})=&\label{eq:unitary_func_1}\\[0.2em]
        (\langle 0|_{r_2, \text{mult}} \otimes I_{r_1 r_3}) S_{r_0 r_1} \big[U_{M P (A\otimes B) P^\dagger M^\dagger}\big] S_{r_0 r_1} (| 0\rangle_{r_2, \text{mult}} \otimes I_{r_1 r_3})=&\label{eq:unitary_func_2}\\[0.2em]
        |0\rangle\langle 0|_{r_0} \otimes (I_{r_1}\otimes (A\circ B)_{r_3}) + \cdots.& \label{eq:unitary_func_3}
    \end{align}
    This product uses one copy each of the block encoding unitaries for $A$ and $B$, a constant amount of additional space, and ultimately frees the register $r_0$ initialized to $|0\rangle_{r_0}$, \emph{if the block encoding is successfully applied}. Explicitly, we give the circuit form of Fig.~\ref{fig:full_prod} in (\ref{eq:unitary_func_1}), apply the compression gadget multiplications in (\ref{eq:unitary_func_2}), and apply our swap-copy identity in (\ref{eq:unitary_func_3}). Note finally the register labeled \emph{mult} (later referred to as $r_4$ and explicated along with compression techniques in Sec.~\ref{subsec:full_const}), which is used to multiply the block encoded input matrices and mask matrices, and which returns to the state $|0\rangle$ precisely when this multiplication is successful.
\end{proposition}

\begin{proof}
    The proof of this statement follows directly from the known action of the mask unitary multiplication, which selects out only the $|0\rangle\langle 0|_{r_1}$ sub-block of the block encoded $P (A\otimes B) P^\dagger$, followed by the swap-copy operation that ensures that all computational basis states in $r_3$ are acted on \emph{as if} $r_1$ had been in the state $|0\rangle$ (as this state is indeed swapped in from $r_0$), resulting in the overall action of $I_{r_1}\otimes (A\circ B)_{r_3}$, \emph{under the assumption} that $r_2$ begins and ends in the state $|0\rangle$.
\end{proof}

In the above proposition, note that there can be two different swap operations (denoted $S$, and depicted in Fig.~\ref{fig:full_prod}), one between the two registers originally holding $A$ and $B$, i.e., $(A_{r_1}\otimes B_{r_3})$, as well as one between $r_0$ and $r_1$ (used for the weaving lemma; Lem.~\ref{lem:rec_swap_copy}).

\begin{remark}[On error in block encoding inputs] \label{rem:be_error}
    The main theorems in this section are presented in the case that the input block encodings have no error, i.e., are of the form $(\alpha_k, a_k, 0)$ for some set of permissible $k$. This has been done both for brevity, and because the calculations involved in the generic case of $(\alpha_k, a_k, \varepsilon_k)$ are well known in the context of common arithmetic operations among block encodings. Given the close relationships we rely on between the element-wise and more common matrix products, and the linearity of this error with respect to the outer loop of LCU (see the main Thm.~\ref{thm:qewt}),\footnote{For the curious reader we note that when applying QSVT this error propagation becomes slightly more complex, as the relation between the operator norm of the difference of two linear operators, and the difference between a chosen function applied to those operators, depends on the function's modulus of continuity (see Sec. 13.7 of \cite{lw_lecture_notes_26}).} error propagation in our results will go through almost identically (and will always have a simple upper bound) in terms of these previous results. We include detailed discussion of these error-propagation techniques\slash known results in the dedicated Appx.~\ref{appx:error_prop}.
\end{remark}

The statement of Prop.~\ref{prop:clean_ewt_prod_form} immediately suggests a definition for the unitary-valued function (\ref{eq:unitary_func_3}) used in the weaving lemma (Lem.~\ref{lem:rec_swap_copy}). In the following proposition we give this function explicitly, before discussing how it interacts with other space-compression techniques in block encoding manipulation towards the formal statement of our main result in the following Sec.~\ref{subsec:full_const}.

\begin{proposition}[Operator-valued function for element-wise products] \label{prop:swap_copy_unitary_func}
    As hinted in the statement of Prop.~\ref{prop:clean_ewt_prod_form}, the operator-valued function to which we'll apply Lem.~\ref{lem:rec_swap_copy} is the same as that given by (\ref{eq:operator_func_def}), which has the form required by the weaving lemma when considered as a function of its block encoding inputs, i.e.:
    \begin{equation} \label{eq:weaving_lem_func}
        W(\ast_1, \ast_2)_{r_0 r_1 r_3} \equiv |0\rangle\langle 0|_{r_0} \otimes U(\ast_1, \ast_2)_{r_1 r_3} + \sum_{r', r'' \neq r} |r'\rangle\langle r''|_{r_0} \otimes [\;\ast\;]_{r_1 r_3}.
    \end{equation}
    In this case the sets $S, R$ comprise respectively, (1) block encodings of square operators of the same dimension, (2) the quantum state $|0\rangle$ on a register of $n$ qubits, and (3) any quantum state on a register of $2n$ qubits.

    Note for now that the registers $r_2, r_4$, which handle our outer block selection and block encoded matrix multiplication respectively, are assumed to begin in and return to the state $|0\rangle$ to ensure the success of our protocol (i.e., respectively ensuring we correctly end up in the top-left outer block of our block encoding, and that the masking operation succeeded). In other words, the full $W_{r_0 r_1 r_2 r_3 r_4}$ unitary-valued function itself has block structure, the $|0\rangle\langle0|_{r_2 r_4}$ block of which is the expression in \eqref{eq:weaving_lem_func}. In the following sections we further explicate how these registers function, their required sizes, the success probability of the resulting protocol, and how this action interacts with the swap-copy operation, the weaving lemma, and recursively defined higher-degree protocols.
\end{proposition}

To summarize our progress before the full construction, we have created a gadget in Lem.~\ref{lem:rec_swap_copy} that recursively applies the swap-copy operation of Lem.~\ref{lem:swap_copy}; in turn this allows quantum circuits whose unitaries have nested block-structure to be composed in a space-efficient way \emph{up to assumptions that we can post-select on a subset of auxiliary qubits and end up in critical blocks}.\footnote{I.e., the same assumptions that are usually made in block-encoding algorithms, though careful analysis of the success probability of these post-selections is required, and the subject of both Secs.~\ref{subsec:full_const} and \ref{sec:alt_const}.} In the next section we build on top of these lemmas to provide explicit circuits for both element-wise powers and linear combinations of these element-wise powers (the latter of which is not trivial, and exposes oversights in prior work that we rectify). An important aspect of these extensions is that we maintain logarithmic space auxiliary usage in the degree of the applied polynomial element-wise function, which requires us to keep track of both (1) the usual compression gadget accumulation register as well as (2) the recursively-applied swap-copy register. While we leave this discussion to the later section on alternative constructions (Sec.~\ref{sec:alt_const}) and the final Sec.~\ref{sec:discussion}, it is an interesting research direction to explore improved block encoding manipulation subroutines under assumptions of block structure or other side information.

%%%%%%%%%%%%%%%%%%%%%%%%%%%%%%%%%%%%%%%%%%%%%
\subsection{Full construction} \label{subsec:full_const}

\noindent We build on the construction of the previous section (which used the weaving lemma, Lem.~\ref{lem:rec_swap_copy}, to convert a circuit computing element-wise products into one computing matrix products) and apply a series of known techniques for space reduction (so-called \emph{compression gadgets}, see Appx.~\ref{appx:comp_gadget}) in block encoding algorithms for computing matrix products. Most of the remaining work thus lies in showing these existing methods are compatible with our previous constructions, as well as examining their behavior under small errors in the input block encoding, ultimately showing we can achieve exponential space reduction in the degree of the applied element-wise function.

We first discuss properties of these \emph{compression gadgets}, followed by an explicit circuit for the recursively defined element-wise product (Thm.~\ref{thm:iterated_prod}), which can be bootstrapped through LCU techniques (with an important detour to correct prior work; Appx.~\ref{appx:lcu_nonstandard}) to yield our main result (Thm.~\ref{thm:qewt}) characterizing the space and query complexity of applying element-wise functions to block encodings. Along the way we introduce a variety of block encoding manipulation techniques, indicating where they require modification.

The compression gadget, introduced in \cite{lw_comp_gadget_19} in the context of Hamiltonian simulation, applied widely in dynamical systems simulation \cite{flt_time_marching_23}, and recently extended in \cite{vg_reducing_space_25}, allows for dramatic space-savings when block encoding high-degree \emph{matrix products} of block encodings. The core idea is to encode the success or failure of intermediate block encoding applications in binary on an auxiliary register, rather than unary; this exchanges unimportant information on which intermediate block encoding application failed in the event of an out-of-block measurement for exponential space savings in the degree of the product. We restate this result more casually below, and refer interested readers to its formal statement and associated circuit in Appx.~\ref{appx:comp_gadget} (which also points toward tangential but interesting developments of related space-saving gadgets).

\begin{lemma}[Compression gadget; \emph{informal}] \label{lem:comp_gadget_informal}
    Let $U_k$ for $k \in [K]$ a series of unitaries respectively block encoding (square) $A_k$. Then we can construct a unitary using one copy each of the $U_k$ and $\log{(K)}$ auxiliary qubits that block encodes the product
        \begin{equation}
            B = A_1 A_2 \cdots A_K,
        \end{equation}
    whereas naïve repetition of the usual product gadget would require $\mathcal{O}(K)$ auxiliary qubits. The success probability for applying $B$ to the input quantum state depends on the norm-squared of the resulting state.
\end{lemma}

Our aim is to apply compression gadget techniques to higher-degree variants of the construction of Prop.~\ref{prop:clean_ewt_prod_form}. In that section we provided a `clean' (i.e., ready for repeated self-composition) circuit computing element-wise products in terms of matrix products. Successfully combining this construction with the compression gadget requires addressing the following questions:
\begin{quote}
    \emph{\textbf{Q1:} We have multiple sources of auxiliary space including (a) the registers used by the swap-copy operations\slash weaving lemma, (b) the accumulation register used by the compression gadget, and (c) the auxiliary space used by each input block encoding. Does the resulting unitary have sensible block structure with respect to these registers?}\footnote{We need to ensure, under the assumption that the auxiliary qubits all begin and end in the $|0\rangle$ state, that the desired element-wise function is applied to the input block encoding.}

    \emph{\textbf{Q2:} The log-space use of our protocol relies critically on the re-use of the weaving state from Lem.~\ref{lem:rec_swap_copy} across each recursive level (i.e., that this state is \upshape{catalytic}); is this compatible with the compression gadget, and what is the effect on the protocol's overall success probability?}
\end{quote}
We address the first question by careful bookkeeping of the registers involved, which are depicted across multiple stages of our construction in Figs.~\ref{fig:full_prod} (our `cleaned' matrix product to element-wise product construction), \ref{fig:reduced_form_prod} (a reduced form of this construction as an operator-valued function), and \ref{fig:recur_prod} (the second level of our recursive construction for the iterated self element-wise product). The second question, which we address after the statement of our main result, follows from basic calculations on the success of each constituent block encoding.

\begin{theorem}[Iterated element-wise products; \emph{formal}] \label{thm:iterated_prod}
    Let $U_A$ an $(\alpha, a, \varepsilon)$ block encoding of an $n$-qubit operator $A$. Then an $(\alpha^d, \mathcal{O}(a + n\log{d}), (\alpha + \varepsilon)^d - \alpha^d)$ block encoding of the $d$-th element-wise power of $A$, denoted $A^{\circ d}$, can be constructed using $d$ copies of $U_A$.
    \begin{proof}
    At a high level, proof follows directly from an application of the \emph{compression gadget} of Low and Wiebe \cite{lw_comp_gadget_19} (see Lem.~\ref{lem:comp_gadget}) to the recursively-defined product of (\ref{eq:weaving_lem_func}) in Lem.~\ref{lem:rec_swap_copy}, along with the following two observations. 
    
    The first is that unrolling our $\log{d}$-level recursion gives a matrix product of block encodings of length $\mathcal{O}(d)$ interspersed with known unitaries (which have no bearing on the operation of the compression gadget). The second is that, under the condition that the compression gadget succeeds (i.e., the auxiliary register is measured in the $|0\rangle_{\text{comp}}$ state), the weaving lemma (Lem.~\ref{lem:rec_swap_copy}) ensures that all $\mathcal{O}(\log{d})$ initial $|0\rangle_{r_0}, |0\rangle_{r_0'}, \cdots$ states, one for each recursion level, are used catalytically, as $r_1 r_2 r_3$, assuming the mask unitaries' successful application (kept track of by the compression gadget), encounter only \emph{block-diagonal block encodings}, for which the $|0\rangle_{r_0^{(m\prime)}}$ is a $(+1)$-eigenstate. 
    
    The overall space use is thus simply that which is required for the input block encoding, $a$, plus the $\log{d}$ copies of $|0\rangle_{r_0}$, plus the $\mathcal{O}(\log{d})$ qubits required for the compression gadget. More precisely we end up with a $(1, \mathcal{O}(a + n\log{d}), 0)$ block encoding of $A^{\circ d}$ in the case that $\alpha = 1$ and $\varepsilon = 0$, with the general case following from standard error propagation techniques in block encoding (Appx.~\ref{appx:error_prop}).

    More formally, as our construction is recursive, careful proof relies on showing the validity of the recursive step, as the base case already achieves a block encoding of the element-wise product. The input for our protocol is a block encoding (according to the registers $r_2$ and $r_4$) that acts as the identity on $r_1$ and catalytically uses $r_0$ (and possibly more copies of it, e.g., $r_0', r_0'', \cdots$). Algebraically this input thus has the following form:
    \begin{equation} \label{eq:input_be}
        U_{A} \equiv 
        |0\rangle\langle0|_{r_2, r_4}\otimes 
        \left(|0\rangle\langle0|_{r_0, \cdots} \otimes A_{r_3}\otimes I_{r_1} + \cdots \right) + \cdots,
    \end{equation}
    where the dropped terms are supported by orthogonal projectors on $r_0$, $r_2$ or $r_4$. I.e., that on the success of $r_2$ and $r_4$, we are guaranteed that the state $|0\rangle$ in $r_0$ (and any copies of it) is only catalytically used. We now show that the output of our protocol also has this form.

    Our argument follows by taking two inputs of the form in (\ref{eq:input_be}) and applying our base-case circuit, which includes the masking unitaries, swap-copy, and permutation between $r_1$ and $r_3$ on one of the arguments. Our aim is to show that the resulting unitary satisfies the input conditions we've just established, save now block encoding $(A \circ B)$ and possibly using more copies of $r_0$. Crucially, note that we always perform controlled additions to the same $r_4$, and include a new\footnote{Here we use the notation $r_0'$ to denote any new addition to the possibly many copies of $r_0$ instantiated at previous steps in the recursion.} $r_0'$ register for the swap copy. The constituent pieces are repeated below:\footnote{We repeat these here for clarity, but their formal definitions were given in Def.~\ref{def:submat_permut} (perm), Def.~\ref{def:mask} (mask), and Lem.~\ref{lem:comp_gadget} (add), and align with standard textbooks \cite{nc_textbook_11}.}
    \begin{align}
        \text{swap}_0 &\equiv \sum_{j,k} |jk\rangle\langle kj|_{r_1 r_3},\\
        \text{swap}_1 &\equiv \sum_{j,k} |jk\rangle\langle kj|_{r_0' r_1},\\
        \text{add} &\equiv \sum_{j, k} |j\rangle\langle j |_{r_2}\otimes |(k + j) \bmod{2^{\lceil \log{d} \rceil}}\rangle\langle k|_{r_4},\\
        \text{mask} &\equiv \vphantom{\sum_{j,k}} |0\rangle\langle 0|_{r_2} \otimes |0\rangle\langle 0|_{r_1}\otimes I_{r_3} + \cdots\\
        \text{perm} &\equiv \sum_{j, k} |j\rangle\langle j|_{r_1}\otimes |j \oplus k\rangle\langle k|_{r_3},
    \end{align}
    where we should note the difference between \emph{add} and \emph{perm} is that the latter is doing a bit-wise XOR conditioned on $r_1$, while the former is truly performing modular addition to the compression register $r_4$ (which is chosen to have adequate size to prevent rollover). Note for clarity that $r_0$, $r_1$, and $r_3$ are on $n = \log{N}$ qubits, while $r_2$ is a single qubit, and $r_4$ has size at most $\lceil \log{d}\rceil$ qubits. Properly composing these operations, the circuit we want to investigate has the sequential form (breaking across lines into segments for semantic clarity):
    \begin{align}
        \vdash\; \hphantom{\rightarrow} [\text{swap}_0] \rightarrow&\\
        [\text{mask}] \rightarrow
        [\text{add}] \rightarrow&\\
        [\text{perm}] \rightarrow&\\
        [U_A] \rightarrow
        [\text{add}] \rightarrow&\\
        [\text{swap}_1] \rightarrow
        [U_B] \rightarrow
        [\text{swap}_1] \rightarrow
        [\text{add}] \rightarrow&\\
        [\text{perm}] \rightarrow&\\
        [\text{mask}] \rightarrow
        [\text{add}] \rightarrow&\\
        [\text{swap}_0] \hphantom{\rightarrow}& \;\dashv.
    \end{align}
    A key property is that every matrix multiplication \emph{between block encoded matrices} must be followed by an addition to the compression gadget register, and that $r_1$ and $r_3$ are swapped only for the application of $U_B$. We can then trace the action of each of these segments toward showing the desired recursive assumption holds. The easiest way to do this is to note that a final measurement of $|0\rangle_{r_4}$ \emph{necessitates} (as all of the interspersed additions from $r_2$ to $r_4$ must have not occurred, assuming no rollover), the realization of the intermediate block encoding:
    \begin{equation}
        |0\rangle\langle 0|_{r_2 r_4} \otimes 
        |0\rangle\langle0|_{r_1}\otimes (A\circ B)_{r_3}
        + \cdots,
    \end{equation}
    where \emph{no other blocks} according to $r_1$ appear with a non-zero coefficient \emph{when in tensor product} with $|0\rangle\langle 0|_{r_2}$. In this case the swap-copy operates as intended, transforming, again under the assumption that $r_2$ begins and ends in the all zeros state, into the operation $(A\circ B)_{r_3}\otimes I_{r_1}$. This means the full unitary must be of the desired form from our inductive hypothesis:
    \begin{equation}
        U_{A\circ B} \equiv 
        |0\rangle\langle0|_{r_2, r_4}\otimes 
        \left(|0\rangle\langle0|_{r_0}\otimes |0\rangle\langle 0|_{r_0'} \otimes (A \circ B)_{r_3}\otimes I_{r_1} + \cdots \right) + \cdots,
    \end{equation}
    where the only difference now is the inclusion of the new $r_0'$ register for the top-level swap-copy (where we did not have to perform intermediate additions from $r_0$ to $r_4$ precisely because the additions from $r_2$ and the catalytic structure of the input block encodings mean adding from $r_2$ to $r_4$ completely captures the catalysis of $r_0$). Finally, recall again that all of this depends on $r_4$ being large enough such that, even if all constituent block encodings fail, the register does not wrap back around to zero under modular addition; choosing $r_4$ of size $\mathcal{O}(\log{d})$ is sufficient.
    \end{proof}
\end{theorem}

\begin{figure}
    \centering
    \includegraphics[width=\linewidth]{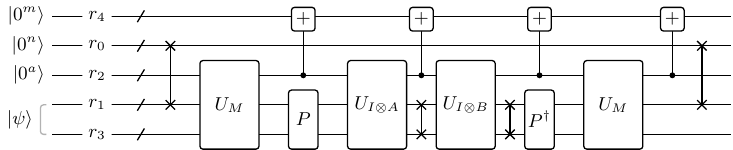}
    \caption{Depiction of the circuit for \emph{clean element-wise products} of block encodings, i.e., block encodings of $I_{r_1}\otimes(A\circ B)_{r_3}$ using auxiliary space $r_4 r_2$. Note if $r_4 r_2$ is in the all-zeros state, then the state $|0^n\rangle$ in $r_0$ is necessarily used catalytically. Here the $(+)$ notation indicates projection-controlled addition with target register $r_4$ used for the compression gadget (see Lem.~\ref{lem:comp_gadget}); when $r_2$ is not a single qubit, this is a $\Pi$-controlled addition on $r_4$ for $\Pi$ a projector with form $(I - |0\rangle\langle 0|_{r_2})$. Note also the non-standard ordering of registers to make depiction easier. In our recursive construction, each $U_{(I\otimes \ast)}$ is replaced by an analogous circuit which acquires (1) a new copy of $r_0$ for each recursive level, and (2) contributes to the size of $r_4$ to multiply all block encodings, but otherwise acts identically on $r_1$, $r_3$, and $r_2$ (depictions of this composition are given across Figs.~\ref{fig:reduced_form_prod} and \ref{fig:recur_prod}). Finally, recall if the registers $r_2$ and $r_4$ are measured in the all-zeros state, then the $r_0$ and $r_1$ registers experience only the identity operation, and $(A\circ B)$ is applied to $r_3$.
    }
    \label{fig:full_prod}
\end{figure}

\begin{figure}
    \centering
    \includegraphics[width=\linewidth]{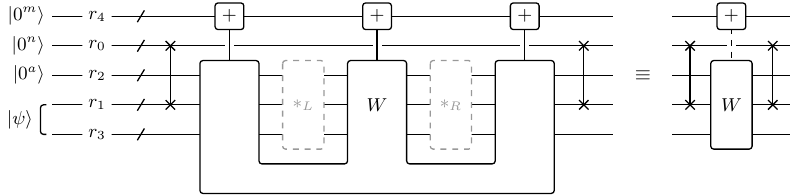}
    \caption{Two equivalent, simplified depictions of the circuit in Fig.~\ref{fig:full_prod}. On the left the unitary-valued function $W$ takes two unitary-valued inputs, here denoted $(\ast_L)$ and $(\ast_R)$, which in our setting will both be block encoding unitaries of matrices that have the form $(I\otimes M)$ for some square matrix $M$. This unitary-valued function (i.e., a circuit with holes) is further condensed on the right, where the free arguments (or holes) are suppressed. Note that wires with slashes represent multi-qubit registers, between which swap gates denote qubit-wise swaps. Additionally, the projection controlled addition operators (see Lem.~\ref{lem:comp_gadget}), here with their targets denoted with $(+)$, have their controls absorbed into $W$; on the left-hand side, the sequence of such addition operators is represented all together by a dotted line.}
    \label{fig:reduced_form_prod}
\end{figure}

\begin{figure}
    \centering
    \includegraphics[width=0.83\linewidth]{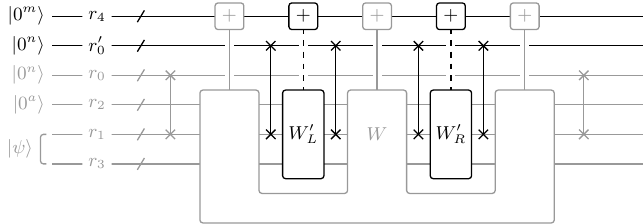}
    \caption{A depiction of the second level of the recursive use of the circuit form of Fig.~\ref{fig:full_prod} by itself, making use of both forms of the condensed notation introduced in Fig.~\ref{fig:reduced_form_prod}. Note here that the register used for the compression gadget, $r_4$, is not duplicated, and remains the target for the addition gadgets of protocols deeper in the recursive construction. A new register of size $n$, namely $r_0'$, is however initialized to $|0\rangle$, with a new copy introduced at each level of recursion, though \emph{re-used between subroutines at the same level}, given the catalytic behavior of the original circuit. This catalysis succeeds under the assumption that all block encodings are successfully applied (equivalently, that the registers $r_2$ and $r_4$ both end up in the all zeros state). Note that all of the swaps have the same target register, $r_1$, and thus for each level of the recursion subroutines receive `fresh' copies of $|0^n\rangle$ to correctly carry out the swap-copy operation (Lem.~\ref{lem:swap_copy}); this allows both $W_L$ and $W_R$ to be `seen' by the outer protocol $W$ as block encodings of the form that $W$ expects, as discussed in the weaving lemma (Lem.~\ref{lem:rec_swap_copy}).}
    \label{fig:recur_prod}
\end{figure}

Having constructed block encodings of the $d$-th element-wise product of a block-encoded matrix with itself in Thm.~\ref{thm:iterated_prod}, we can now work towards applying a desired \emph{element-wise polynomial function} by taking linear combinations of these unitaries. As indicated in the basic statement of LCU (Thm.~\ref{thm:be_lcu}), this requires us to construct a so-called \textsc{select} unitary, which applies an indexed unitary $U_k$ coherently when an auxiliary register is in the state labeled $|k\rangle$ (i.e., a quantum state labeled by a binary encoding of the integer $k$). Here we merely give our lemma characterizing the cost of building \textsc{select} given access to \emph{controlled block encodings} of a desired matrix $A$. While this construction `sits above' our technique for taking element-wise powers of matrices, we emphasize that the technique used in the prior work \cite{gychanar_had_prod_25} to build \textsc{select} is under-specified and (as implied) insufficient\slash incorrect. We take time to raise and rectify this gap in the dedicated Appx.~\ref{appx:lcu_nonstandard}, which also includes the constructive proof of Lem.~\ref{lem:ew_select}.

\begin{lemma}[Building \textsc{select} for the element-wise product] \label{lem:ew_select}
    Let $U_A$ a $(1, a, 0)$ block encoding of an $n$-qubit linear operator $A$ and $d \in \mathbb{N}$. Then the \textsc{select} unitary for element-wise powers of $A$ up to degree $d$, i.e.,
    \begin{equation} \label{eq:qewt_select}
        U_\text{\textsc{select}} \equiv \sum_{k = 0}^{d} 
        |k\rangle\langle k |\otimes U_{A^{\circ k}} + \cdots,
    \end{equation}
    can be built using $d$ copies of \emph{controlled} $U_A$, $\log{d}$ qubits of additional space, and $\mathcal{O}(d\log{d})$ single and two qubit gates. Here we have suppressed the extra terms which appear when $d$ is not a power of two, as shown in (\ref{eq:lcu_select_def}). Note here that the notation $A^{\circ 0}$ is taken to be the \emph{usual identity matrix} on $n$ qubits. For a proof of this statement see Lem.~\ref{lem:ew_select_proof} in Appx.~\ref{appx:lcu_nonstandard}.
\end{lemma}

\begin{theorem}[Quantum element-wise transform; \emph{formal}] \label{thm:qewt}
    Let $f$ be a polynomial (with a trivial constant term) of degree $d$ satisfying $|f| \leq 1$ on $[-1,1]$, with coefficients $c_{k}, k \in \{0, \dots, d\}$, $c_0 = 0$. Then, given access to controlled-$U$, where $U$ is an $(\alpha, a, \varepsilon)$ block encoding of an $n$-qubit operator $A$, there exists a quantum circuit preparing a $(\beta, b, \delta)$ block encoding of $f^\circ (A/\alpha)$---\emph{the element-wise application of $f$ to $A/\alpha$}---that requires $\mathcal{O}(d)$ copies of controlled-$U$, as well as $\mathcal{O}(d\log{d})$ single\slash two-qubit gates, where $(\beta, b, \delta)$ satisfy
        \begin{equation}
            \beta = \sum_{k = 1}^{d} |c_k|,
            \quad
            b = \mathcal{O}(a + n\log{d}),
            \quad
            \delta = \sum_{k = 1}^{d} |c_k|\,\Big[(1 + \varepsilon/\alpha)^k - 1\Big].
        \end{equation}
    \begin{proof}
        Applying LCU to our element-wise product gadget (Thm.~\ref{thm:iterated_prod}) requires us to construct the \textsc{select} unitary as mentioned previously and discussed variously in foundational work on LCU \cite{cw_lcu_12, bccks_ham_lcu_14, bck_ham_lcu_15, gslw_19}. The ultimate result, presented in condensed form in the previous Thm.~\ref{thm:be_lcu}, is the following:
        \begin{equation}
            V = \sum_{k = 1}^{d} |k\rangle\langle k|_{r_5} \otimes \big[U_{A^{\circ k}}\big]_{r_0 r_1 r_2 r_3 r_4} + \left[I_{r_5} - \sum_{k = 1}^{d}|k\rangle\langle k|_{r_5}\right]\otimes I_{r_0 r_1 r_2 r_3 r_4}.
        \end{equation}
        Notably, we cannot, as incorrectly implied by previous work, simply matrix multiply together element-wise products $A^{\circ 2^j}$ per Lem.~\ref{lem:gate_cost_standard} from \cite{cks_linear_17}); instead we construct projection-controlled versions of controlled $U$ per Lem.~\ref{lem:ew_select}, before applying our element-wise power construction (Thm.~\ref{thm:iterated_prod}). This requires an additional register $r_5$ of size $\mathcal{O}(\log{d})$, as is expected, and gives the block encoding parameters quoted in the theorem statement. Note that the form of $\delta$ follows from the linear relation between input and output error in both (1) each element-wise power $k$ and (2) each coefficient $c_k$; a detailed calculation of this error is provided in the dedicated Appx.~\ref{appx:error_prop} (specifically Prop.~\ref{prop:error_lcu_ewt}), where we identify a few errors in extant literature.
    \end{proof}
\end{theorem}

\begin{figure}
    \centering
    \includegraphics[width=0.68\linewidth]{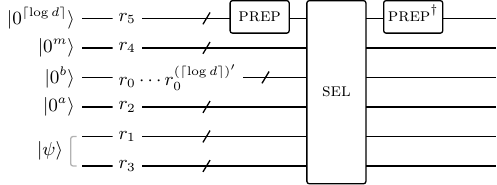}
    \caption{The application of the standard LCU circuit to generate the full quantum element-wise transform of Thm.~\ref{thm:qewt}; unpacking this figure through previous definitions and circuits gives the full transform. I.e., here \textsc{sel} denotes the \textsc{select} unitary (\ref{eq:qewt_select}) constructed in Lem.~\ref{lem:ew_select} (whose proof and discussion is contained in Appx.~\ref{appx:lcu_nonstandard}) for the element-wise product circuit of Thm.~\ref{thm:iterated_prod} (depicted in Fig.~\ref{fig:recur_prod}). Here $b = \mathcal{O}(n\log{d})$ follows from the number of qubits (copies of $r_0$) required to recursively apply the weaving lemma (Lem.~\ref{lem:rec_swap_copy}), and the new register $r_5$ on $\mathcal{O}(\log{d})$ qubits has been instantiated large enough to encode $d$ terms using standard LCU techniques (see Thm.~\ref{thm:be_lcu}).}
    \label{fig:ewt_full}
\end{figure}

The astute reader might notice that the depth of our log-space protocol is larger than that of the previous linear-space one, here due to our inclusion of mask unitaries. Consequently, the overall success probability of our protocol depends in some complex way on the successful application of all of these $\mathcal{O}(d)$ masks. This success probability is dependent on both the initial quantum state and the input block encodings, and thus just as in the application of the standard compression gadget (as well as prior constructions for element-wise products, which also multiply $\mathcal{O}(d)$ block encodings), there may exist initial states and input block encodings which can suppress this success probability arbitrarily (e.g., when block encoding the all zeros matrix). In the following Sec.~\ref{sec:alt_const} we address this problem in specific contexts where this success probability can be bounded usefully from below, and show that for a variety of non-trivial block encoding inputs (e.g., the near-identity setting, as also investigated in \cite{vg_reducing_space_25}) our log-space construction succeeds with good probability. In general the study of the success probability of these basis-dependent matrix transforms is quite rich, as unlike in QSP\slash QSVT, or other spectral mapping techniques, the relation between the chosen function and the resultant effect on the input's spectrum is only cursorily understood in theory \cite{hj_topics_91} (e.g., among the most celebrated results in the study of element-wise products is that they preserve the cone of PSD matrices). Consequently we expect the utility of the QEWT will be quite dependent on known properties of the input.

As a small addendum on this topic, we can show a simple corollary regarding space-time trade-offs for algorithms computing element-wise functions. I.e., as the inputs to our compression scheme are arbitrary block encodings we can easily hybridize the standard and compressed approaches to iterated element-wise products. Such hybrid methods may be of independent interest under experimental\slash architectural resource constraints, or for known instance sizes, and allow for flexibility in compilation.

\begin{corollary}[QEWT time-space tradeoff] \label{cor:time_space}
    Let $f$ of degree $d$ defined as in the main Thm.~\ref{thm:qewt}. Then given some $q \in \mathbb{N}$ we can compute the QEWT according to any division of $\mathcal{O}(d)$ element-wise product gadgets into at most $\lceil d^{1/q}\rceil$ blocks of size at most $\lceil d^{(q - 1)/q}\rceil$, where the blocks are computed using our logarithmic-space method and element-wise multiplied using the standard linear-space method. Then the overall protocol for computing the same element-wise transform uses an auxiliary register of size $b'$ (using the same conventions as the main theorem):
        \begin{equation}
            b = \mathcal{O}\left(a + n\left[1 - \frac{1}{q}\right]\log{d}\right),
            \quad\quad
            b' = \mathcal{O}\left(b\, d^{1/q} + n d^{1/q}\right),
        \end{equation}
    where we have computed the intermediate space complexity $b$ for the inner logarithmic-space protocol, which gives a total auxiliary space complexity $b'$ of the composition of the linear- and logarithmic-space protocols:
        \begin{equation}
            b' = 
            \mathcal{O}\left(ad^{1/q} + n\left[2 - \frac{1}{q}\right] d^{1/q}\log{d}\right).
        \end{equation}
    Moreover, the depth of the resulting circuit is based on the depth of the largest sequentially-computed block, which is $\mathcal{O}(d^{(q - 1)/q})$, which approaches $\mathcal{O}(d)$ as expected for large $q$.

    \begin{proof}
        The proof, which we do not provide in detail, follows straightforwardly from the generic input assumptions of both the logarithmic-space and linear-space element-wise product gadgets; i.e., the inner protocol (logarithmic-space) is treated as a black box by the outer protocol (linear-space), meaning we can compute the overall asymptotic auxiliary space complexity purely from the block size and sum the results.
    \end{proof}
\end{corollary}

As an example of Cor.~\ref{cor:time_space}, for $q = 2$ we have, for each element-wise product, at most $\sqrt{d}$ blocks of size at most $\sqrt{d}$, yielding an intermediate auxiliary space register of size $\mathcal{O}([a + n]\sqrt{d}\log{d})$ and circuit depth $\mathcal{O}(\sqrt{d})$. For $q$ small (i.e., where each block is of size at most one) this hybrid approach reduces to our logarithmic-space construction. While beyond the scope of this work, this division of a single element-wise transform between these linear- and logarithmic-space methods is quite generic, does not require uniform block size, and offers wide opportunity for optimizations given side-information about the input block encoding. At first glance (beyond circuit depth reduction), the merit of this hybridization is not obvious; however, as is discussed in the alternative constructions of the next Sec.~\ref{sec:alt_const}, as well as in the appendix devoted to properties of the element-wise product (Appx.~\ref{appx:ewt_prop}), the behavior of basic properties of linear operators under element-wise products (e.g., rank, condition number, spectral norm), and thus the success probability of the resulting block encoding circuits, is much more erratic than in the usual matrix multiplication setting---complicating questions of optimal resource complexity.

%%%%%%%%%%%%%%%%%%%%%%%%%%%%%%%%%%%%%%%%%%%%%
\section{Alternative constructions and special cases} \label{sec:alt_const}

\noindent The main construction of the previous section was presented piece-wise to clarify the interaction between its components---here the compression gadget, weaving lemma, and relation between the Kronecker, matrix, and element-wise products. Here we provide an alternative construction for computing the same element-wise transform, as well as identify common settings\slash assumptions under which these alternative constructions allow further performance improvements. The high-level action is that we dispense with the `masking unitaries' (Def.~\ref{def:mask}) and instead employ LCU techniques to achieve the block-diagonality that our swap-copy operation requires. These LCU methods also allow us to employ a few common block encoding manipulation techniques (namely, oblivious and fixed-point amplitude amplification, defined briefly below) to further improve the success probability of our protocol when certain preconditions on the input are met.

\begin{theorem}[Fixed-point amplitude amplification; Thm.~27 from \cite{gslw_19}] \label{thm:fpaa}
    Let $U$ a unitary and $\Pi$ an orthogonal projector such that $\alpha |\psi_G\rangle = \Pi U |\psi_0\rangle$ and $\alpha > \eta > 0$. Then there is a unitary $U'$ such that $\lVert |\psi_G\rangle - U'|\psi_0\rangle \rVert \leq \varepsilon$ which uses a single auxiliary qubit, and comprises $\mathcal{O}((1/\eta)\log{(1/\varepsilon)})$ copies of the following gates: $U$, $U^\dagger$, $C_\Pi\text{NOT}$, $C_{|\psi_0\rangle\langle\psi_0|}\text{NOT}$, and single-qubit rotations.
\end{theorem}

\begin{theorem}[Robust, fixed-point, oblivious amplitude amplification; modified from Thm.~28 of \cite{gslw_19}] \label{thm:roaa}
    Let $U$ a unitary and $\Pi, \Pi'$ orthogonal projectors such that
        \begin{equation}
            \big\lVert\alpha W - \Pi' U \big\rVert \leq \varepsilon,
        \end{equation}
    for $\alpha > \eta > 0$ and where $W$ is an isometry ($W: \text{img}(\Pi) \rightarrow \text{img}(\Pi')$). Then we can construct a unitary $U'$ such that for all $|\psi\rangle \in \text{img}(\Pi)$:
        \begin{equation}
            \Big\lVert W|\psi\rangle - \Pi' U'|\psi\rangle\Big\rVert \leq \varepsilon/\alpha,
        \end{equation}
    which uses a single auxiliary qubit and $\mathcal{O}((1/\eta)\log{(1/\varepsilon)})$ copies of the following gates: $U$, $U^\dagger$, $C_\Pi\text{NOT}$, $C_{\Pi'}\text{NOT}$, and single-qubit rotations. 
\end{theorem}

\begin{theorem}[Uniform singular value amplification; modified Thm.~30 of \cite{gslw_19}] \label{thm:uniform_sv_amp}
    Let $U$ an $(\alpha, a, 0)$-block encoding of $A$. Then we can construct a $(\lVert A\rVert_{1}/(1 - \eta), a + 1, \varepsilon \lVert A \rVert_2)$-block encoding of $A$ using $\mathcal{O}((\gamma\eta)^{-1}\log{[(\gamma\varepsilon)^{-1}]})$ queries to $U$, $U^\dagger$, where $\gamma \equiv \lVert A \rVert_1/\alpha$.
\end{theorem}

One of the ways we will modify our initial element-wise product algorithm is through using, at each level of our recursive construction, multiple copies of the input block encodings. It is thus necessary to establish how this additional processing affects the total query, gate, and space complexity of the full algorithm. We give a remark below on this topic before providing a few more helper-lemmata and propositions, then explicitly providing our LCU- and amplitude-amplification-dependent constructions.

\begin{remark}[Query complexity for alternative protocols] \label{rem:alt_construct}
    In our recursively defined quantum element-wise product the query complexity of our `base case' (e.g., the circuit in Fig.~\ref{fig:recur_prod}) completely defines the query complexity of the general protocol. Computing this query complexity in the two-argument case is simple: if in the base case each argument, $(\ast_1)$ and $(\ast_2)$, requires $m$ copies, then the $k$-th recursive level yields a degree $2^k$ product at a cost of $(2m)^k$. Consequently achieving a degree $d$ product requires $k = \log{d}$, for a total query complexity of $\mathcal{O}((2m)^{\log{d}}) = \mathcal{O}(d^{1 + \log{m}})$. For constant $m$ this is polynomial in $d$.

    Moreover, we will also investigate the approximate setting (e.g., Thm.~\ref{thm:amp_ewt}) where $m$ is not constant but depends weakly on the approximation error, i.e., $m = \mathcal{O}((1/\delta)\log{(1/\varepsilon)})$ for $\delta$ bounded below by some constant. Here the total query complexity to achieve error $\varepsilon$ would scale instead as $d^{\log{((1/\delta)\log(d/\varepsilon))}}$, which is quasi-polynomial in $d$.
\end{remark}

The key observation toward our construction is that the application of the swap-copy operation (Lem.~\ref{lem:swap_copy}) does not require that the block encoding be zero outside of the copied block, but only \emph{block diagonal}. An easier way to achieve this instead of masking unitaries $U_M$ depicted in Fig.~\ref{fig:full_prod} is through the use of linear combination of unitaries (LCU) techniques (Thm.~\ref{thm:be_lcu}). In our case the LCU which achieves the required block diagonal structure is quite simple (see Fig.~\ref{fig:lcu_ewt_gadget}), and does not require controlled access to nor multiple copies of the input block encodings. In what follows we give the circuit for this LCU (Fig.~\ref{fig:lcu_ewt_gadget}), recovering a new statement for the complexity of computing element-wise powers (Lem.~\ref{lem:lcu_alt_ewt}) and thus transforms (Thm.~\ref{thm:lcu_alt_ewt_func}). We then investigate means of increasing the success probability through amplitude amplification subroutines (Thm.~\ref{thm:amp_ewt}) followed by analysis in special cases (Cors.~\ref{cor:near_iden_ewt} and \ref{cor:approx_ewt_func}). For simplicity we first state a fact on the success probability of generic LCU protocols.

\begin{proposition}[LCU protocol success probability] \label{prop:lcu_success}
    Let a non-unitary $V$ be expressible as the sum of $K$ unitaries according to coefficients $c_k$ for $k \in \{0, \dots, K - 1\}$. Then the standard LCU technique to apply $V$ to some initial state $|\psi\rangle$ uses an auxiliary register of $\lceil\log{K}\rceil$ qubits and succeeds with probability
        \begin{equation}
            p_{\text{succ}} = \lVert V|\psi\rangle \rVert^2\,\Big(\sum_{k} |c_k|\Big)^{-2}.
        \end{equation}
    The success probability is affected by both (1) the non-unitarity of the desired operator, and (2) the $\ell_1$ norm of the coefficients of the linear combination. Note that the spectral norm of the constituent unitaries ensures this success probability lies between zero and one.
\end{proposition}

We can investigate this generic statement in the specific setting of using LCU to achieve the self-element-wise product of a block encoding with itself. Here we will use a simple, two-coefficient LCU subroutine (depicted in Fig.~\ref{fig:lcu_ewt_gadget}) to recover the preconditions for the application of our swap-copy operation.

\begin{lemma}[LCU-based element-wise product] \label{lem:lcu_alt_ewt}
    Let $U_A$ be $(1, a, 0)$ block encoding of an $n$-qubit operator $A$. Then we can construct a $(1, a + n[2\log{d} + 1], 0)$ block encoding of the $d$-th order element-wise product of $A$ with itself ($A^{\circ{d}}$) using $d$ copies of $U_A$ and $\mathcal{O}(d)$ additional quantum gates.
    Moreover, the probability of successfully applying this element-wise product depends on both $A$ and the initial state $|\psi\rangle$:
    \begin{equation}
        p_{\text{succ}} 
        =  
        \lVert A^{\circ d} |\psi\rangle\rVert^2.
    \end{equation}
    \begin{proof}
        In this alternative construction, rather than applying masking unitaries, we use a simple LCU to zero-out the off-diagonal blocks of our input, such that the swap-copy (Lem.~\ref{lem:swap_copy}) can be applied. The key observation is that applying a sign to the off-diagonal (non-square) blocks of our block encoding of $P(A\otimes B)P^\dagger$ can be achieved by a unitary that acts like a $Z$ gate on subspaces, where denoted by the unitary $Z^\perp$:
        \begin{equation}
            Z^{\perp} \equiv\left(- |0\rangle\langle0|_{r_1} + \sum_{j = 1}^{2^n - 1} |j\rangle\langle j|_{r_1}\right)\otimes\sum_{k = 0}^{2^n - 1} |k\rangle\langle k|_{r_3}
            =
            (I_{r_1} - 2\Pi_{|0\rangle\langle 0|_{r_1}})\otimes I_{r_3},
        \end{equation}
        which is just a reflection about the state $|0\rangle$ in the $r_1$ register. When applied to a block matrix (taking for the moment $A \in \mathbb{C}^{2^n\times 2^n}$ and $D \in \mathbb{C}^{(2^{2n} - 2^n)\times(2^{2n} - 2^n)}$ and $B, C$ such that the entire matrix is square) this appends a sign to only the off-diagonal (non-square) blocks:
        \begin{equation}
            Z^{\perp}
            \equiv
            \begin{bmatrix}
                - I_{2^n} & \\
                & I_{(2^{2n} - 2^n)}
            \end{bmatrix}_{r_1 r_3}
            \implies
            Z^{\perp}
            \begin{bmatrix}
                A & B\\
                C & D
            \end{bmatrix}
            Z^{\perp}
            =
            \begin{bmatrix}
                A & -B\\
                -C & D
            \end{bmatrix}.
        \end{equation}
        We can now use standard LCU techniques to take the linear combination of $U_{A}$ and $Z^{\perp} U_{A} Z^{\perp}$, weighted by coefficients $c_0 = c_1 = 1/\sqrt{2}$. A circuit for this unitary is given in Fig.~\ref{fig:lcu_ewt_gadget}, uses only a single copy of $U_A$, and can be recursively self-composed almost identically to our original circuit in Fig.~\ref{fig:recur_prod}, save with our new single-qubit LCU register $r_5$, and a controlled addition operation from this register to the compression gadget register $r_4$ after each pair of Hadamard gates, at each recursive level. This requires a single additional qubit, and succeeds with probability subnormalized by the squared one-norm of the LCU coefficients $(|c_0| + |c_1|)^2 = 2$.
    \end{proof}
\end{lemma}

\begin{remark}[On the difference between Lem.~\ref{lem:lcu_alt_ewt} and Thm.~\ref{thm:iterated_prod}] \label{rem:lcu_iter_diff}
    Looking closer at the circuit in Fig.~\ref{fig:lcu_ewt_gadget}, we can see that if one were to insert a pair of Hadamard gates on $r_5$, split the wire between them, and create two qubits experiencing the first and second half of the original gates interacting with $r_5$, they would recover the construction in our original Thm.~\ref{thm:iterated_prod}. The saving here stems from the fact that we only require \emph{block-diagonality} for the swap-copy operation, and not all zeros outside the top left block. This allows us to apply an XOR rather than a true sum to the register $r_5$, resulting in saving an additional qubit per recursive self-composition of our circuit.

    Concretely, the space requirements scale with the number of controlled-additions to the compression gadget register at each recursive level. In the usual setting this per-level complexity is $4$, while in the LCU case it is approximately $2$ (where the single qubit LCU level is shared across recursive levels, with periodic additions made to the compression gadget register). This gives a total compression gadget size in the two cases of
    \begin{align}
        \log{4^{\lceil\log{d}\rceil}} = \lceil\log{d}\rceil\log{(4)} &= 2\lceil\log{d}\rceil,\\
        \log{2^{\lceil\log{d}\rceil}} + \log{\left(\sum_{k = 0}^{\lceil\log{d}\rceil - 1} \lceil \log{(d/2^k)} \rceil \right)} &= \lceil\log{d}\rceil + \mathcal{O}(\log{\log{d}}).
    \end{align}
    We also note that this space use is distinct from the identical $\mathcal{O}(\log{d})$ space used taking when linear combinations of element-wise powers.
\end{remark}

\begin{figure}
    \centering
    \includegraphics[width=0.8\linewidth]{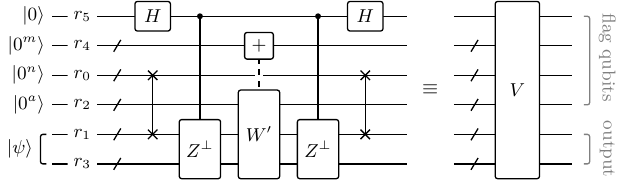}
    \caption{Alternative circuit for constructing a block encoding unitary $V$ of the cleaned element-wise product (Prop.~\ref{prop:clean_ewt_prod_form}) of two block encoded matrices. Here we use the same reduced notation for $W'$ as in Fig.~\ref{fig:reduced_form_prod}, where here $W'$ does not contain the masking unitaries of Fig.~\ref{fig:full_prod}, and where recursive self-composition proceeds the same as in Fig.~\ref{fig:recur_prod}, save that $r_5$ (a single-qubit register) also becomes the control of a projector-controlled addition to $r_4$, and is re-used across the same recursive level. Here $Z^\perp$ is a controlled reflection about the $|0^n\rangle$ state in $r_1$, as defined in Lem.~\ref{lem:lcu_alt_ewt}.}
    \label{fig:lcu_ewt_gadget}
\end{figure}

Once we have a protocol for the element-wise product we can use the same techniques as in Thm.~\ref{thm:qewt} to build out the full element-wise transform. We give this extension in Thm.~\ref{thm:lcu_alt_ewt_func}, followed by further discussion of the success probability, and select instances in which it can be improved by amplitude amplification techniques.

\begin{theorem}[LCU-based element-wise function] \label{thm:lcu_alt_ewt_func}
     Let $f$ be a polynomial (with a trivial constant term) of degree $d$ satisfying $|f| \leq 1$ on $[-1,1]$, with coefficients $c_{k}, k \in \{0, \dots, d\}$, $c_0 = 0$. Then, given access to controlled-$U$ and controlled-$U^\dagger$, where $U$ is an $(\alpha, a, \varepsilon)$ block encoding of an $n$-qubit operator $A$, there exists a quantum circuit preparing a $(\beta, b, \delta)$ block encoding of $f^\circ (A/\alpha)$---\emph{the element-wise application of $f$ to $A/\alpha$}---that requires $\mathcal{O}(d)$ copies of controlled-$U$ and $U^\dagger$, as well as $\mathcal{O}(d)$ single\slash two-qubit gates, where $(\beta, b, \delta)$ satisfy
        \begin{equation}
            \beta = \sum_{k = 1}^{d} |c_k|,
            \quad
            b = \mathcal{O}(a + n\log{d}),
            \quad
            \delta = \sum_{k = 1}^{d} |c_k|\,\big[(1 + \varepsilon/\alpha)^k -  1\big].
        \end{equation}
    Moreover, the probability of successfully applying this element-wise function has the following form, dependent on $f$, $A$, and the initial state:
    \begin{equation}
        p_{\text{succ}} = 
        \Big(\sum_{k} |c_k|\Big)^{-2}
        \,\lVert f^{\circ}(A/\alpha) |\psi\rangle\rVert^{2}. 
    \end{equation}
    \begin{proof}
        Proof of this statement follows immediately from the term-wise complexity of the element-wise power subroutine in Lem.~\ref{lem:lcu_alt_ewt} and the previous construction of the full element-wise transform in Thm.~\ref{thm:qewt} using standard LCU techniques as depicted in Fig.~\ref{fig:ewt_full}.
    \end{proof}
\end{theorem}

A drawback of the constructions we have given thus far is that in the worst case (both for the log-space and linear-space protocols for element-wise products and functions) the probability of successfully applying the transformed block encoding can decrease exponentially\footnote{Note this is distinct from the separate but similarly thorny problem that if an initial state is nearly annihilated by the block encoding, i.e., that $\lVert A|\psi\rangle\rVert^2$ is small, which can generically occur if $A$ is far from unitary, then our protocol can also fail with high probability.} in the degree of the applied polynomial function. In what follows we discuss a few techniques by which (under certain conditioning assumptions) this problem can be mitigated. For the most part these problems have been encountered in prior work for taking long matrix products of block encodings, which can be modified to apply to our Kronecker-to-element-wise product identities.

Restricting ourselves to the usual matrix product compression gadget, we first deal with the possibility that each input $(\alpha_k, a_k, \varepsilon_k)$ block encoding for some linear operator $A$ has some non-optimal subnormalization $\alpha_k \geq \lVert A \rVert$. In the following lemma, adapted from Thm.~4 of \cite{flt_time_marching_23}, it can be shown that uniform singular value amplification (Thm.~\ref{thm:uniform_sv_amp}) can be used to multiply this subnormalization by a factor of $\gamma \equiv \alpha(1 - \delta)/\lVert A\rVert$, up to precision $\varepsilon$, in turn mitigating the effect of these $\alpha_k$ when they appear in long products of block encoded linear operators.

\begin{lemma}[Amplified compression gadget; adapted from \cite{flt_time_marching_23}] \label{lem:amp_comp_gadget}
    Let $U_{A_k}$ be $(\alpha_k, a_k, \varepsilon_k \lVert A_k\rVert)$ block encodings of (possibly non-unitary) $A_k$ for $k \in \{0,1, \dots, K - 1\}$ and $\sum_{k} \varepsilon_k \leq 1/2$. Then for any $0 < \varepsilon < 1/(2K)$ we can construct an $(\alpha', a', \varepsilon')$ block encoding of the matrix product $A_{K-1}\cdots A_{1} A_{0}$, where $\alpha'$ satisfies
        \begin{equation}
            \frac{1}{2}(1 - \delta)^{-K}\prod_{k} \lVert A_k\rVert 
            \leq 
            \alpha' 
            \leq 
            e^{1/2}(1 - \delta)^{-K}\prod_{k} \lVert A_k\rVert,
        \end{equation}
    and $a'$ and $\varepsilon'$ satisfy
    \begin{equation}
        a' = \max{a_k} + \log{(K)} + 2,
        \quad\quad
        \varepsilon' = e^{1/2} \left[K \varepsilon + \sum_{k} \varepsilon_k \right]\prod_{k} \lVert A_k\rVert.
    \end{equation}
    Moreover, this amplification procedure uses $\mathcal{O}(\alpha_k/(\delta\lVert A_k\rVert)\log{(\alpha_k/(\varepsilon\lVert A_k\rVert))})$ copies of each controlled-$U_k$ and its inverse.
\end{lemma}

The point of the provided Lem.~\ref{lem:amp_comp_gadget}, e.g., choosing $\delta = \mathcal{O}(1/K)$, is that it modifies the success probability of the compression gadget to near its best possible value:
    \begin{equation}
        \frac{\lVert A_{K-1}A_{K-2}\cdots A_{0}|\psi\rangle\rVert^2}{\lVert A_{K - 1}\rVert^2 \cdots \lVert A_{1}\rVert^2 \lVert A_{0}\rVert^2},
    \end{equation}
i.e., up to the chosen error $\varepsilon$, on which the query complexity depends only weakly. Consequently if we have good knowledge of an upper bound on the spectral norm of the element-wise products at each level of our recursion, we improve our success probability at the cost of query complexity. In the next statement we apply the results of Lem.~\ref{lem:amp_comp_gadget} to our compressed element-wise product gadget (Thm.~\ref{thm:iterated_prod}); as mentioned this improves the success probability of our protocol in terms of the subnormalization factors of the input block encodings at the cost of additional query complexity (and thus a small amount of additional space) as well as an introduced approximation error.

\begin{theorem}[Amplified element-wise product] \label{thm:amp_ewt}
    Let $U_A$ an $(\alpha, a, \varepsilon)$ block encoding of an $n$-qubit operator $A$. Then we can construct an $(\gamma, \mathcal{O}(a + n\log{d}), \varepsilon')$ block encoding of the $d$-th element-wise product of $A$ with itself, denoted $(A)^{\circ d}$ using $Q = \mathcal{O}(d^{\log{(\eta/\delta)\log{(\eta/\varepsilon)}}})$ queries to $U_A$, where these block encoding parameters have the form
        \begin{equation}
            \gamma = \mathcal{O}{\Big(\alpha^{Q}\Big)},
            \quad\quad 
            \varepsilon' = \mathcal{O}{\Big([\alpha + \varepsilon]^{Q} \!- \alpha^{Q}\Big)},
        \end{equation}
    where $\eta \equiv \alpha/\lVert A \rVert \geq 1$, and where we will usually take $\delta = \mathcal{O}(1/d)$. Moreover, the lower bound on the success probability of the resulting protocol, $p_{\text{succ}}$, compared to that of the construction in Thm.~\ref{thm:iterated_prod}, $p_{\text{succ}}'$, is improved:
    \begin{equation}
        p_{\text{succ}}' = \alpha^{-2d}\lVert A^{\circ d} |\psi\rangle \rVert^2
        \;\;\mapsto\;\;
        p_{\text{succ}} \geq (1 - \varepsilon)^{2d}(1 - \delta)^{2d}\,\lVert A^{\circ d} |\psi\rangle \rVert^2.
    \end{equation}

    \begin{proof}
        The query complexity of this amplification at any given recursive level is the quoted
        \begin{equation}
            \mathcal{O}(\alpha/(\delta\lVert A\rVert)\log{(\alpha/(\varepsilon\lVert A\rVert))}) \equiv \mathcal{O}(\eta/\delta \log{(\eta/\varepsilon)}),
        \end{equation}
        where $\eta = \alpha/\lVert A \rVert \geq 1$. Using the general relation worked out in Rem.~\ref{rem:alt_construct} we can compute the asymptotic scaling of the total query complexity for achieving a degree-$d$ element-wise product of amplified block encodings
        \begin{equation}
            Q = \mathcal{O}\left(d^{\log{([\eta/\delta]\log{(\eta/\varepsilon)})}}\right).
        \end{equation}
        The choice of $\delta$ required to bring the success probability to within a $\Omega(1)$ multiplicative factor of $\lVert A^{\circ d} |\psi\rangle\rVert$ is thus $\delta = \mathcal{O}(1/d)$, which shows that our resulting query complexity is quasi-polynomial in $d$ (i.e., $\mathcal{O}(d^{\log{(d\log{(1/\varepsilon)})}})$).
    \end{proof}
\end{theorem}

It is well-known in literature (albeit often misquoted\footnote{For a good explanation of why this is, see the discussion given in the pedagogical \cite{lw_lecture_notes_26}.}) that techniques like fixed-point, oblivious amplitude amplification (see Thm.~\ref{thm:roaa}) rely crucially on the block encoded linear operator being (close to) an isometry. When given a general block encoding, as shown in the previous Thm.~\ref{thm:amp_ewt}, the success probability cannot generally be driven close to one for all input states (i.e., obliviously). In the following few statements we will consider the much stronger assumption that our input block encoding (and indeed, its element-wise powers) is close to an isometry. Note for element-wise powers this means our block encoding is close to the (matrix product) identity, which while restrictive is analogous to assumptions made in previous work on compression gadgets \cite{lw_comp_gadget_19, vg_reducing_space_25, flt_time_marching_23}, and allows us to apply robust, oblivious amplitude amplification (ROAA). Following these statements we will discuss a few additional reasonable and useful input assumptions for the QEWT.

\begin{definition}[Near-identity block encodings] \label{def:near_iden_be}
    Let $U_A(k)$ a $k$-parameterized family of $(\alpha, a, \varepsilon)$ block encodings for $k \in \mathbb{N}$. Then we say that $U_A(k)$ is $k$-\emph{near-identity} if
        \begin{equation}
            \lVert U_A(k) - I \rVert = \mathcal{O}(1/k),
        \end{equation}
    equivalently that the maximum singular value of the difference between $U_A(k)$ (the entire block encoding unitary) and the identity matrix of the same dimension scales as $1/k$.
\end{definition}

While the condition of Def.~\ref{def:near_iden_be} appears quite restrictive, we argue, as is also noted in \cite{vg_reducing_space_25}, that similar bounds appear in a wide array of block encoding applications including Hamiltonian simulation and generic differential equation solvers (e.g., for short time evolutions $e^{iHt}$ can yield near-identity block encodings). As in the cited prior work and here, we are interested in this property as it provides a lower bound on the success probability of the (element-wise) multiplication of multiple near-identity block encodings.

\begin{proposition}[Element-wise multiplication of near-identity block encodings] \label{prop:ewt_succ_near_iden}
    Let $U_{A_k}(K)$ for $k \in \{0, 1, \dots, K-1\}$ each $K$-near-identity. Then the element-wise product of all of the $A_k$ succeeds with probability $p_{\text{succ}} = \Omega(1)$, and moreover the product $B = (A_0\circ A_1\circ\cdots\circ A_{K-1})$ satisfies $\lVert B - I\rVert = \mathcal{O}(1)$.

    \begin{proof}
        It can easily be checked that the success probability for a $K$-near-identity block encoding scales as $\text{tr}{(\rho |A|^2)} = \Omega(1 - 1/K)$ for any fixed initial state $\rho$, and consequently the overall success probability scales as $\Omega((1 - 1/K)^K) = \Omega(1)$. The bound on the difference between the element-wise product and the identity follows from the submultiplicativity of the element-wise product in the operator norm (see Appx.~\ref{appx:ewt_prop}).
    \end{proof}
\end{proposition}

\begin{corollary}[Element-wise products between operators near the identity] \label{cor:near_iden_ewt}
    Let $U_A$ an $(\alpha, a, \varepsilon)$ block encoding of some $n$-qubit $A$, and moreover let $U_A$ $d$-near-identity. Then we can construct a $(\beta, b, \delta)$ block encoding of the $d$-th element-wise product of $A$ with itself, denoted $A^{\circ d}$ using $d$ queries to $U_A$, and where $\beta, b, \delta$ satisfy
        \begin{equation}
            \beta = \Omega(1),\quad 
            b = \mathcal{O}(a + n\log{d}),\quad 
            \delta = (\alpha + \varepsilon)^d - \alpha^{d}.
        \end{equation}
    Moreover the probability of successfully applying $A^{\circ d}$ to $|\psi\rangle$ scales as $\Omega(1)$.
    \begin{proof}
        Proof follows immediate by restricting our previous amplified element-wise product (Thm.~\ref{thm:amp_ewt}) to the case that $A$ is $d$-near-identity, which constrains $\alpha = \Omega(1 - 1/d)$ and thus also provides the $\Omega(1)$ lower bound on both $\beta$ and the overall success probability. The lack of a requirement for amplification also reduces us from quasipolynomial to linear query complexity in $d$.
    \end{proof}
\end{corollary}

\begin{corollary}[Approximate element-wise transform] \label{cor:approx_ewt_func}
    Let $U_A$ as in the main Thm.~\ref{thm:qewt}. Then we can create a block encoding of $f^{\circ}(A)$, the element-wise application of $f$ to $A$, with block encoding parameters $(\beta, b, \varepsilon')$ of the form
    \begin{align}
        \beta &= \sum_{k = 1}^{d} |c_k|\,(1 - \delta)^{-k},
        \quad\;\;
        b = \mathcal{O}\Big(a + n(\log{d})\log{(\log{(1/\varepsilon))}}\Big),\\
        \varepsilon' &= \sum_{k = 1}^{d} |c_k|\,\Big[\big([1 - \delta]^{-1} + \varepsilon\big)^k \!- (1 - \delta)^{-k}\Big],
    \end{align}
    where the query complexity to the input block encoding is $d$, and the probability of successfully applying $f^{\circ}(A)$ to the input state is bounded below by $(1 - \varepsilon)^{2d}(1 - \delta)^{2d}\,\lVert f^{\circ}(A) |\psi\rangle\rVert^2/\lVert f\rVert_1^2$, where $\lVert \ast \rVert_1$ refers to the $\ell_1$ norm of the coefficients of $f$. Moreover, if $U_A$ is $d$-near identity, then all of the factors dependent on the amplification, $\delta$, in the block encoding parameters and success probability, reduce to zero.
\end{corollary}

Before moving on to Sec.~\ref{sec:applications}, where we construct and discuss applications of QEWTs, we briefly summarize what the above calculations express regarding the conditions under which applying element-wise functions is expected to be efficient\slash succeed with high probability. In the general setting the best we can do is uniformly amplify all singular values of the input block encoding, which in general costs, at each level of our recursive construction, queries proportional to $\alpha/\lVert A\rVert \geq 1$ (as well as logarithmically on the inverse approximation error), and for non-isometric $A$ can still permit success probability for certain input states to decay exponentially in $d$. Being able to arbitrarily drive the success probability to one requires \emph{both} that $A$ is a near-isometry and that the element-wise product preserves this near-isometric property, which restricts us to the near-identity\footnote{More generally to near-permutation matrices where each element is free up to a complex phase; these are sometimes referred to as \emph{monomial matrices}, and are a representation of a \emph{wreath product} $\mathbb{T} \wr S_{2^n}$. These turn out to be the largest class of matrices preserved under element-wise powers.} setting.

While beyond the scope of this work (and briefly expanded on in our discussion of open directions in Sec.~\ref{subsec:open_dir}), it is easy to enumerate assumptions which may allow more efficient application of the QEWT. These include cases where the input matrices have known block-structure, non-trivial null-space, or known\slash easily preparable eigenvectors, which may either obviate the need to allocate additional space or to apply a masking step before performing a swap-copy operation. We also note the recent explosion of randomized\slash parallelized\slash Trotterized techniques in block encoding algorithms \cite{campbell_qdrfit_19, cstwz_trotter_comm_21, chak_lcu_near_24, martyn_rall_halving_24}, which may provide practical benefit when running these algorithms on constrained devices.

%%%%%%%%%%%%%%%%%%%%%%%%%%%%%%%%%%%%%%%%%%%%%
\section{Applications} \label{sec:applications}

\noindent The original impetus for this work is rooted in the ubiquitous appearance of element-wise products and functions in applied numerical linear algebra. In this section we cover a few of these applications and the induced resource complexity for quantum algorithms addressing them; a few of these, for instance the section on quantum algorithms for common subroutines in machine learning inference (Sec.~\ref{subsec:ml_inf}), e.g., computing attention by way of an element-wise softmax, are in direct comparison to prior work. We also examine more generic applications to computing derived transforms common in classical signal processing, e.g., matrix convolutions (Sec.~\ref{subsec:app_2d_conv}) and block generalizations of the Kronecker and element-wise products (Sec.~\ref{subsec:app_nonstandard_prod}), each of which inherits the space-savings of our main construction.

%%%%%%%%%%%%%%%%%%%%%%%%%%%%%%%%%%%%%%%%%%%%%
\subsection{Machine learning inference} \label{subsec:ml_inf}

\noindent Here we focus on a key subroutine of a previous work on (linear-space) element-wise transforms \cite{gychanar_had_prod_25} which involves applying an element-wise function to a product of block encodings as part of a critical subroutine (self-attention) in transformer model inference \cite{attention_17}. Our construction provides immediate improvements to Thms.~S6, S8, and S9 of \cite{gychanar_had_prod_25}, which we exposit briefly below along with a few critical definitions. We note that our construction specifically applies to the block encoding input model, which subsumes the related \emph{state encoding} (Def.~\ref{def:state_encoding}) model, which should be thought of as a block encoding where only the first column (possibly subnormalized) is specified and transformed. The benefit of our method, as noted also in the cited work, is that element-wise transforms of block encodings are in general much more difficult to achieve than element-wise transforms of state encodings, given the ability in the latter to reduce to the diagonal case where the matrix and element-wise products coincide.

\begin{definition}[State encoding; Def.~2 in \cite{gychanar_had_prod_25}] \label{def:state_encoding}
    A unitary $U_\phi$ is a $(\alpha, a, \varepsilon)$ \emph{state encoding} of an $n$ qubit quantum state $|\phi\rangle$ if
        \begin{equation}
            \lVert |\phi\rangle - \alpha (\langle 0^{a}|\otimes I_{n}) U_\phi (|0^{a + n}\rangle)\rVert_{\infty} \leq \varepsilon.
        \end{equation}
    In other words the unitary $U$ prepares, when applied to the all zeros state, a quantum state $|\phi'\rangle$ that $\varepsilon$-approximates a state with $1/\alpha$ overlap\footnote{In this equation and through the rest of the definition, we correct an improper normalization\slash under-specification of the orthogonal component of the state given in \cite{gychanar_had_prod_25}.} with the desired $|\phi\rangle$:
        \begin{equation}
            U_{\phi}|0^{a + n}\rangle = \alpha^{-1} |0^{a}\rangle|\phi'\rangle + (1 - \alpha^{-2})^{1/2}|\phi^\perp\rangle,
        \end{equation}
    where $|\phi^\perp\rangle$ is some unknown quantum state orthogonal to $|\phi\rangle$ (and whose first $a$ qubits have zero support on $|0^a\rangle$). Here (and in the referenced work) error is taken in the $\ell_{\infty}$-norm, as the interest is in element-wise transformations.
\end{definition}

\begin{theorem}[Block encoding softmax; extension of Thm.~S8 in \cite{gychanar_had_prod_25}] \label{thm:softmax}
    Let $U_A$ an $(\alpha, a, \varepsilon)$ block encoding of a square matrix $A \in \mathbb{C}^{N\times N}$ for $N = 2^n$, and $0 \leq j \leq N-1$ an integer. Then one can prepare a $(1, \mathcal{O}(a + n\log{\ell}), \mathcal{O}(\sqrt{\varepsilon\sqrt{N/Z_j}}))$ state encoding of the softmax of the $j$-th column of $A$:
        \begin{equation}
            |A_j\rangle 
            \equiv
            \frac{1}{\sqrt{Z_j}} \sum_{k = 0}^{N-1} f^{\circ} (A_{kj} /2\alpha) |k\rangle,
        \end{equation}
    where $f(\ast) = \exp{(\ast)}$ is the exponential function, and $Z_j$ is the sum of $f^{\circ}(A_{jk}/2\alpha)$ over all $k$. This requires using $\mathcal{O}(\ell\sqrt{N/Z_j})$ copies of controlled $U_A$, where $\ell = \mathcal{O}(n \log{(1/\varepsilon)})$.
\end{theorem}

\begin{problem}[Defining block encoding self-attention] \label{prob:self_att}
    Let $n, d \in \mathbb{N}$ and $N = 2^n$, and assume access to the unitary $U_{S}$ which is an $(\alpha_s, a_s, \varepsilon_s)$ block encoding of $S \in \mathbb{R}^{N \times d}$. Moreover, for matrices $W_{q}$, $ W_{k}$, and $W_{v}$ all in $\mathbb{R}^{d\times d}$, assume access to unitaries $U_{W_{\ast}}$ which are $(\alpha_w, a_w, \varepsilon_w)$ block encodings of each respectively. Then define $Q = S W_{q}$, $K = SW_{k}$, and $V = SW_{v}$. Given an index $j \in \{0, 1, \dots, N-1\}$, computing the self attention according to $Q, K , V$ requires preparing a block encoding of $G_{j, \ast}$ with the form
        \begin{equation}
            G_{j, \ast} = (\text{softmax}(QK^{\intercal}/\alpha_0)\,V)_{j, \ast}.
        \end{equation}
\end{problem}

\begin{theorem}[Block encoding self-attention; extension of Thm.~S9 in \cite{gychanar_had_prod_25}] \label{thm:self_attention}
    Let $j \in \{0, 1, \dots, N-1\}$ and $\alpha_0 \equiv \alpha_s^2 \alpha_w^2$, and consider the objects\slash conventions established in Prob.~\ref{prob:self_att}. Then we can construct an
    \begin{equation}
    (\alpha_s\alpha_w,\, 
    \mathcal{O}(a_s + a_w + n\log{\ell}),\, 
    \mathcal{O}(\alpha_s\alpha_w\sqrt{(\varepsilon_s + \varepsilon_{w})N/Z_j}))    
    \end{equation}
    block encoding of a matrix $G$ where $G_{j, \ast}$ has the form
    \begin{equation}
        G_{j, \ast} = (\text{softmax}(QK^{\intercal}/\alpha_0)\,V)_{j, \ast},
    \end{equation}
    using $\mathcal{O}(\ell\sqrt{N/Z_j})$ queries to $U_S$, $U_{W_{q}}$, $U_{W_{k}}$ and $U_{W_{v}}$, where $Z_j$ has the form
        \begin{equation}
            Z_j \equiv \sum_{k = 1}^{N} (e^{\circ}(QK^{\intercal})/\alpha_0)_{j, k},
        \end{equation}
    and $\ell = \mathcal{O}(n\log{(\varepsilon_{s} + \varepsilon_{w})^{-1}})$, where we've used the notation $e^{\circ}(\ast)$ to denote the element-wise application of the exponential function.
\end{theorem}

We note again that the application of our construction to the softmax and self-attention subroutines presented above applies primarily in the block encoding input model; as noted in \cite{gychanar_had_prod_25, rr_nonlinear_23}, if our goal instead is to just compute non-linear transformations of state encodings, then we can appeal to other methods that can convert our problem to one involving diagonal block encodings for which even more space efficient QSVT methods can apply. In this sense our method, while not necessary for certain state preparation problems, is strictly more powerful, and does not necessarily acquire the exponential (in dimension, or subnormalization) overhead that column-wise application of non-linear state transformation methods would accrue. 

%%%%%%%%%%%%%%%%%%%%%%%%%%%%%%%%%%%%%%%%%%%%%
\subsection{Matrix convolutions} \label{subsec:app_2d_conv}

\noindent Among the most straightforward applications of element-wise transforms of block encodings appears in multi-dimensional signal processing generally \cite{learning_had_prod_25, hj_topics_91}, which has also been a recent focus of other block encoding-based quantum algorithms \cite{rhgpr_acc_inf_25, mw_gqsp_24}. Concretely for discrete signals of finite support the element-wise product has a well-known and particularly nice interpretation as a certain convolution in Fourier space. In this section we define these convolutions in the two-dimensional setting and provide a simple quantum algorithm (Thm.~\ref{thm:be_2d_conv}) based on our previous construction for computing them; this provides an immediate extension to the one-dimensional or vectorized convolution algorithms of previous work \cite{rhgpr_acc_inf_25, mw_gqsp_24}, and suggests a wealth of further applications based on performing variant quantum element-wise transforms in the Fourier basis.

\begin{definition}[Discrete 2D convolutions] \label{def:2d_conv}
    Let $h(k_1, k_2)$ and $g(k_1, k_2)$ two discrete signals of finite extent, where $0 \leq k_1, k_2 \leq K - 1$.\footnote{In general the domains of support can be different in each index for each signal, but we keep the support domain square for simplicity.} Then the \emph{linear convolution} of $h, g$ is given by
        \begin{equation}
            (h \ast g)_{\text{lin}}(k_1, k_2)
            \equiv
            \sum_{j_1}\sum_{j_2} h(j_1, j_2) g(k_1 - j_1, k_2 - j_2).
        \end{equation}
    The region over which $j_1, j_2$ are summed is left unspecified here, though it can be bounded if $h, g$ have finite support; regardless, it can be shown that the region of support for the resulting convolution is $0 \leq k_1, k_2 \leq 2K - 2$. If we instead wanted to compute the \emph{circular convolution}, we simply treat the finite signals in a periodic way:
        \begin{equation}
            (h \ast g)_{\text{circ}}(k_1, k_2)
            \equiv
            \sum_{j_1 = 0}^{K - 1}\sum_{j_2 = 0}^{K - 1} h(j_1, j_2) g(k_1 - j_1\!\!\!\pmod{K}, k_2 - j_2\!\!\!\pmod{K}),
        \end{equation}
    which has support on $0 \leq k_1, k_2 \leq K - 1$. It is this \emph{circular convolution} that corresponds to element-wise multiplication of the 2D discrete Fourier transforms of our signals (see Lem.~\ref{lem:dis_2d_conv}). Finally, note that the linear convolution can be recovered using the circular convolution by padding the signals with zeros until they are of length $2K - 1$.
\end{definition}

\begin{lemma}[Discrete 2D convolution theorem] \label{lem:dis_2d_conv}
    Let $h(k_1, k_2)$ and $g(k_1, k_2)$ two discrete signals of finite extent, where $0 \leq k_1, k_2 \leq K - 1$. Then the 2D circular convolution of $h$ and $g$ can be written as the inverse Fourier transform applied to the element-wise product of the Fourier transforms $H, G$ of $h, g$:
        \begin{equation}
            (h \ast g)_{\text{circ}}(k_1, k_2)
            =
            \mathcal{F}^{-1} (H \circ G)
            =
            \mathcal{F}^{-1} (\mathcal{F}(h) \circ \mathcal{F}(g)),
        \end{equation}
    where the 2D Fourier transform is defined in the expected way,
        \begin{equation}
            H(k_1, k_2)
            \equiv
            \sum_{j_1 = 0}^{K - 1}\sum_{j_2 = 0}^{K - 1} h(j_1, j_2) \,e^{-2\pi i (j_1 k_1 + j_2 k_2)/K},
        \end{equation}
    and the element-wise product in the discrete 2D setting is the usual one (Def.~\ref{def:matrix_prod}). Note that while this definition (and the convolution in Def.~\ref{def:2d_conv}) are presented in the 2D setting, a similar relation extends simply to higher dimensions.

    In the 2D setting, however, we have the particularly nice relation that the 2D DFT of a matrix can be written as the composition of the row-wise and column-wise 1D DFTs of the matrix, i.e.,
        \begin{equation}
            \text{DFT}_{2D}(A)
            =
            F A F^{T}
            \implies
            \text{DFT}_{2D}(A)_{jm}
            =
            \sum_{k, \ell} F_{jk} A_{k \ell} F_{m \ell},
        \end{equation}
    where $F_{jk} = \omega^{jk}/\sqrt{K}$ and $\omega = e^{-2\pi i/K}$ is the usual DFT matrix. This implies that our 2D circular convolution can be written as a series of matrix and element-wise products:
        \begin{equation}
            (h \ast g)_{\text{circ}}/K
            =
            F^{\dagger}([F h F^{T}] \circ [F g F^{T}]) F^{*}
            =
            F^\dagger([F h F] \circ [F g F]) F^\dagger,
        \end{equation}
    where we've treated all among $F, g, h$ as $K$-dimensional matrices, and made the additional simplification given that the DFT matrix $F$ is both symmetric and unitary. Note crucially that the element-wise product and matrix products are not mutually associative.
\end{lemma}

Having defined the discrete 2D circular convolution, as well as the associated convolution theorem, a quantum algorithm for block encoding such convolutions is almost immediate from our previous constructions. Note also that if one wishes to take the convolution of $d$ total block encodings on $n$ qubits, our methods yield a quantum algorithm which uses only $\mathcal{O}(n\log{d})$ additional space, as it would for any combination of matrix and element-wise products of $\mathcal{O}(d)$ total block encodings.

\begin{theorem}[Block encoding the 2D convolution] \label{thm:be_2d_conv}
    Let $U_A$ and $U_B$ be $(\alpha, a, \varepsilon)$ and $(\beta, b, \delta)$ block encodings of square $n$-qubit operators $A, B$ respectively. Then we can construct a $(\alpha\beta, \max{(a, b) + 1, \alpha \delta + \beta \varepsilon + \varepsilon \delta})$ block encoding of the circular convolution of $A$ and $B$, denoted $(A \ast B)$ using a single query to each of $U_A$ and $U_B$, $\mathcal{O}(n^2\log{n})$ single and two-qubit gates, and $\mathcal{O}(\max(a,b))$ additional space.

    \begin{proof}
        Proof follows directly from constructing the quantum circuit achieving the block encoding of the product in Lem.~\ref{lem:dis_2d_conv}, which involves conjugating an element-wise product of two block encodings in the Fourier basis with a quantum Fourier transform (whose gate complexity goes as $\mathcal{O}(n^2 \log{n})$). The error propagation follows from the fact that the element-wise product is a principal submatrix of the Kronecker product, and that the quantum Fourier transform is unitary.
    \end{proof}
\end{theorem}

%%%%%%%%%%%%%%%%%%%%%%%%%%%%%%%%%%%%%%%%%%%%%
\subsection{The element-wise sign function} \label{sec:ew_sign}

\noindent Here we consider a concrete family of functions---approximations to the sign function; Def.~\ref{def:sign_func}---that have the special property that they can be written as the limit of repeated compositions of low-degree polynomials satisfying the conditions of Thm.~\ref{thm:qewt}. By composing QEWT protocols for these low-degree polynomials, we show that the subnormalization can be improved at the cost of a slight increase in query complexity. Later in Sec.~\ref{sec:lower_bounds} we show in fact that this particular class of functions also beats the generic subnormalization lower bounds for QEWTs (which are taken over all permissible functions of a given degree, and thus may not bound any specific subset of functions).

On the applications front, element-wise sign functions appear as a basic primitive from which other useful non-linear operations (e.g., step functions, band-pass functions, ReLU) can be immediately built, and also appear in isolation when quantizing machine learning models \cite{binary_neural_16}, various optimization subroutines \cite{sign_sgd_18}, and when studying spectral properties in random matrix theory \cite{opt_pca_18}.

For now we motivate and build up to the main Thm.~\ref{thm:ew_sign}, which constructs the element-wise sign function for block encodings. We cover relevant definitions, followed by two helper lemmas (Lem.~\ref{lem:ns_iter_1} and Lem.~\ref{lem:ns_iter_2}) on Newton-Schulz iteration, followed by the main construction. The purpose of these lemmas is to accurately capture the convergence behavior of Newton-Schulz iteration (largely following Ch.~5 of \cite{higham_mat_func_08}), which completely determines the query- and space-complexity of the resulting quantum algorithm.

\begin{definition}[Sign function] \label{def:sign_func}
    We define the sign function $\text{sgn}(x)$ as taking the value $+1$ for $x > 0$, $-1$ for $x < 0$, and $0$ for $x = 0$. 
\end{definition}

\begin{lemma}[Newton-Schulz iteration to 1/2] \label{lem:ns_iter_1}
    Let $x_0 = \delta \in (0, 1/2]$ and $x_j \equiv f(x_{j - 1})$ where $f(x) \equiv (3x - x^3)/2$. Then $x_k \geq 1/2$ for some $k = \mathcal{O}(\log_{3/2}(1/\delta))$.
    \begin{proof}
        First note that $f$ is monotone and increasing on $[0, 1]$, and that $f(x)/x$ (the multiplicative increase factor at each iteration) has the form $(3 - x^2)/2$. We can write this product at step $k$:
        \begin{equation}
            x_k = (3/2)^k \delta \prod_{j < k} (1 - x_{j}^2/3).
        \end{equation}
        We need to show that the product is bounded below by a constant in $k$ to yield the desired result; but this follows from the convergence of the sum of the $x_j^2$, which must decrease at least geometrically. This can be checked as our assumption that $x_{k-1} < 1/2$ implies (working backwards)
        \begin{equation} \label{eq:series_def}
            x_{k - n} < (2/3)^j(12/11)^j f(1/2).
        \end{equation}
        We can then bound the logarithm of the product from below by a constant, recalling first that $\log{(1- x)} \geq -x/(1-x)$ for $x \in [0, 1]$. Explicitly,
        \begin{equation}
            \sum_{j < k} x_j^2 \leq \left(\frac{11}{16}\right)^2 \sum_{\ell} \left(\frac{8}{11}\right)^{2\ell} = \frac{121}{228},
        \end{equation}
        where $11/16 = f(1/2)$ and $8/11 = (2/3)(12/11)$ derive from (\ref{eq:series_def}). We then apply our logarithm inequality for $x = x_j^2/3 \leq 1/12$:
        \begin{equation}
            \log{\prod_{j < k}(1 - x_j^2/3)} \geq -\frac{12}{11}\cdot\frac{1}{3}\cdot\frac{121}{228} = -\frac{11}{57},
        \end{equation}
        yielding a lower bound of $e^{-11/57} > 4/5$. Computing when $x_k \geq (4/5)(3/2)^k \delta$ exceeds $1/2$ immediately yields the stated result.
    \end{proof}
\end{lemma}

\begin{lemma}[Newton-Schulz iteration to $\varepsilon$] \label{lem:ns_iter_2}
    Let $x_k \geq 1/2$ and $x_{k + 1} = f(x_k)$ for $f(x) \equiv (3x - x^3)/2$. Then $|1 - x_{k + j}| \leq \varepsilon$ for $j = \mathcal{O}(\log{\log{(1/\varepsilon)}})$.
    \begin{proof}
        We take $y_{k} = 1 - x_{k}$, noting that $f(1 - y) = 1 - (3/2)y^2 + (1/2)y^3$, from which we see that $y_{k + 1} \leq (3/2) y_k^2$ for all $y \in [0, 1]$. Then taking $z_k = (3/2) y_k$ we see $z_{k + 1} \leq z_{k}^2$, from which we can quickly determine the required $j$ such that $z_{k + j}$ (assuming $y_{k} \leq 1/2$, equivalently $x_{k} \geq 1/2$) is less than $3\varepsilon/2$ (recalling our substitution, and desire to bound $x_{k + j}$ close to 1), i.e.,
        \begin{equation}
            (3/4)^{2^j} \leq 3\varepsilon/2 \implies j = \mathcal{O}(\log{\log{(1/\varepsilon)}}).
        \end{equation}
    \end{proof}
\end{lemma}

\begin{theorem}[Element-wise sign function] \label{thm:ew_sign}
    Let $U$ be an $(\alpha, a, 0)$ block encoding of an $n$ qubit operator $A$, and fix $0 < \varepsilon, \delta < 1$. Then there is a circuit preparing a $(\beta, b, 0)$ block encoding of a matrix $B$ satisfying
    \begin{equation}
        \lvert B_{ij} - \text{sign}(A_{ij})\rvert \leq \varepsilon,
        \quad 
        \forall (i, j) \text{ s.t. } |A_{ij}/\alpha| \geq \delta,
    \end{equation}
    which uses a number of queries $q$ to controlled-$U$ and controlled-$U^\dagger$ that satisfies
    \begin{equation}
        q = \mathcal{O}(\delta^{-\log_{3/2}{3}}\log^{\log_{2}{3}}{(1/\varepsilon)}),
    \end{equation}
    where all logarithms are natural unless otherwise notated. Moreover, the subnormalization $\beta$ and auxiliary space use $b$ satisfy
    \begin{align}
        \beta &= \mathcal{O}(\delta^{-\log_{3/2}2}\log{(1/\varepsilon)}),\\
        b & = \mathcal{O}(a + n\log{(1/\delta) + n\log{\log{(1/\varepsilon)}}}).
    \end{align}
    \begin{proof}
        Our argument makes use of well-known properties of Newton-Schulz methods (originally applied to computing matrix sign functions, explicated in Ch.~5 of the standard \cite{higham_mat_func_08}). Specifically, consider the function
        \begin{equation}
            f(x) \equiv (3x - x^3)/2,
        \end{equation}
        which evidently satisfies $\lVert f \rVert_{\infty} \leq 1$. $f(0) = 0$, and $\lVert f \rVert_{1} = 2$. We will study the repeated composition of $f$ with itself, taking $f_{0} = x$ and $f_{m} \equiv f \circ f_{m - 1}$. Observe that for $|x| \leq 1/2$ action of $f$ increases the input by a geometric factor at least $5/4$, meaning that to amplify an argument at $|x| \geq \delta$ to at least $1/2$ (chosen here greater than $1/3$ for reasons to be explained in the next step) takes a number of queries $q_1$ scaling as (by Lem.~\ref{lem:ns_iter_1})
        \begin{equation}
            q_1 = \mathcal{O}(\log_{3/2}{(1/\delta)}).
        \end{equation}
        Changing variables $x = 1 - y$ for convenience we can compute $f(x) = 1 - (3/2) y^2 + (1/2)y^3$, from which we can see that for $|x| > 1/2$ we have doubly exponential convergence in the number of iterations as $y \mapsto (3/2)y^2 - (1/2)y^3 < (3/2) y^2$, which requires only $|x| > 1/3$ as stated and yields (by Lem.~\ref{lem:ns_iter_2})
        \begin{equation}
            q_2 = \mathcal{O}(\log{\log}(1/\varepsilon)).
        \end{equation}
        The resulting query complexity and subnormalization from these two stages can be computed as $q = 3^{q_1 + q_2}$ and $\beta = 2^{q_1 + q_2}$, whose computation yields the stated parameters in the theorem statement. The auxiliary space use  $b$ follows from its logarithmic dependence on the query complexity.
    \end{proof}
\end{theorem}

We pause before moving onto the next section to highlight that Newton-Schulz iteration (originally developed for matrix functions) is limited to specific classes of functions, including the sign function, matrix inverse, and square root, with extensions to piecewise constant functions having also received some study \cite{higham_mat_func_08}. Generically, however, high-degree polynomials cannot be written exactly as a composition of polynomials of small degree (with the extremely limited exceptional cases enumerated by Ritt's theorems \cite{ritt_22}). We thus leave open the general (and most likely numerically-driven) question of finding good \emph{approximate} decompositions for wider classes of functions to future work.

%%%%%%%%%%%%%%%%%%%%%%%%%%%%%%%%%%%%%%%%%%%%%
\subsection{Non-standard block matrix products} \label{subsec:app_nonstandard_prod}

\noindent It turns out the constructions of Sec.~\ref{sec:main_const} yield quantum algorithms for computing other non-standard matrix products and functions. In this section we define, discuss, and build quantum algorithms for realizing various block-matrix generalizations to the Kronecker and element-wise products. These products appear in statistical physics, the analysis of linear systems, and control theory \cite{ts_prod_72, kr_prod_68}, and have been subjected to extensive study both through applied \cite{tk_tr_prod_app_89, hc_block_kronecker_89} and foundational matrix analytic lenses \cite{mn_mat_comm_79, liu_ts_kr_prod_99}.

\begin{definition}[Tracy--Singh product \cite{ts_prod_72}] \label{def:ts_prod}
    Let $A \in \mathbb{C}^{m\times n}$ and $B \in \mathbb{C}^{p\times q}$ where $A$ can be decomposed into $(m_i \times n_j)$-sized blocks, and $B$ can be decomposed into $(p_k \times q_\ell)$-sized blocks. Then the \emph{Tracy--Singh product} (or TS product) of $A$ and $B$, often denoted $(A \boxtimes B)$ has the form\footnote{Here what is meant by the nested indices is that the $ij$-th sub-block of the $(mp \times nq)$ matrix is itself divided into blocks, the $k\ell$-th of which is the Kronecker product $A_{ij} \otimes B_{k\ell}$.}
        \begin{equation}
            (A \boxtimes B) \equiv
            ((A_{ij} \otimes B_{k\ell})_{k\ell})_{ij},
        \end{equation}
    which should be interpreted as the pair-wise \emph{Kronecker product} according to the partitions of the two input block matrices. Note the shape of the resulting matrix is $(mp\times nq)$.
\end{definition}

\begin{proposition}[Relation between the Kronecker and TS product; adapted from \cite{ts_prod_lin_op_25}] \label{prop:kron_ts_prod}
    Let the \emph{commutation matrix} $K_{m,n}$ have the form (written suggestively in braket notation)
        \begin{equation} \label{eq:commutation_matrix}
            K_{m,n} = \sum_{i,j}^{m, n} |i\rangle\langle j| \otimes |j\rangle\langle i|.
        \end{equation}
    Then the TS product $(A \boxtimes B)$ between square $A$ and $B$ of size $n^2$ and $m^2$ respectively (with square sub-blocks of dimension $n$ and $m$ respectively) is related to the Kronecker product $(A \otimes B)$ by the commutation matrix:
        \begin{equation} \label{eq:kron_ts_prod}
            A \boxtimes B = (I_{n} \otimes K_{m,n}\otimes I_m)(A \otimes B)(I_{n} \otimes K_{n,m}\otimes I_m).
        \end{equation}
\end{proposition}

Note also that the commutation matrix is named as such given the related identity $K_{m^2, n^2}(A\otimes B) K_{n^2, m^2} = (B\otimes A)$. I.e., the commutation matrix in (\ref{eq:commutation_matrix}), interpreted as a quantum gate, is just the swap operation (taking the dimensions of $A, B$ equal). Evidently it is also its own inverse when $n = m$, and the Kronecker and TS products are related by a simple \emph{block-wise swap}, which can easily be computed quantum mechanically as in Thm.~\ref{thm:ts_prod_be}.

\begin{definition}[Khatri--Rao product \cite{kr_prod_68}] \label{def:kr_prod}
    Let $A \in \mathbb{C}^{m\times n}$ and $B \in \mathbb{C}^{p\times q}$ where $A$ can be decomposed into $(m_i \times n_j)$-sized blocks, and $B$ can be decomposed into $(p_k \times q_\ell)$-sized blocks, where $i, k$ and $j, \ell$ range over the same sets respectively. Then the \emph{Khatri--Rao product} (or KR product) of $A$ and $B$, denoted $(A \odot B)$, has the form
        \begin{equation}
            (A \odot B) \equiv
            (A_{ij} \otimes B_{ij})_{ij},
        \end{equation}
    which should be interpreted as a block-wise (as in between blocks indexed by the same index pair $(i,j)$) Kronecker product between the two block matrices. Unlike the TS product (Def.~\ref{def:ts_prod}), Kronecker products between sub-blocks of $A$ and $B$ do not cover all pairs of index pairs $(ij, k\ell)$, but only one set of indices $(ij)$. The shape of the resulting matrix is $(m_i p_i) \times (n_j q_j)$, where sums over $i, j$ are implicitly taken \emph{inside} their encompassing parentheses.
\end{definition}

In the case that $A, B$ in the above products have only trivial (i.e., scalar) sub-blocks, the TS and KR products reduce to the usual Kronecker and element-wise products. Moreover, the TS product shares many desirable properties with the Kronecker product, including associativity, distribution over sums, and that $(A \boxtimes B)^{-1} = A^{-1}\boxtimes B^{-1}$. In what follows we make the correspondence between these pairs of products more concrete by giving quantum algorithms to compute the TS and KR products between block encoded matrices. For simplicity we work in a qubit-adapted setting, though standard methods in matrix padding (see discussion in \cite{gslw_19}) can be used to achieve these non-standard products between block matrices of arbitrary (conformable) shape and dimension.

\begin{example}[TS and KR products; minimal examples] \label{ex:ts_kr_prod}
    For the visual reader we include brief, explicit examples of the Tracy-Singh (Def.~\ref{def:ts_prod}) and Khatri--Rao (Def.~\ref{def:kr_prod}) products between $(2\times 2)$ block matrices. Note that defining these products in principle requires no additional assumptions on the block sizes beyond their mutual compatibility to constitute a rectangular $M_1, M_2$, defined below:
    \begin{equation}
        M_1 =
        \left[
        \begin{array}{c|c}
            A_{00} & A_{01}\\\hline
            A_{10} & A_{11}
        \end{array}
        \right],
        \quad\quad 
        M_2 =
        \left[
        \begin{array}{c|c}
            B_{00} & B_{01}\\\hline
            B_{10} & B_{11}
        \end{array}
        \right].
    \end{equation}
    In this case the $(4\times 4)$ block matrix TS product can be expressed as the following:
    \begin{equation}
        (M_1 \boxtimes M_2) \equiv 
        \left[
        \begin{array}{cc|cc}
            A_{00}\otimes B_{00} & A_{00}\otimes B_{01} & A_{01}\otimes B_{00} & A_{01}\otimes B_{01}\\
            A_{00}\otimes B_{10} & A_{00}\otimes B_{11} & A_{01}\otimes B_{10} & A_{01}\otimes B_{11}\\\hline
            A_{10}\otimes B_{00} & A_{10}\otimes B_{01} & A_{11}\otimes B_{00} & A_{11}\otimes B_{01}\\
            A_{10}\otimes B_{10} & A_{10}\otimes B_{11} & A_{11}\otimes B_{10} & A_{11}\otimes B_{11}
        \end{array}
        \right],
    \end{equation}
    where we see that the larger $2\times 2$ blocks (each containing four sub-blocks) are each indexed by the same block $A_{ij}$ entering into the left argument of a Kronecker product. Just as in the usual Kronecker product, we can see that the four corner blocks of this block matrix contain Kronecker products between sub-blocks of $M_1$ and $M_2$ with the same index. Permuting these blocks and post-selecting on the resulting top-left submatrix yields the KR product, summarized (dropping indices for now) below:
    \begin{align}
        (M_1 \odot M_2) &\equiv 
        (\langle 0 |\otimes I)
        P
        \left[
        \begin{array}{cc|cc}
            A_{00}\otimes B_{00} & \ast & \ast & A_{01}\otimes B_{01}\\
            \ast & \ast & \ast & \ast\\\hline
            \ast & \ast & \ast & \ast\\
            A_{10}\otimes B_{10} & \ast & \ast & A_{11}\otimes B_{11}
        \end{array}
        \right]
        P^\dagger
        (|0\rangle \otimes I)
        ,\\
        &= 
        \left[
        \begin{array}{c|c}
            A_{00}\otimes B_{00} & A_{01}\otimes B_{01}\\\hline
            A_{10}\otimes B_{10} & A_{11}\otimes B_{11}
        \end{array}
        \right].
    \end{align}
    From this it can be seen, at least at a moral level, that the quantum algorithms to compute these products involve applying our previously constructed algorithms to sub-registers of block-encoded matrices with block structure.
\end{example}

\begin{theorem}[Tracy--Singh product between block encodings] \label{thm:ts_prod_be}
    Let $U_A$ and $U_B$ be $(\alpha, a, 0)$ and $(\beta, b, 0)$ block encodings of $(n + m)$-qubit operators $A$ and $B$ respectively. Moreover, consider $A$ and  $B$ as each divided into $2^n$ square blocks of size $2^m$, and that matrix elements within these blocks are indexed by computational basis states for qubits over disjoint registers, e.g., $A_{r_1 r_2}$ where $r_1$ (labeling the blocks) has size $n$, while $r_2$ has size $m$. Then we can construct a unitary $V$ that is a $(\alpha\beta, a + b, 0)$-block encoding of the Tracy--Singh product $(A\boxtimes B)$ using one query to each of $U_A$ and $U_B$ and $2\max{(m, n)}$ two-qubit swap gates. Here $(A\boxtimes B)$ is a square matrix of dimension $2^{2(n + m)}$.

    \begin{proof}
        The construction in this qubit setting follows from the identity in Prop.~\ref{prop:kron_ts_prod}, as the TS product is not only a block-wise analogue of the Kronecker product, but also its row- and column-permutation. I.e., ignoring the unitaries $U_A$ and $U_B$ for the moment, the TS product satisfies, modifying the notation of (\ref{eq:kron_ts_prod}):
        \begin{equation}
            (A\boxtimes B)_{r_1 r_2 r_3 r_4}
            =
            S_{r_2 r_3} (A_{r_1 r_2} \otimes B_{r_3 r_4}) S_{r_2 r_3}.
        \end{equation}
        As the Kronecker product of two block encoded operators is a simple construction (and indeed is a subroutine in our element-wise product circuit, using only the union of the auxiliary space required by each input), the only addition is the conjugating qubit-wise swap gates on the $m$- and $n$-qubit registers $r_2$ and $r_3$ (here denoted $S$ with the registers swapped in the subscript). Note that here the swap is performed in-place between registers of different sizes, $n, m$, and thus requires $\max{(m, n)}$ two-qubit swap gates. Finally, as this product is just a permutation of the Kronecker product, error propagation is identical tho that case, yielding an $(\alpha\beta, a + b, \alpha\delta + \beta\varepsilon + \delta\varepsilon)$ block encoding of $(A \boxtimes B)$ in the case that the input block encodings have error $\varepsilon$ and $\delta$ respectively.
    \end{proof}
\end{theorem}

\begin{theorem}[Khatri--Rao product between block encodings] \label{thm:kr_prod_be}
    Let $U_A$ and $U_B$ as in Thm.~\ref{thm:ts_prod_be}. Then we can construct a unitary $V$ that is a $(\alpha\beta, a + b, 0)$-block encoding of the Khatri--Rao product $(A \odot B)$ using one query to each of $U_A$ and $U_B$ and $\max{(m, n)}$ swap gates. Here $(A \odot B)$ is a square matrix of dimension $2^{(n + m)}$ (the same as $A$ and $B$).

    \begin{proof}
        Proof follows from the observation that the KR product is a submatrix of the TS product analogously to how the element-wise product is a submatrix of the Kronecker product. I.e., we conjugate by the same permutation employed in our original element-wise construction, save here only on the block-labeling register of the first block encoding.
        \begin{equation}
            (A\odot B)_{r_3 r_4}
            =
            (\langle 0|_{r_1 r_2}\otimes I_{r_3 r_4})
            P_{r_1 r_3}
            (A \boxtimes B)_{r_1 r_2 r_3 r_4}
            P^{\dagger}_{r_1 r_3}
            (|0\rangle_{r_1 r_2}\otimes I_{r_3 r_4}).
        \end{equation}
        As in our construction of the TS product, error propagation here for block encodings of $A$ and $B$ with error $\varepsilon$ and $\delta$ respectively is exactly analogous to the standard element-wise product, i.e., the result is an $(\alpha\beta, a + b, \alpha\delta + \beta\varepsilon + \delta\varepsilon)$ block encoding of $(A \odot B)$.
    \end{proof}
\end{theorem}

%%%%%%%%%%%%%%%%%%%%%%%%%%%%%%%%%%%%%%%%%%%%%
\section{Discussion} \label{sec:discussion}

\noindent In this work we construct a family of quantum algorithms for computing element-wise functions of block encodings, whose space use scales logarithmically rather than linearly in the degree of the applied function---this scaling, matching that of existing methods for applying spectral maps to or taking linear combinations of block encodings, thus represents a new and practically-implementable subroutine in the block encoding toolkit. Moreover, we raise and rectify errors in previous constructions (Appx.~\ref{appx:lcu_nonstandard}), establishing a robust foundation for the general theory of \emph{quantum element-wise transforms} (QEWTs).

The core of our construction depends on the orchestration of both well-known and (to the authors' knowledge) new but simple subroutines for manipulating block encodings. These include (1) a known identity relating the element-wise product of matrices to a principal submatrix of their Kronecker product, (2) a known construction for exponentially reducing the auxiliary space required when \emph{matrix multiplying} many block encodings at once (the so-called \emph{compression gadget} of Low and Wiebe \cite{lw_comp_gadget_19, vg_reducing_space_25}), and (3) a method for cheaply duplicating blocks of block encodings \emph{which themselves have block structure} (which we call a \emph{swap-copy}, Lem.~\ref{lem:swap_copy}) \emph{as well as} a way to efficiently re-use the auxiliary space required by this swap-copy within recursively self-composed versions of our protocol, enabling space-efficient computation of high-degree element-wise products (our \emph{weaving lemma}, Lem.~\ref{lem:rec_swap_copy}). These, together with our corrected scheme for constructing LCU \textsc{select} unitaries (Lem.~\ref{lem:ew_select}), and application of various amplitude amplification\slash uniform singular value amplification techniques (Sec.~\ref{sec:alt_const}), compose our central contribution. Crucially, beyond merely introducing (1-3) above, we show that these techniques can be made to interact while maintaining both correctness and $\mathcal{O}(\log{d})$ auxiliary space.

Beyond introducing a class of improved quantum algorithms, we devote Sec.~\ref{sec:applications} to specific applications across machine learning inference, multidimensional signal processing, and the computation of generalized block-matrix products. These applications are partly direct improvements to existing constructions \cite{gychanar_had_prod_25}, as well as strict generalizations \cite{rr_nonlinear_23, rhgpr_acc_inf_25, mw_gqsp_24}, and novel quantum algorithms for existing, classically useful matrix\slash tensor functions \cite{ts_prod_72, kr_prod_68}. In addition to applications, however, a tandem aim of this work is the development of quantum algorithms for core numerical linear algebraic subroutines that can be \emph{efficiently and interpretably combined}, as well as modified to (1) more easily incorporate classical techniques and (2) admit the use of numerical tools for compilation\slash verification\slash optimization long common in classical computing.

%%%%%%%%%%%%%%%%%%%%%%%%%%%%%%%%%%%%%%%%%%%%%
\subsection{On space, query, and subnormalization lower bounds} \label{sec:lower_bounds}

\noindent While the primary focus of this work is showing large space improvements and explicit quantum circuits for computing element-wise functions of block encodings, it is worth investigating the ultimate performance limitations of any quantum algorithm for this task. The element-wise product, given its strange character, does not immediately permit circuit depth, space, and query lower bounds already known for matrix products. Nevertheless in this section we show we can establish numerous non-trivial bounds, specifically on the query- and space-complexity, as well as the subnormalization of the resulting transformed block encodings. In the two items below we establish the expected query- and space- complexity lower bounds, followed by a more in-depth discussion of the more theoretically interesting question of block encoding subnormalization.

\begin{enumerate}
    \item Most simply the element-wise product and matrix product coincide for diagonal block encodings, which establishes an $\Omega{(d)}$ query complexity lower bound, as well as (in the multivariable setting) an $\Omega{(\log{d})}$ \emph{space} lower bound (following from the space lower bound of exact compression gadgets shown in \cite{vg_reducing_space_25}).

    \item Secondly, the linear-space protocol identifies the $d$-th element-wise product as a submatrix of the Kronecker product $A^{\otimes d}$, and multiplying block encodings of a generic $A$ requires linear depth in the degree of the product \cite{gslw_19}, following from generic no-fast forwarding results. Additionally, the non-space-optimized circuit including permutations described in prior work \cite{gychanar_had_prod_25} shares this linear depth.
\end{enumerate}

{
% For providing some breathing room in tables; applying only locally
\renewcommand{\arraystretch}{1.25}
\begin{table}
    \centering
    \begin{tabular}{l | l l l l l l}
    & Prior work & Lower bound & QEWT & Sign function & QSVT & LCU \\\hline
    Space & $\mathcal{O}(nd)$ & $\Omega(\log{d})$ & $\mathcal{O}(n\log{d})$ & $\mathcal{O}(n\log{d})$ & $\mathcal{O}(1)$ & $\mathcal{O}(\log{d})$ \\
    Query & $d$ & $d$ & $d$ & $d$ & $d$ & $\mathcal{O}(d)$ \\
    Depth & $\mathcal{O}(d)$ & $\Omega(d)$ & $\mathcal{O}(d)$ & $\mathcal{O}(d)$ & $\mathcal{O}(d)$ & $\mathcal{O}(d)$\\
    Subnorm\; & $\lVert f \rVert_{1}$ & $\Omega(\min{(d^2, N)})$ & $\lVert f \rVert_{1}$ & $\mathcal{O}(d^{\log_{3/2}{2}}) \approx \mathcal{O}(d^{1.71})$ & $\lVert f \rVert_{\infty}$ & $\lVert f \rVert_{1}$
    \end{tabular}
    \caption{
        A summary of the space use, query complexity, circuit depth, and block encoding subnormalization achieved by prior work \cite{gychanar_had_prod_25} (which did not explicitly prove lower bounds), our new lower bounds (taken over all permissible $f$ up to degree $d$), and our main QEWT (Thm.~\ref{thm:qewt}), where $n = \log{N}$ and $d$ are per the main Tab.~\ref{tab:resource_comp}, and $\lVert f\rVert_1$ is the standard LCU coefficient one-norm \cite{cw_lcu_12}. For \emph{specific} useful classes of functions, e.g., the element-wise sign function (Sec.~\ref{sec:ew_sign}) we can improve the subnormalization below our generic lower bound. For comparison we also give the resource complexities of QSVT and LCU for the same function applied to the singular values.
    } \label{tab:lower_bounds}
\end{table}
}

We now turn to discussion of the block encoding subnormalization for QEWTs. One strength of block encoding quantum algorithms like QSVT over linear combination of unitaries (LCU) is that the former, when applying a polynomial transformation to the input block encoding, accrues a subnormalization bounded above by the infinity norm, $L^\infty$, of the applied function,\footnote{In all cases we will consider this norm over the interval $[-1, 1]$.} while LCU is instead subnormalized by the $\ell_1$ norm \emph{of the coefficients} of the applied function. As block encodings must be normalized by a scalar $\alpha$ (Def.~\ref{def:be_formal}) bounded below by this subnormalization, and the probability of successfully applying a desired block encoding is proportional to this subnormalization, minimizing $\alpha$ is of significant practical importance. For common problems, e.g., Hamiltonian simulation, the difference between $\lVert f\rVert_{L^\infty}$ (henceforth denoted $\lVert f \rVert_{\infty}$) and $\lVert f \rVert_{\ell_1}$ (henceforth denoted $\lVert f \rVert_{1}$) can be unbounded (and practically large) in the degree of the applied function. Consequently it is natural to ask whether our QEWT, which in part depends on an LCU subroutine to combine element-wise powers, can also have its subnormalization improved from the $\ell_1$-dependent LCU limit to the $L^{\infty}$-dependent QSVT limit. We show that the answer to this question is negative in general, but that there exists both (1) a substantial gap between our LCU-based QEWT subnormalization for general functions and the proved lower bound, and (2) this general lower bound can be improved for certain useful, explicit families of functions.

In what follows we provide a few definitions relating to the QEWT-induced subnormalization. Our aim will be to show that, for valid block encoding inputs, element-wise functions of degree $d$ can only increase the spectral norm by a factor at worst quadratic in $d$.

\begin{definition}[QEWT subnormalization] \label{def:qewt_subnorm}
    Let $f$ satisfying $\lVert f \rVert_{\infty} \leq 1$ over $[-1,1]$, $f(0) = 0$, and $\text{deg}(f) \leq d$. We define the QEWT subnormalization $\Gamma(f, N)$:
        \begin{equation}
            \Gamma(f, N) \equiv 
            \sup_{\lVert A \rVert \leq 1} \lVert f^{\circ}(A) \rVert,
            \quad 
            \Gamma(d, N) \equiv 
            \sup_{\text{deg}(f) \leq d} \Gamma(f, N),
        \end{equation}
    where $A \in \mathbb{C}^{N\times N}$ is a matrix.
\end{definition}

Our primary interest is the behavior of the generic subnormalization $\Gamma(d, N)$, which represents the worst-case growth of the spectral norm of valid block encodable matrices under element-wise transforms of bounded degree. Our basic strategy follows from well-known relationships between the maximum modulus of a polynomial and bounds on its derivatives (so-called Bernstein-type theorems). In what follows we state the main result (Thm.~\ref{thm:qewt_subnorm_bounds}), followed by a proof of the asymptotic upper and lower bounds composing the statement.

\begin{theorem}[Bounds on the QEWT subnormalization] \label{thm:qewt_subnorm_bounds}
    Let $\Gamma(d, N)$ as in Def.~\ref{def:qewt_subnorm}. Then $\Gamma(d, N) = \Theta(\min{(d^2, N)})$.
    \begin{proof}
        We show an upper bound on $\Gamma(d, N)$ in the following Lem.~\ref{lem:qewt_subnorm_upper}, and then provide an explicit input and function pair achieving this bound in Prop.~\ref{prop:qewt_subnorm_lower}. Together these establish the $\Theta(\min{(d^2, N)})$ character of $\Gamma(d, N)$.
    \end{proof}
\end{theorem}

\begin{lemma} \label{lem:qewt_subnorm_upper}
    Let $\Gamma(d, N)$ as in Def.~\ref{def:qewt_subnorm}. Then $\Gamma(d, N) = \mathcal{O}(\min{(d^2, N)})$. Equivalently, for all $A$ s.t. $\lVert A \rVert \leq 1$ and for all $f$ s.t. $\lVert f \rVert_{\infty} \leq 1$, there exists a constant $K$ independent of $d, N$ s.t.
    \begin{equation}
        \lVert f^{\circ }(A)\rVert \leq K d^2.
    \end{equation}
    \begin{proof}
        The aim of this proof is to eventually apply the Schur test to our element-wise transformed matrix: bounding its spectral norm in terms of bounds on $\ell_1$ norms of its rows and columns. We define the \emph{quadratic modulus} $R(f; \xi)$ as the magnitude of the error between our function and its linear approximation up to finite distance from the origin
        \begin{equation}
            R(f; \xi) =
            \sup_{0 \leq x \leq \xi} \frac{\lvert f(x) - x f'(0)\rvert}{x^2}.
        \end{equation}
        We first claim that for any $\lVert A \rVert \leq 1$, and any $f$ s.t. $f(0) = 0$, that the following holds:
        \begin{equation} \label{eq:spec_norm_bound}
            \lVert f^{\circ}(A) \rVert \leq 
            \lvert f'(0)\rvert \lVert A\rVert 
            + 
            R(f; \xi) \lVert A \rVert^2.
        \end{equation}
        The proof of this is straightforward, first splitting $f = f'(0) t + g(t)$, where $g(0) = g'(0) = 0$ by design. The first term in (\ref{eq:spec_norm_bound}) appears immediately applying a triangle inequality. For the second term it must be the case that every column has $\ell_2$ norm bounded above by $\lVert A \rVert$, and thus we can bound the $\ell_1$ norm of the columns in terms of our quadratic modulus:
        \begin{equation}
            \forall j,\quad \sum_{k = 0}^{N - 1} |g(A_{jk})| 
            \leq 
            R(f; \xi) \sum_{k} |A_{jk}|^2
            \leq 
            R(f; \xi) \lVert A \rVert^2.
        \end{equation}
        The same argument goes through identically for the rows, meaning we can use the Schur test which states that matrices with upper bounds $C_r$ and $C_l$ on the $\ell_1$ norms of their rows and columns have spectral norm bounded by $\sqrt{C_r C_l}$, and thus $\lVert g^{\circ}(A) \rVert \leq R(f; \xi) \lVert A \rVert^2$. Putting these together we get
            \begin{equation}
                \lVert f^{\circ}(A) \rVert \leq f'(0) \lVert A \rVert + R(f; \xi) \lVert A \rVert^2.
            \end{equation}
        We now provide $d$-dependent bounds on $f'(0)$ and $R(f; \xi)$, and will eventually take $\xi = 1$. The first of these is a direct consequence of Bernstein's theorem, which states that $\lVert f \rVert_{\infty} \leq 1$ on $[-1, 1]$ implies $f'(0) \leq d$. For the second derivative, we can iterate Bernstein's theorem to show $|f''(x)| \leq 3d^2$ for $|x| \leq 1/2$. Concretely this follows from the fact that a polynomial of degree $d$ which is bounded in modulus by $1$ on $[-1,1]$ must satisfy
        \begin{equation}
            |f'(x)| \leq \frac{d}{\sqrt{1 - x^2}}, \quad x \in (-1,1),
        \end{equation}
        which for choice of $|x| \leq 1/2$ allows us to bound $|f'(x)| = \mathcal{O}(d)$ on a (constant-size) sub-interval of $[-1,1]$. This in turn allows us to bound the modulus of the second derivative of $f$ on a slightly smaller (constant-size) sub interval by  $\mathcal{O}(d^2)$, again by Bernstein's theorem (see use in Ch.~5 of the standard textbook \cite{be_polynomial_ineq_95}).
        Bounding the Taylor series error then gives us
            \begin{equation}
                |g(x)| \leq 
                \frac{1}{2}x^2\sup_{|y| \leq |x|} |f''(y)| 
                \leq 
                \frac{3}{2} d^2 x^2,\quad 
                |x| \leq \frac{1}{2}.
            \end{equation}
        For the remaining $1/2 \leq |x| \leq 1$, we simply apply the triangle inequality to $|g(x)| = |f(x) - f'(0) x|$, which must be bounded by $1 + d$ on $[-1,1]$, meaning that $R(f; 1)$ must be bounded above by $4(1 + d)$ for $1/2 \leq |x| \leq 1$. Together these establish the desired $\lVert f^\circ(A)\lVert = \mathcal{O}(d^2)$. Finally note that the upper bound of $N$ is immediately established by $\lVert f^{\circ}(A) \rVert \leq \lVert f^{\circ}(A) \rVert_{F} \leq N$, where the latter is the Frobenius norm.
    \end{proof}
\end{lemma}

In the following proposition we show that the upper bound given above for the growth of the spectral norm cannot in general be improved upon, i.e., that there exist $A$ and $f$ such that the spectral norm of $f^{\circ}(A)$ grows with the square of $d$. This establishes that $\Gamma(d, N) = \Theta(\min{(d^2, N)})$, rather than just $\mathcal{O}(\min{(d^2, N)})$.

\begin{proposition} \label{prop:qewt_subnorm_lower}
    Let $A = H_N/\sqrt{N}$ (the generalized Hadamard matrix of dimension $N$ where every element is $\pm 1/\sqrt{N}$, which has spectral norm $1$). Moreover fix $d^2 = N$, and $f(x) = [T_{d}(x) - T_d(0)]/2$ for $d \equiv 0 \pmod{4}$ where $T_d$ is the $d$-th Chebyshev polynomial. Note that $f(x)$ is even, satisfies $f(0) = 0$, and is bounded in magnitude by $1$ on the interval $[-1, 1]$. Then $f^{\circ}(A)$ has spectral norm
        \begin{equation}
            \lVert f^{\circ}(A) \rVert
            =
            d^2 f(1/d)
            \approx 
            d^2 (1 - \cos{1})/2
            =
            d^2 \sin^2{(1/2)}.
        \end{equation}
\end{proposition}

Finally we show that for our main construction, Thm.~\ref{thm:qewt}, the subnormalization can be generically much worse (in $d$) than the lower bound proven in Thm.~\ref{thm:qewt_subnorm_bounds}; one way of interpreting these results is the ill conditioned nature of the standard monomial basis for approximation on the interval (a classically well-known fact \cite{tref_approx_19}), leaving open intriguing questions about efficient realization of, e.g., element-wise Chebyshev polynomials.

\begin{proposition}[Maximal subnormalization for LCU-based QEWT]
    Let $f$ s.t. $\lVert f \rVert_{\infty} \leq 1$, and $\beta$ the subnormalization achieved by the LCU-based QEWT construction in Thm.~\ref{thm:qewt}, i.e., $\lVert f \rVert_{1}$. Then
    \begin{equation}
        1 \leq \frac{\lVert f \rVert_1}{\Gamma(f, N)} \leq (1 + \sqrt{2})^d,
    \end{equation}
    where this upper bound (up to $d^2$ subnormalization) is achieved for $(T_{d}(x) - 1)/2$.
    \begin{proof}
        The $\ell_1$ norm of the monomial coefficients of Chebyshev polynomials is well-known (simply evaluate them at $z = i$ in the complex plane), and consequently we only need to show that this is the maximal case. But this is also a well-known statement, namely that $|c_k| \leq |T_{d, k}|$ \emph{for any} function $f$ of degree less than $d$ such that $\lVert f \rVert_{\infty} \leq 1$, where $c_k$ is the $k$-th monomial coefficient and $T_{d, k}$ is the $k$-th monomial coefficient of the $d$-th Chebyshev polynomial (see again Ch.~5 of the standard \cite{be_polynomial_ineq_95}).
    \end{proof}
\end{proposition}

The key takeaway from Thm.~\ref{thm:qewt_subnorm_bounds} is that the hope for subnormalization depending only on $\lVert f \rVert_{\infty}$, as is the case in QSVT-based spectral maps, is not possible, but that there is significant (i.e., at worst exponential in $d$) room between the $\lVert f \rVert_{1}$ subnormalization of our main Thm.~\ref{thm:qewt} and the $\Theta(d^2)$ lower bound derived above. While a generic scheme for constructing a QEWT protocol generically matching this subnormalization lower bound is beyond the scope of this work, we devote Sec.~\ref{sec:ew_sign} to a specific family of element-wise sign functions that achieve \emph{subquadratic}\footnote{Here we beat the lower bound because we've specified a specific degree-$d$ function, while $\Gamma(d, N)$ is a supremum over all such functions.} subnormalization in $d$ and slightly super-linear query complexity as compared with a spectral transform according to the same sign function approximation. This improvement follows from a special representation of the sign function as a repeated composition of low-degree polynomials (the so-called Newton-Schulz iterate: expanded on in Sec.~\ref{sec:ew_sign}), suggesting numerous interesting open questions regarding hybrid block encoding techniques for the maximally efficient realization of various (restricted but useful) matrix function families.

%%%%%%%%%%%%%%%%%%%%%%%%%%%%%%%%%%%%%%%%%%%%%
\subsection{On dequantizing QEWTs} \label{sec:dequant}

\noindent A secondary benefit of the development of quantum algorithms based around block encodings has been the refinement of conditions under which one can \emph{dequantize} a quantum computation \cite{chia_low_rank_22, gl_dequant_23}. I.e., there can exist combinations of input state, block encoding, transformation, measurement, and required precision whose corresponding quantum algorithm can be emulated \emph{classically and efficiently} in the defining parameters. In this way the unifying character of block encoding algorithms lifts to a flexible suite of classical tools delimiting necessary conditions on the existence of large quantum-classical separations.

In this section we investigate the question of dequantizing QEWTs, which by merit of the unique properties of element-wise functions, act quite differently than previous block encoding primitives with respect to standard dequantization tools. The results of this investigation can be summarized in three related statements:
\begin{enumerate}[label=(\alph*)]
    \item Thm.~\ref{thm:qewt_sq_access}. A single layer of QEWT preserves \emph{oversampling and query access} (Def.~\ref{def:osq_access}). In other words, the ability to efficiently classically emulate a block encoding implies the ability to efficiently classically emulate its QEWT. This is expected, as the function can be applied post-query by standard techniques.

    \item Prop.~\ref{prop:qewt_stable_rank}. The action of a QEWT generically increases the \emph{stable rank} (Def.~\ref{def:stab_rank}) of a block encoding \emph{exponentially} in the function degree. 

    \item Thm.~\ref{thm:comp_qewt_qsvt_sq}. The composition of a single layer of QEWT followed by a single layer of QSVT generically destroys \emph{oversampling and query access}. Note that \emph{either of} the QEWT or QSVT alone preserves \emph{oversampling and query access}, and thus this failure is unique to the interoperation of element-wise and spectral transformations. 
    Consequently, existing standard methods \cite{chia_low_rank_22} are insufficient to dequantize these compositions.
\end{enumerate}

We emphasize that this section does not establish that it is impossible to dequantize QEWT/QSVT compositions, only that the standard dequantization techniques sufficient for all previously known quantum algorithms based on block encodings fail.
It is thus an interesting open question whether one can show not only that QEWT/QSVT compositions break oversampling and query access, but that there exist more fundamental reasons for the hardness of their classical emulation.
That is, we do not investigate separations of the form of \cite{gl_dequant_23}, which establish classical hardness for emulating QSVT in certain high-precision regimes by tethering block encoding algorithms to the BQP-complete guided local Hamiltonian problem. 

In the remainder of this section we provide necessary definitions and basic theorems in dequantizing block encoding algorithms, followed by their modification and application to QEWTs. We start by recalling \emph{sample and query} access (Def.~\ref{def:sq_access}), introduced in \cite{tang_19} to investigate dequantization, and extended to \emph{oversampling and query access} (Def.~\ref{def:osq_access}) to capture attributes unique to the classical emulation of block encodings.

\begin{definition}[Sample and query access; from \cite{chia_low_rank_22}] \label{def:sq_access}
    For $v \in \mathbb{C}^N$ we say we have sample and query (SQ) access to $v$ if we can (1) query any individual entry $v_i$, (2) sample an index $i$ according to probability $|v_i|^2/\lVert v \rVert^2$, and (3) query $\lVert v \rVert$ (the $\ell_2$ norm). For a matrix $A$, SQ access to $A$ means (1) SQ access to every row of $A$, as well as (2) SQ access to the vector of row norms of $A$.
\end{definition}

\begin{definition}[Oversampling and query access; from \cite{chia_low_rank_22}] \label{def:osq_access}
    For $v \in \mathbb{C}^{N}$ we say we have $\phi$-\emph{oversampling and query access} $\text{SQ}_{\phi}(v)$ if we have (1) query access to elements of $v$, and (2) we have SQ access to an `entry-wise upper bound' vector $v'$ satisfying $\lVert v' \rVert^2 = \phi \lVert v \rVert^2$ and $v_{i}' \geq |v_i|$ for all $i$. Extension to matrices is the same as in Def.~\ref{def:sq_access}.
\end{definition}

The purpose of both SQ and $\text{SQ}_{\phi}$ access is, in hindsight, that they capture conditions on \emph{classically available} data such that, if an efficient quantum algorithm exists given access to a quantum state $|v\rangle$ or a block encoding of $A$, then there is also an efficient classical algorithm with SQ access to $v$ or $\text{SQ}_{\phi}$ access to $A$. Moreover, many desirable schemes for loading classical data into quantum computations (like QRAM) yield this access, underlying the broad dequantization statements of \cite{chia_low_rank_22}.

Before showing that QEWTs preserve oversampling and query access, we first reiterate the analogous statement for QSVT from \cite{chia_low_rank_22}---in essence, their core result shows the ability of quantum-inspired classical algorithms to leverage SQ access to input matrices into $\text{SQ}_{\phi}$ access to complex arithmetic on these input matrices. We refer readers with interest to the excellent original reference (which goes far beyond dequantizing QSVT), and quote the theorem as an example of the common form of dequantization results.

\begin{theorem}[QSVT preserves $\text{SQ}_{\phi}$ access; modified Thm.~5.1 of \cite{chia_low_rank_22}] \label{thm:qsvt_sq_access}
    Let $A \in \mathbb{C}^{m\times n}$ satisfy $\lVert A \rVert_{F} = 1$ and be given by oracles $\text{SQ}(A)$, $\text{SQ}(A^{\dagger})$, with $\text{\textbf{sq}}(A)$ and $\text{\textbf{sq}}(A^{\dagger})$ of order $\log{(mn)}$. Assume also $\text{SQ}(b)$ for $b \in \mathbb{C}^{n}$ with $\lVert b \rVert = 1$ and $\text{\textbf{sq}}(b)$ of order $\log{n}$. Finally let $p(x)$ a degree $d$ polynomial of parity-$d$ such that $|p(x)| \leq 1$ for $x \in [-1,1]$. Then with probability greater than $1 - \delta$, and for $\varepsilon > 0$ sufficiently small, we have $\text{SQ}_\phi (v)$ for $\lVert v - p(A)b \rVert \leq \varepsilon\, \lVert p(A)b\rVert$ in time
    \begin{equation}
        t = \mathcal{O}\left(\text{poly}\Big[d,\frac{1}{\lVert p(A)b\rVert} \frac{1}{\varepsilon}, \frac{1}{\delta}, \log{(mn)}\Big]\right),
    \end{equation}
    where $p(A)$ is the matrix whose singular values have been transformed by $p$, and lower-case bold $\text{\textbf{sq}}(v)$ is the maximum of the \emph{costs} of sampling, querying, or norm querying $v$, and likewise for matrix arguments.
\end{theorem}

By a comparatively simpler argument, we can also show that bare QEWTs also preserve oversampling and query access (Thm.~\ref{thm:qewt_sq_access}). The key extension we show in the rest of this section is that the ways in which QSVT (and block encoding arithmetic generally) and QEWT preserve $\text{SQ}_{\phi}$ access are in general incompatible, meaning that their composition (even a single layer) cannot be efficiently classically emulated by the same techniques.

\begin{theorem}[QEWT preserves $\text{SQ}_{\phi}$ access] \label{thm:qewt_sq_access}
    Let $f$ satisfy $f(0) = 0$ and define $\sigma \equiv \max_{A_{ij} \neq 0} |f(A_{ij})/A_{ij}|$. Then $\text{SQ}(A)$ yields $\text{SQ}_{\phi}(f^{\circ}(A))$ where
    \begin{equation}
        \phi \equiv \sigma^2 \lVert A \rVert_{F}^2/\lVert f^{\circ}(A) \rVert_{F}^2,
    \end{equation}
    with constant additional cost per operation. This yields, by definition, true $\text{SQ}(f^{\circ}(A))$ in $\mathcal{O}(\phi)$ expected queries.
    \begin{proof}
        Note that $\sigma A$ majorizes $f^{\circ}(A)$ element-wise by definition, i.e., $|f(A_{ij})| \leq \sigma| A_{ij}|$. Consequently sampling from $\sigma A$ is just sampling from a constant multiple of $A$, its row norms are known in terms of those of $A$ and $\sigma$, and entries of $f^{\circ}(A)$ are computed by applying $f$ post-query. Finally, observe that this construction in no way depends on the rank or dimension.
    \end{proof}
\end{theorem}

Our argument proceeds by investigating the action of element-wise transforms on key properties of sketchable matrices, specifically centering on their approximability by low rank objects. We provide relevant definitions, followed by an investigation of the setting where QEWTs and other block encoding primitives are composed.

\begin{definition}[Stable rank] \label{def:stab_rank}
    We denote the stable rank of a matrix $A$ by $\rho(A)$, which is equal to $\lVert A \rVert_F^2/\lVert A \rVert^2$, the ratio of the (squared) Frobenius and spectral norms. The stable rank is bounded above by the rank.
\end{definition}

\begin{proposition}[Rank blowup under QEWTs] \label{prop:qewt_stable_rank}
    Let $A \in \mathbb{C}^{N\times N}$ have rank $r$ and $f$ a polynomial of degree $d$. Then the rank of $f^{\circ}(A)$ satisfies
    \begin{equation}
        \text{rank}(f^{\circ}(A)) \leq \sum_{k = 1}^{d}\binom{r + k - 1}{k} \leq \binom{r + d}{d} - 1,
    \end{equation}
    with equality generically when all polynomial coefficients are non-zero, and the rank does not exceed $N$. Moreover no bound on the stable rank $\rho(f^{\circ}(A))$ follows from $\rho(A)$ and $d$ alone.
    \begin{proof}
        Note that each entry of $f^{\circ}(A)$ is a polynomial of the entries of $A$ of degree $d$ by definition. If we write $A$ as a sum over $s \leq r$ rank one matrices $u_{s} v_{s}^{\dagger}$, then each $A^{k}_{ij}$ is a product of $k$ factors from up to $r$ terms, for which an upper bound is the quoted sum of binomial coefficients. Finally, note that the \emph{stable rank} need not necessarily blow up if the rank does, as a small number of directions may absorb most of the Frobenius mass (per Def.~\ref{def:stab_rank}).
    \end{proof}
\end{proposition}

In the next statement (the core result of this subsection) we show that nevertheless, the stable rank (and thus oversampling and query access) for compositions of the QEWT and QSVT (or block encoding arithmetic more generally) can be sufficiently destroyed, and thus prevent the techniques of \cite{chia_low_rank_22} from going through.

\begin{theorem}[Composing QSVT--QEWT breaks $\text{SQ}_{\phi}$ access] \label{thm:comp_qewt_qsvt_sq}
    Let $N = 2^m$ and take $A_{xy} = x\cdot y + 1/2$ for $x, y$ bit-strings of length $m$, and take $f$ to be the degree $m + 1$ polynomial with $f(0) = 0$ and $f(k + 1/2) = (-1)^{k}$ for $k \in \{0, 1, \cdots, m\}$. Then:
    \begin{enumerate}[label=(\alph*)]
        \item The rank of $A$ is $m + 1$, the stable rank of $A$ is $1 + \mathcal{O}(1/\log{N})$, and SQ(A) is efficiently (i.e., $\text{polylog}(N)$ time) implementable. Moreover, the matrix $f^{\circ}(A)$ is $H_N$, the Hadamard matrix, which has stable rank $N$.
        
        \item Let $M = H_{N}/\sqrt{N}$. There is an absolute constant $\varepsilon$ such that any classical procedure given $\text{SQ}(v)$ at a cost of $T$ queries per operation, that provides $\text{SQ}_{\phi}(w)$ where $\lVert w - M v \rVert \leq \varepsilon$ requires $\phi T = \tilde{\Omega}(\lVert M \rVert_{F}) = \Omega(\sqrt{N}/\log{N})$.
    \end{enumerate}
    \begin{proof}
        The first condition holds immediately from known properties of the Hadamard matrix, which results from the application of the defined $f$ to $A$ element-wise again by definition. For the second condition, take $u, v$ as $\sqrt{N}$-subnormalized bit strings of length $N$ and elements only $\pm 1$, in which case $\langle u | M | v\rangle$ is precisely the forrelation $\Phi(h, g)$. We can construct an unbiased estimator for the forrelation given $\text{SQ}_\phi(w)$, i.e., $\bar{u_i} w_i \lVert \tilde{w} \rVert^2/\lvert \tilde{w_i} \rvert^2$; this has second moment bounded by $\phi \lVert w\rVert^2$, and thus requires $\mathcal{O}(\phi^2/\eta^2)$ samples to to estimate to error $\eta \lVert w \rVert$, meaning we use $\mathcal{O}(\phi T/\eta^2)$ queries. It is known that $\Omega(\sqrt{N}/\log{N})$ classical queries are necessary for deciding forrelation \cite{aa_forrelation_18}, leading to the stated result.
    \end{proof}
\end{theorem}

At a base level, the results shown above indicate that, while both the QEWT and methods for block encoding arithmetic (including QSVT and LCU) can be dequantized under suitably strong \emph{classical} input assumptions, this dequantizability goes through by distinct arguments. Additionally, given the complex reciprocal action of spectral- and element-wise-functions on matrix rank and conditioning, the standard techniques of \cite{chia_low_rank_22} (in the study of oversampling and query access) do not survive the composition of these primitives. We leave as an open question whether there may be related extensions of sample and query access that recover dequantizability, noting that questions of subnormalization versus functional degree specific to QEWTs may permit matrix sketching techniques that were too expensive for the QSVT setting to be recovered. We also leave to future work the investigation of super-constant interleaved instances of QEWT and QSVT, which may be even harder to dequantize given the poor error propagation of standard sketching methods. As always, dequantization here requires \emph{classically provided} input data (e.g., in QRAM), and thus none of the arguments above are relevant if the provided block encoding is assumed to be the result of some quantum mechanical process.

%%%%%%%%%%%%%%%%%%%%%%%%%%%%%%%%%%%%%%%%%%%%%
\subsection{Open directions} \label{subsec:open_dir}

\noindent Our work suggests numerous near-term extensions, including the use of randomized\slash approximative techniques to further improve performance under architectural constraints, as is widely applied to other block encoding algorithms \cite{martyn_rall_halving_24, mrclc_parallel_qsp_24, rcc_modular_qsp_23}. We expect that additional knowledge of the input (e.g., low rank or well-conditioned) can also be leveraged for performance, and that variants of our swap-copy and weaving lemma can yield independently interesting methods for manipulating block encodings with block structure or known eigensubspaces. Here we briefly enumerate a few concrete directions and associated questions related to this work.
 
At a high level these are clustered around (a) direct performance improvements for the subroutines presented in the main text under reasonable assumptions on the input, (b) the extension of our methods to other common but under-studied (in a quantum computing context) matrix\slash tensor functions, (c) fundamental questions on resource lower bounds, and finally (d) the proposal of concrete applications with reasonable instance sizes.

\begin{enumerate}[label=(\alph*)]
    \item As discussed in the main text, various assumptions on the input block encoding (e.g., the diagonal case) can change the space cost of implementing QEWTs or their approximate versions. Promising assumptions include the near-diagonal setting (or their Fourier-space equivalents of matrices with displacement-structure, including Toeplitz, Hankel, and their block-variants). Moreover, there exist matrices whose block diagonalization (after which a swap-copy can be applied) can be accomplished in other ways than our masking procedure (e.g., block-triangular matrices in \cite{lw_mat_eq_25} in the context of solving linear matrix equations). How can QEWTs be combined with other block encoding algorithms using small amounts of additional space to achieve non-standard functions (e.g., Jordan block transformations \cite{ls_qep_26, gls_qet_arb_26})?
    
    \item While our construction works given black-box block encodings, there exist assumptions stronger than block encodings (e.g., sparse element-wise oracles) under which computing QEWTs becomes significantly easier. Do there exist less trivial access assumptions, stronger than black box access, which permit even more space-efficient QEWTs? Are any of these \emph{dequantizable} by existing matrix sketching methods for block encoding algorithms? \cite{chia_low_rank_22}
    
    \item While we can show space and query lower bounds for generic multivariable element-wise transforms (e.g., by reducing to the diagonal case where matrix multiplication and element-wise products coincide, which can require $\Omega(\log{d})$ auxiliary space \cite{vg_reducing_space_25}), this argument \emph{does not} apply to an element-wise function of a \emph{single} block encoding. How can we show less trivial space and query lower bounds in restricted settings?

    \item While the element-wise product appears ubiquitously in machine learning \cite{learning_had_prod_25}, control theory \cite{js_rga_86}, the study of the stability of dynamical systems, and classical signal processing \cite{hj_topics_91}, this work does not study the interaction between common conditioning\slash rank assumptions for these problems and the subnormalizations\slash success probabilities of the QEWT. Under what assumptions on the input block encoding should we expect, e.g., the polynomial advantages in \cite{gychanar_had_prod_25} to be preserved?
    
\end{enumerate}

It is the authors' aim that the techniques presented in this work provide not only direct utility but basic inspiration for building quantum algorithms for common matrix functions\slash tensorial manipulations---there is a huge variety of well-developed classical tools and intuition for such matrix functions, and a key component for honest comparison between classical and quantum algorithms is the ability to incorporate structural information about the input data (as we constantly leverage in the classical setting) to improve quantum algorithmic performance.

%%%%%%%%%%%%%%%%%%%%%%%%%%%%%%%%%%%%%%%%%%%%%
\section{Acknowledgments} \label{sec:ack}

\noindent The authors are grateful to John Martyn, Isaac Chuang, and Mark Wilde for discussions on the topic of this work. ZMR acknowledges funding from the Japan Society for the Promotion of Science (JSPS) Postdoctoral Fellowship for Research in Japan (KAKENHI 24KF0136). RS was supported by the U.S. Department of Energy, Office of Science, Accelerated Research in Quantum Computing Centers, Quantum Utility through Advanced Computational Quantum Algorithms, grant no. DE-SC0025572. Part of this work was done while ZMR was visiting the Simons Institute.

\textbf{On LLM use:} The planning, main results and constructions, composition, and reference compilation for this document were completed \emph{without} the use of a large language model (LLM). In the final stages of this work's composition LLMs (here \emph{Claude Opus 5 \& 5.5}) were consulted asynchronously and without edit permissions to identify typos, obvious errors, redundancies, or notational deficiencies; any suggestions were reviewed and implemented serially and by-hand if judged to be correct. The most substantive of these suggestions included improved proofs of the lower bounds in Sec.~\ref{sec:lower_bounds}, which were checked and re-written by hand to improve accessibility.

%%%%%%%%%%%%%%%%%%%%%%%%%%%%%%%%%%%%%%%%%%%%%
%%%%%%%%%%%%%%%%%%%%%%%%%%%%%%%%%%%%%%%%%%%%%
%%%%%%%%%%%%%%%%%%%%%%%%%%%%%%%%%%%%%%%%%%%%%
\bibliography{main}

@article{lyc_optimal_pulses_14,
    title={Optimal arbitrarily accurate composite pulse sequences},
    author={Low, Guang Hao and Yoder, Theodore J and Chuang, Isaac L},
    journal={Phys. Rev. A},
    volume={89},
    number={2},
    pages={022341},
    year={2014},
    publisher={APS}
}

@article{lyc_equiangular_16,
    title = {Methodology of Resonant Equiangular Composite Quantum Gates},
    author = {G. H. Low and T. J. Yoder and I. L. Chuang},
    journal = {Phys. Rev. X},
    volume = {6},
    issue = {4},
    pages = {041067},
    numpages = {13},
    year = {2016},
    publisher = {American Physical Society},
    doi = {10.1103/PhysRevX.6.041067},
    url = {https://link.aps.org/doi/10.1103/PhysRevX.6.041067}
}

@article{lc_ham_sim_17,
    title = {Optimal {H}amiltonian Simulation by Quantum Signal Processing},
    author = {G. H. Low and I. L. Chuang},
    journal = {Phys. Rev. Lett.},
    volume = {118},
    issue = {1},
    pages = {010501},
    numpages = {5},
    year = {2017},
    publisher = {American Physical Society},
    doi = {10.1103/PhysRevLett.118.010501},
    url = {https://link.aps.org/doi/10.1103/PhysRevLett.118.010501}
}

@article{lc_qubitization_19,
    title={Hamiltonian Simulation by Qubitization},
    volume={3},
    doi={10.22331/q-2019-07-12-163},
    journal={Quantum},
    author={G. H. Low and I. L. Chuang},
    year={2019},
    pages={163},
    url={http://dx.doi.org/10.22331/q-2019-07-12-163}
}

@article{gslw_19,
    title={Quantum singular value transformation and beyond: exponential improvements for quantum matrix arithmetics},
    ISBN={9781450367059},
    url={http://dx.doi.org/10.1145/3313276.3316366},
    DOI={10.1145/3313276.3316366},
    journal={Proceedings of the 51st Annual ACM SIGACT Symposium on Theory of Computing (STOC)},
    publisher={ACM},
    author={Gilyén, András and Su, Yuan and Low, Guang Hao and Wiebe, Nathan},
    year={2019}
}

@article{mw_gqsp_24,
    title = {Generalized Quantum Signal Processing},
    author = {Motlagh, Danial and Wiebe, Nathan},
    journal = {PRX Quantum},
    volume = {5},
    issue = {2},
    pages = {020368},
    numpages = {16},
    year = {2024},
    month = {Jun},
    publisher = {American Physical Society},
    doi = {10.1103/PRXQuantum.5.020368},
    url = {https://link.aps.org/doi/10.1103/PRXQuantum.5.020368}
}

@book{nc_textbook_11,
    author = {Nielsen, Michael A. and Chuang, Isaac L.},
    title = {Quantum {C}omputation and {Q}uantum {I}nformation},
    year = {2011},
    isbn = {1107002176},
    publisher = {Cambridge University Press},
    address = {USA},
    edition = {10th}
}

@article{lw_mat_eq_25,
    title={Quantum algorithm for linear matrix equations}, 
    author={Rolando D. Somma and Guang Hao Low and Dominic W. Berry and Ryan Babbush},
    year={2025},
    journal={arXiv preprint, arXiv:2508.02822},
    url={https://arxiv.org/abs/2508.02822}, 
}

@book{tb_num_alg_22,
    author = {Trefethen, Lloyd N. and Bau, David},
    title = {Numerical Linear Algebra, Twenty-fifth Anniversary Edition},
    publisher = {Society for Industrial and Applied Mathematics},
    year = {2022},
    doi = {10.1137/1.9781611977165},
    address = {Philadelphia, PA},
    edition   = {},
    URL = {https://epubs.siam.org/doi/abs/10.1137/1.9781611977165},
    eprint = {https://epubs.siam.org/doi/pdf/10.1137/1.9781611977165}
}

@article{ts_prod_lin_op_25,
    author = {Fabienne Chouraqui},
    title = {When the {Tracy-Singh} product of matrices represents a certain operation on linear operators},
    journal = {Communications in Algebra},
    volume = {53},
    number = {1},
    pages = {1--10},
    year = {2025},
    publisher = {Taylor \& Francis},
    doi = {10.1080/00927872.2024.2371560},
    URL = {https://doi.org/10.1080/00927872.2024.2371560},
    eprint = {https://doi.org/10.1080/00927872.2024.2371560}
}

@book{tref_approx_19,
    author = {Trefethen, Lloyd N.},
    title = {Approximation Theory and Approximation Practice, Extended Edition},
    publisher = {Society for Industrial and Applied Mathematics},
    year = {2019},
    doi = {10.1137/1.9781611975949},
    address = {Philadelphia, PA},
    edition   = {},
    URL = {https://epubs.siam.org/doi/abs/10.1137/1.9781611975949},
    eprint = {https://epubs.siam.org/doi/pdf/10.1137/1.9781611975949}
}

@article{lw_comp_gadget_19,
    title={Hamiltonian Simulation in the Interaction Picture}, 
    author={Guang Hao Low and Nathan Wiebe},
    year={2019},
    journal={arXiv preprint, arXiv:1805.00675},
    url={https://arxiv.org/abs/1805.00675}, 
}

@article{gmkf_nonlinear_24,
    title = {Nonlinear transformation of complex amplitudes via quantum singular value transformation},
    author = {Guo, Naixu and Mitarai, Kosuke and Fujii, Keisuke},
    journal = {Phys. Rev. Res.},
    volume = {6},
    issue = {4},
    pages = {043227},
    numpages = {9},
    year = {2024},
    month = {Dec},
    publisher = {American Physical Society},
    doi = {10.1103/PhysRevResearch.6.043227},
    url = {https://link.aps.org/doi/10.1103/PhysRevResearch.6.043227}
}

@article{ls_qep_26,
    author = {Low, Guang Hao and Su, Yuan},
    title = {Quantum Eigenvalue Processing},
    journal = {SIAM Journal on Computing},
    volume = {55},
    number = {1},
    pages = {135-215},
    year = {2026},
    doi = {10.1137/24M1689363},
    URL = {https://doi.org/10.1137/24M1689363},
    eprint = {https://doi.org/10.1137/24M1689363}
}

@article{fast_spec_amp_25,
    title = {Fast Quantum Simulation of Electronic Structure by Spectral Amplification},
    author = {Low, Guang Hao and King, Robbie and Berry, Dominic W. and Han, Qiushi and DePrince, A. Eugene and White, Alec F. and Babbush, Ryan and Somma, Rolando D. and Rubin, Nicholas C.},
    journal = {Phys. Rev. X},
    volume = {15},
    issue = {4},
    pages = {041016},
    numpages = {41},
    year = {2025},
    month = {Oct},
    publisher = {American Physical Society},
    doi = {10.1103/pb2g-j9cw},
    url = {https://link.aps.org/doi/10.1103/pb2g-j9cw}
}

@article{zs_low_energy_24,
    doi = {10.22331/q-2024-08-27-1449},
    url = {https://doi.org/10.22331/q-2024-08-27-1449},
    title = {Hamiltonian simulation for low-energy states with optimal time dependence},
    author = {Zlokapa, Alexander and Somma, Rolando D.},
    journal = {{Quantum}},
    issn = {2521-327X},
    publisher = {{Verein zur F{\"{o}}rderung des Open Access Publizierens in den Quantenwissenschaften}},
    volume = {8},
    pages = {1449},
    month = aug,
    year = {2024}
}

@article{nkl_dense_rank_be_22,
    doi = {10.22331/q-2022-12-13-876},
    url = {https://doi.org/10.22331/q-2022-12-13-876},
    title = {Block-encoding dense and full-rank kernels using hierarchical matrices: applications in quantum numerical linear algebra},
    author = {Nguyen, Quynh T. and Kiani, Bobak T. and Lloyd, Seth},
    journal = {{Quantum}},
    issn = {2521-327X},
    publisher = {{Verein zur F{\"{o}}rderung des Open Access Publizierens in den Quantenwissenschaften}},
    volume = {6},
    pages = {876},
    month = dec,
    year = {2022}
}

@article{clvy_q_circ_be_24,
    author = {Camps, Daan and Lin, Lin and Van Beeumen, Roel and Yang, Chao},
    title = {Explicit Quantum Circuits for Block Encodings of Certain Sparse Matrices},
    journal = {SIAM Journal on Matrix Analysis and Applications},
    volume = {45},
    number = {1},
    pages = {801-827},
    year = {2024},
    doi = {10.1137/22M1484298},
    URL = {https://doi.org/10.1137/22M1484298},
    eprint = {https://doi.org/10.1137/22M1484298}
}

@inproceedings{cv_fable_be_22,
    title={Fable: Fast approximate quantum circuits for block-encodings},
    author={Camps, Daan and Van Beeumen, Roel},
    booktitle={2022 IEEE International Conference on Quantum Computing and Engineering (QCE)},
    pages={104--113},
    year={2022},
    organization={IEEE}
}

@article{scc_be_struct_24,
    title={Block-encoding structured matrices for data input in quantum computing},
    author={S{\"u}nderhauf, Christoph and Campbell, Earl and Camps, Joan},
    journal={Quantum},
    volume={8},
    pages={1226},
    year={2024},
    publisher={Verein zur F{\"o}rderung des Open Access Publizierens in den Quantenwissenschaften}
}

@article{yuan_cobble_25,
    title={Cobble: Compiling Block Encodings for Quantum Computational Linear Algebra}, 
    author={Charles Yuan},
    year={2025},
    journal={arXiv preprint, arXiv:2511.01736},
    url={https://arxiv.org/abs/2511.01736}, 
}

@inproceedings{bck_ham_lcu_15,
    title={Hamiltonian simulation with nearly optimal dependence on all parameters},
    author={Berry, Dominic W. and Childs, Andrew M. and Kothari, Robin},
    booktitle={{2015 IEEE 56th Annual Symposium on Foundations of Computer Science (FOCS)}},
    pages={792--809},
    year={2015},
    organization={IEEE}
}

@article{ks_mat_arith_26,
    title={Quantum matrix arithmetics with {Hamiltonian} evolution}, 
    author={Christopher Kang and Yuan Su},
    year={2026},
    journal={arXiv preprint, arXiv:2510.06316},
    url={https://arxiv.org/abs/2510.06316}, 
}

@misc{lw_lecture_notes_26,
    author = "Lin Lin and Nathan Wiebe",
    title  = "Quantum Algorithms for Scientific Computation",
    note   = "url: \url{https://math.berkeley.edu/~linlin/qasc/live_notes_0429.pdf}",
    year = 2026,
}

@inproceedings{bccks_ham_lcu_14, 
    series={STOC ’14},
    title={{Exponential improvement in precision for simulating sparse Hamiltonians}},
    url={http://dx.doi.org/10.1145/2591796.2591854},
    DOI={10.1145/2591796.2591854},
    booktitle={Proceedings of the forty-sixth annual ACM symposium on Theory of computing (STOC)},
    publisher={ACM},
    author={Berry, Dominic W. and Childs, Andrew M. and Cleve, Richard and Kothari, Robin and Somma, Rolando D.},
    year={2014},
    month=May, 
    pages={283–292},
    collection={STOC ’14} 
}

@article{gychanar_had_prod_25,
    title={Quantum Transformer: Accelerating model inference via quantum linear algebra}, 
    author={Naixu Guo and Zhan Yu and Matthew Choi and Yizhan Han and Aman Agrawal and Kouhei Nakaji and Alán Aspuru-Guzik and Patrick Rebentrost},
    year={2025},
    journal={arXiv preprint, arXiv:2402.16714},
    url={https://arxiv.org/abs/2402.16714}, 
}

@article{rhgpr_acc_inf_25,
    title={Accelerating Inference for Multilayer Neural Networks with Quantum Computers}, 
    author={Arthur G. Rattew and Po-Wei Huang and Naixu Guo and Lirandë Pira and Patrick Rebentrost},
    year={2025},
    journal={arXiv preprint, arXiv:2510.07195},
    url={https://arxiv.org/abs/2510.07195}, 
}

@article{rr_nonlinear_23,
    title={Non-Linear Transformations of Quantum Amplitudes: Exponential Improvement, Generalization, and Applications}, 
    author={Arthur G. Rattew and Patrick Rebentrost},
    year={2023},
    journal={arXiv preprint, arXiv:2309.09839},
    url={https://arxiv.org/abs/2309.09839}, 
}

@article{hcss_nonlinear_23,
    title = {Nonlinear transformations in quantum computation},
    author = {Holmes, Zo\"e and Coble, Nolan J. and Sornborger, Andrew T. and Subaşı},
    journal = {Phys. Rev. Res.},
    volume = {5},
    issue = {1},
    pages = {013105},
    numpages = {20},
    year = {2023},
    month = {Feb},
    publisher = {American Physical Society},
    doi = {10.1103/PhysRevResearch.5.013105},
    url = {https://link.aps.org/doi/10.1103/PhysRevResearch.5.013105}
}

@article{mrtc_unification_21,
    title={Grand unification of quantum algorithms},
    author={Martyn, John M and Rossi, Zane M and Tan, Andrew K and Chuang, Isaac L},
    journal={PRX Quantum},
    volume={2},
    number={4},
    pages={040203},
    year={2021},
    publisher={APS}
}

@book{hj_matrix_analysis_85, 
    place={Cambridge}, 
    title={Matrix Analysis}, 
    publisher={Cambridge University Press}, 
    author={Horn, Roger A. and Johnson, Charles R.}, 
    year={1985}
}

@book{hj_topics_91, 
    place={Cambridge}, 
    title={Topics in Matrix Analysis}, 
    publisher={Cambridge University Press}, 
    author={Horn, Roger A. and Johnson, Charles R.}, 
    year={1991}
}

@book{be_polynomial_ineq_95,
    author    = {Borwein, Peter and Erd{\'e}lyi, Tam{\'a}s},
    title     = {Polynomials and Polynomial Inequalities},
    series    = {Graduate Texts in Mathematics},
    volume    = {161},
    publisher = {Springer-Verlag},
    address   = {New York},
    year      = {1995},
    doi       = {10.1007/978-1-4612-0793-1},
    isbn      = {978-0-387-94509-5}
}

@book{higham_mat_func_08,
    author    = {Higham, Nicholas J.},
    title     = {Functions of Matrices: Theory and Computation},
    publisher = {Society for Industrial and Applied Mathematics},
    address   = {Philadelphia, PA},
    year      = {2008},
    doi       = {10.1137/1.9780898717778},
    isbn      = {978-0-89871-646-7}
}

@article{cks_linear_17,
    title={Quantum Algorithm for Systems of Linear Equations with Exponentially Improved Dependence on Precision},
    volume={46},
    ISSN={1095-7111},
    url={http://dx.doi.org/10.1137/16M1087072},
    DOI={10.1137/16m1087072},
    number={6},
    journal={SIAM Journal on Computing},
    publisher={Society for Industrial & Applied Mathematics (SIAM)},
    author={Childs, Andrew M. and Kothari, Robin and Somma, Rolando D.},
    year={2017},
    month={Jan}, 
    pages={1920–1950} 
}

@inproceedings{cs_qsvt_tang_tian_23,
    author = {Ewin Tang and Kevin Tian},
    title = {A {CS} guide to the quantum singular value transformation},
    booktitle = {2024 Symposium on Simplicity in Algorithms (SOSA)},
    pages = {121-143},
    doi = {10.1137/1.9781611977936.13},
    URL = {https://epubs.siam.org/doi/abs/10.1137/1.9781611977936.13},
    eprint = {https://epubs.siam.org/doi/pdf/10.1137/1.9781611977936.13},
    year={2024},
    publisher={},
}

@article{hc_block_kronecker_89,
    author = {Hyland, David C. and Collins, Jr., Emmanuel G.},
    title = {Block {Kronecker} Products and Block Norm Matrices in Large-Scale Systems Analysis},
    journal = {SIAM J. Matrix Anal. Appl.},
    volume = {10},
    number = {1},
    pages = {18-29},
    year = {1989},
    doi = {10.1137/0610002},
    URL = {https://doi.org/10.1137/0610002},
    eprint = {https://doi.org/10.1137/0610002}
}

@article{liu_ts_kr_prod_99,
    title = {{Matrix results on the Khatri-Rao and Tracy-Singh products}},
    journal = {Linear Algebra Appl.},
    volume = {289},
    number = {1},
    pages = {267-277},
    year = {1999},
    issn = {0024-3795},
    doi = {https://doi.org/10.1016/S0024-3795(98)10209-4},
    url = {https://www.sciencedirect.com/science/article/pii/S0024379598102094},
    author = {Shuangzhe Liu},
}

@article{kr_prod_68,
    ISSN = {0581572X},
    URL = {http://www.jstor.org/stable/25049527},
    author = {C. G. Khatri and C. Radhakrishna Rao},
    journal = {Sankhyā Ser. A},
    number = {2},
    pages = {167--180},
    publisher = {Springer},
    title = {Solutions to Some Functional Equations and Their Applications to Characterization of Probability Distributions},
    volume = {30},
    year = {1968}
}

@article{tk_tr_prod_app_89,
    ISSN = {03195724},
    URL = {http://www.jstor.org/stable/3314768},
    author = {Derrick S. Tracy and Kankanam G. Jinadasa},
    journal = {Can. J. Stat.},
    number = {1},
    pages = {107--120},
    publisher = {Statistical Society of Canada, Wiley},
    title = {Partitioned {Kronecker} Products of Matrices and Applications},
    volume = {17},
    year = {1989}
}

@article{chia_low_rank_22,
    author = {Chia, Nai-Hui and Gily\'{e}n, Andr\'{a}s Pal and Li, Tongyang and Lin, Han-Hsuan and Tang, Ewin and Wang, Chunhao},
    title = {Sampling-based Sublinear Low-rank Matrix Arithmetic Framework for Dequantizing Quantum Machine Learning},
    year = {2022},
    issue_date = {October 2022},
    publisher = {Association for Computing Machinery},
    address = {New York, NY, USA},
    volume = {69},
    number = {5},
    issn = {0004-5411},
    url = {https://doi.org/10.1145/3549524},
    doi = {10.1145/3549524},
    journal = {J. ACM},
    month = oct,
    articleno = {33},
    numpages = {72}
}

@article{gl_dequant_23,
   title={{Dequantizing the Quantum Singular Value Transformation: Hardness and Applications to Quantum Chemistry and the Quantum PCP Conjecture}},
   volume={52},
   ISSN={1095-7111},
   url={http://dx.doi.org/10.1137/22M1513721},
   DOI={10.1137/22m1513721},
   number={4},
   journal={SIAM Journal on Computing},
   publisher={Society for Industrial & Applied Mathematics (SIAM)},
   author={Gharibian, Sevag and Le Gall, François},
   year={2023},
   month=Aug, pages={1009–1038} 
}

@inproceedings{tang_19, 
    series={STOC ’19},
    title={A quantum-inspired classical algorithm for recommendation systems},
    url={http://dx.doi.org/10.1145/3313276.3316310},
    DOI={10.1145/3313276.3316310},
    booktitle={Proceedings of the 51st Annual ACM SIGACT Symposium on Theory of Computing},
    publisher={ACM},
    author={Tang, Ewin},
    year={2019},
    pages={217–228},
    collection={STOC ’19} 
}

@article{aa_forrelation_18,
    author = {Aaronson, Scott and Ambainis, Andris},
    title = {Forrelation: A Problem That Optimally Separates Quantum from Classical Computing},
    journal = {SIAM Journal on Computing},
    volume = {47},
    number = {3},
    pages = {982-1038},
    year = {2018},
    doi = {10.1137/15M1050902},
    
    URL = { 
        
            https://doi.org/10.1137/15M1050902
        
        
    
    },
    eprint = { 
        
            https://doi.org/10.1137/15M1050902
        
        
    
    }
}

@inproceedings{binary_neural_16,
    author = {Hubara, Itay and Courbariaux, Matthieu and Soudry, Daniel and El-Yaniv, Ran and Bengio, Yoshua},
    booktitle = {Advances in Neural Information Processing Systems},
    editor = {D. Lee and M. Sugiyama and U. Luxburg and I. Guyon and R. Garnett},
    pages = {},
    publisher = {Curran Associates, Inc.},
    title = {Binarized Neural Networks},
    url = {https://proceedings.neurips.cc/paper_files/paper/2016/file/d8330f857a17c53d217014ee776bfd50-Paper.pdf},
    volume = {29},
    year = {2016}
}

@article{ritt_22,
    author    = {Ritt, Joseph Fels},
    title     = {Prime and composite polynomials},
    journal   = {Transactions of the American Mathematical Society},
    volume    = {23},
    number    = {1},
    pages     = {51--66},
    year      = {1922},
    publisher = {American Mathematical Society},
    doi       = {10.2307/1988911}
}

@inproceedings{sign_sgd_18,
    title={signSGD: Compressed optimisation for non-convex problems},
    author={Bernstein, Jeremy and Wang, Yu-Xiang and Azizzadenesheli, Kamyar and Anandkumar, Animashree},
    booktitle={International Conference on Machine Learning},
    pages={560--569},
    year={2018},
    organization={PMLR},
    url={https://proceedings.mlr.press/v80/bernstein18a.html}
}

@article{mn_mat_comm_79,
    ISSN = {00905364, 21688966},
    URL = {http://www.jstor.org/stable/2958818},
    author = {Jan R. Magnus and H. Neudecker},
    journal = {Ann. Stat.},
    number = {2},
    pages = {381--394},
    publisher = {Institute of Mathematical Statistics},
    title = {The Commutation Matrix: Some Properties and Applications},
    volume = {7},
    year = {1979}
}

@article{flt_time_marching_23,
    doi = {10.22331/q-2023-03-20-955},
    url = {https://doi.org/10.22331/q-2023-03-20-955},
    title = {Time-marching based quantum solvers for time-dependent linear differential equations},
    author = {Fang, Di and Lin, Lin and Tong, Yu},
    journal = {{Quantum}},
    issn = {2521-327X},
    publisher = {{Verein zur F{\"{o}}rderung des Open Access Publizierens in den Quantenwissenschaften}},
    volume = {7},
    pages = {955},
    month = mar,
    year = {2023}
}

@article{gls_qet_arb_26,
    title={Quantum Eigenvalue Transformations for Arbitrary Matrices}, 
    author={Xabier Gutiérrez and Lorenzo Laneve and Mikel Sanz},
    year={2026},
    journal={arXiv preprint, arXiv:2604.19688},
    url={https://arxiv.org/abs/2604.19688}, 
}

@article{tensor_decomp_09,
    author = {Kolda, Tamara G. and Bader, Brett W.},
    title = {Tensor Decompositions and Applications},
    journal = {SIAM Review},
    volume = {51},
    number = {3},
    pages = {455-500},
    year = {2009},
    doi = {10.1137/07070111X},
    URL = {https://doi.org/10.1137/07070111X},
    eprint = {https://doi.org/10.1137/07070111X},
}

@article{hl_tensor_np_13,
    author = {Hillar, Christopher J. and Lim, Lek-Heng},
    title = {{Most Tensor Problems Are NP-Hard}},
    year = {2013},
    issue_date = {November 2013},
    publisher = {ACM},
    address = {New York, NY, USA},
    volume = {60},
    number = {6},
    issn = {0004-5411},
    url = {https://doi.org/10.1145/2512329},
    doi = {10.1145/2512329},
    journal = {J. ACM},
    month = nov,
    articleno = {45},
    numpages = {39},
}

@article{cdghs_finding_qsp_angles_20,
    title={Finding Angles for Quantum Signal Processing with Machine Precision},
    author={Rui Chao and Dawei Ding and Andras Gilyen and Cupjin Huang and Mario Szegedy},
    year={2020},
    journal={arXiv preprint, arXiv:2003.02831},
}

@article{dlnw_robust_iter_24,
    author = {Dong, Yulong and Lin, Lin and Ni, Hongkang and Wang, Jiasu},
    title = {Robust Iterative Method for Symmetric Quantum Signal Processing in All Parameter Regimes},
    journal = {SIAM J. Sci. Comput.},
    volume = {46},
    number = {5},
    pages = {A2951-A2971},
    year = {2024},
    doi = {10.1137/23M1598192},
    URL = {https://doi.org/10.1137/23M1598192},
    eprint = {https://doi.org/10.1137/23M1598192}
}

@article{wdl_sym_qsp_energy_22,
    doi = {10.22331/q-2022-11-03-850},
    url = {https://doi.org/10.22331/q-2022-11-03-850},
    title = {On the energy landscape of symmetric quantum signal processing},
    author = {Wang, Jiasu and Dong, Yulong and Lin, Lin},
    journal = {{Quantum}},
    issn = {2521-327X},
    publisher = {{Verein zur F{\"{o}}rderung des Open Access Publizierens in den Quantenwissenschaften}},
    volume = {6},
    pages = {850},
    month = {Nov},
    year = {2022}
}

@article{haah_decomposition_19,
  doi = {10.22331/q-2019-10-07-190},
  url = {https://doi.org/10.22331/q-2019-10-07-190},
  title = {Product Decomposition of Periodic Functions in Quantum Signal Processing},
  author = {Haah, Jeongwan},
  journal = {{Quantum}},
  issn = {2521-327X},
  publisher = {{Verein zur F{\"{o}}rderung des Open Access Publizierens in den Quantenwissenschaften}},
  volume = {3},
  pages = {190},
  month = oct,
  year = {2019}
}

@article{dmwl_efficient_phases_21,
  title={Efficient phase-factor evaluation in quantum signal processing},
  volume={103},
  ISSN={2469-9934},
  url={http://dx.doi.org/10.1103/PhysRevA.103.042419},
  DOI={10.1103/physreva.103.042419},
  number={4},
  journal={Phys. Rev. A},
  publisher={American Physical Society (APS)},
  author={Dong, Yulong and Meng, Xiang and Whaley, K. Birgitta and Lin, Lin},
  year={2021},
  month=apr
}

@article{ts_prod_72,
    author = {Tracy, Derrick S. and Singh, Rana P.},
    title = {A new matrix product and its applications in partitioned matrix differentiation},
    journal = {Statistica Neerlandica},
    volume = {26},
    number = {4},
    pages = {143-157},
    doi = {https://doi.org/10.1111/j.1467-9574.1972.tb00199.x},
    url = {https://onlinelibrary.wiley.com/doi/abs/10.1111/j.1467-9574.1972.tb00199.x},
    eprint = {https://onlinelibrary.wiley.com/doi/pdf/10.1111/j.1467-9574.1972.tb00199.x},
    year = {1972}
}

@article{zzrf_basic_q_lin_alg_21,
    author = {L. Zhao and Z. Zhao and P. Rebentrost and J. Fitzsimons},
    title = {Compiling basic linear algebra subroutines for quantum computers},
    journal = {Quantum Mach. Intell.},
    year = {2021},
    volume = {3},
    number = {21},
    doi={https://doi.org/10.1007/s42484-021-00048-8}
}

@book{gv_mat_computation_13,
    author = {Golub, Gene H. and Van Loan, Charles F.},
    title = {Matrix Computations},
    publisher = {Johns Hopkins University Press},
    year = {2013},
    doi = {10.1137/1.9781421407944},
    address = {Philadelphia, PA},
    URL = {https://epubs.siam.org/doi/abs/10.1137/1.9781421407944},
    eprint = {https://epubs.siam.org/doi/pdf/10.1137/1.9781421407944}
}

@article{vg_reducing_space_25,
    title={Methods for Reducing Ancilla-Overhead in Block Encodings}, 
    author={Francisca Vasconcelos and András Gilyén},
    year={2025},
    journal={arXiv preprint, arXiv:2507.07900},
    url={https://arxiv.org/abs/2507.07900}, 
}

@article{mrclc_parallel_qsp_24,
    doi = {10.22331/q-2025-08-27-1834},
    url = {https://doi.org/10.22331/q-2025-08-27-1834},
    title = {Parallel Quantum Signal Processing Via Polynomial Factorization},
    author = {Martyn, John M. and Rossi, Zane M. and Cheng, Kevin Z. and Liu, Yuan and Chuang, Isaac L.},
    journal = {{Quantum}},
    issn = {2521-327X},
    publisher = {{Verein zur F{\"{o}}rderung des Open Access Publizierens in den Quantenwissenschaften}},
    volume = {9},
    pages = {1834},
    year = {2025}  
}

@article{cw_lcu_12,
    author = {Childs, Andrew M. and Wiebe, Nathan},
    title = {Hamiltonian Simulation Using Linear Combinations of Unitary Operations},
    year = {2012},
    publisher = {Rinton Press, Incorporated},
    address = {Paramus, NJ},
    journal = {Quantum Info. Comput.},
    volume = {12},
    number = {11–12},
    issn = {1533-7146},
    month = {Nov},
    pages = {901–924},
    numpages = {24},
}

@article{styan_had_prod_73,
    title = {Hadamard products and multivariate statistical analysis},
    journal = {Linear Algebra and its Applications},
    volume = {6},
    pages = {217-240},
    year = {1973},
    issn = {0024-3795},
    doi = {https://doi.org/10.1016/0024-3795(73)90023-2},
    url = {https://www.sciencedirect.com/science/article/pii/0024379573900232},
    author = {George P.H. Styan},
}

@book{mn_had_prod_textbook_99,
    author = {Magnus, Jan R. and Neudecker, Heinz},
    edition = {Second},
    isbn = {0471986321 9780471986324 047198633X 9780471986331},
    publisher = {John Wiley},
    title = {Matrix Differential Calculus with Applications in Statistics and Econometrics},
    year={1999},
}

@article{martyn_rall_halving_24,
    title={Halving the Cost of Quantum Algorithms with Randomization}, 
    author={John M. Martyn and Patrick Rall},
    year={2025},
    journal={npj Quantum Inf.},
    volume={11},
    number={47},
}

@article{rc_m_qsp_22,
    doi = {10.22331/q-2022-09-20-811},
    url = {https://doi.org/10.22331/q-2022-09-20-811},
    title = {Multivariable quantum signal processing ({M}-{QSP}): prophecies of the two-headed oracle},
    author = {Rossi, Zane M. and Chuang, Isaac L.},
    journal = {{Quantum}},
    issn = {2521-327X},
    publisher = {{Verein zur F{\"{o}}rderung des Open Access Publizierens in den Quantenwissenschaften}},
    volume = {6},
    pages = {811},
    year = {2022}
}

@article{rcc_modular_qsp_23,
    doi = {10.22331/q-2025-06-18-1776},
    url = {https://doi.org/10.22331/q-2025-06-18-1776},
    title = {Modular quantum signal processing in many variables},
    author = {Rossi, Zane M. and Ceroni, Jack L. and Chuang, Isaac L.},
    journal = {{Quantum}},
    issn = {2521-327X},
    publisher = {{Verein zur F{\"{o}}rderung des Open Access Publizierens in den Quantenwissenschaften}},
    volume = {9},
    pages = {1776},
    month = jun,
    year = {2025}
}

@article{rbmc_su_11_qsp,
    title={Quantum signal processing with continuous variables}, 
    author={Zane M. Rossi and Victor M. Bastidas and William J. Munro and Isaac L. Chuang},
    year={2023},
    journal={arXiv preprint, arXiv:2304.14383},
}

@article{amt_23,
    title={Quantum signal processing and nonlinear {Fourier} analysis}, 
    author={Michel Alexis and Gevorg Mnatsakanyan and Christoph Thiele},
    year={2024},
    journal={Rev Mat Complut},
    volume={37},
    pages={655-694},
    doi={https://doi.org/10.1007/s13163-024-00494-5},
}

@article{szego_nlfa_qsp_24,
    author = {Alexis, Michel and Lin, Lin and Mnatsakanyan, Gevorg and Thiele, Christoph and Wang, Jiasu},
    title = {Infinite quantum signal processing for arbitrary {S}zegő functions},
    journal = {Commun. Pure Appl. Math.},
    volume = {79},
    number = {1},
    pages = {123-174},
    doi = {https://doi.org/10.1002/cpa.70007},
    url = {https://onlinelibrary.wiley.com/doi/abs/10.1002/cpa.70007},
    eprint = {https://onlinelibrary.wiley.com/doi/pdf/10.1002/cpa.70007},
    year = {2026},
}

@article{bw_analytical_phase_26,
    title={Analytical Angle-Finding and Series Expansions for Quantum Signal Processing via Orthogonal Polynomial Theory}, 
    author={Pierre-Antoine Bernard and Nathan Wiebe},
    year={2026},
    journal={arXiv preprint, arXiv:2605.05321},
    url={https://arxiv.org/abs/2605.05321}, 
}

@article{johnson_had_mat_74,
    title={Hadamard products of matrices},
    author={Charles R. Johnson},
    journal={Linear \& Multilinear Algebra},
    year={1974},
    volume={1},
    pages={295-307},
    url={https://api.semanticscholar.org/CorpusID:120925607}
}

@article{js_rga_86,
    author = {Johnson, Charles R. and Shapiro, Helene M.},
    title = {Mathematical Aspects of the Relative Gain Array {$(A \circ A^{-T})$}},
    journal = {SIAM Journal on Algebraic Discrete Methods},
    volume = {7},
    number = {4},
    pages = {627-644},
    year = {1986},
    doi = {10.1137/0607069},
    URL = {https://doi.org/10.1137/0607069},
    eprint = {https://doi.org/10.1137/0607069}
}

@article{learning_had_prod_25,
    author={Chrysos, Grigorios G. and Wu, Yongtao and Pascanu, Razvan and Torr, Philip H.S. and Cevher, Volkan},
    journal={IEEE Transactions on Pattern Analysis and Machine Intelligence}, 
    title={Hadamard Product in Deep Learning: Introduction, Advances and Challenges}, 
    year={2025},
    volume={47},
    number={8},
    pages={6531-6549},
    doi={10.1109/TPAMI.2025.3560423}
}

@article{kn_block_kron_vecb_91,
    title = {Block {K}ronecker products and the vecb operator},
    journal = {Linear Algebra and its Applications},
    volume = {149},
    pages = {165-184},
    year = {1991},
    issn = {0024-3795},
    doi = {https://doi.org/10.1016/0024-3795(91)90332-Q},
    url = {https://www.sciencedirect.com/science/article/pii/002437959190332Q},
    author = {Ruud H. Koning and Heinz Neudecker and Tom Wansbeek},
}

@article{ltlf_matrix_prod_quantum_26,
    title = {Quantum circuit implementation of two matrix product operations and elementary column transformations},
    journal = {Physics Letters A},
    volume = {587},
    pages = {131717},
    year = {2026},
    issn = {0375-9601},
    doi = {https://doi.org/10.1016/j.physleta.2026.131717},
    url = {https://www.sciencedirect.com/science/article/pii/S0375960126003932},
    author = {Yu-Hang Liu and Yuan-Hong Tao and Jing-Run Lan and Shao-Ming Fei},
}

@article{dlx_prod_be_25,
    title={Products between block-encodings}, 
    author={Dekuan Dong and Yingzhou Li and Jungong Xue},
    year={2025},
    journal={arXiv preprint, arXiv:2509.15779},
    url={https://arxiv.org/abs/2509.15779}, 
}

@article{ls_opt_lchs_25,
    title={Optimal quantum simulation of linear non-unitary dynamics}, 
    author={Guang Hao Low and Rolando D. Somma},
    year={2025},
    journal={arXiv preprint, arXiv:2508.19238},
    url={https://arxiv.org/abs/2508.19238}, 
}

@article{acl_lchs_near_opt_23,
    title={Quantum Algorithm for Linear Non-unitary Dynamics with Near-Optimal Dependence on All Parameters}, 
    author={Dong An and Andrew M. Childs and Lin Lin},
    year={2026},
    journal={Commun. Math. Phys.},
    volume={407},
    number={19},
    url={https://doi.org/10.1007/s00220-025-05509-w} 
}

@article{sbs_lin_mat_diffeq_26,
    title={Efficient quantum algorithm for linear matrix differential equations and applications to open quantum systems}, 
    author={Sophia Simon and Dominic W. Berry and Rolando D. Somma},
    year={2026},
    journal={arXiv preprint, arXiv:2605.16195},
    url={https://arxiv.org/abs/2605.16195}, 
}

@article{dll_lchs_state_cost_23,
    title = {Linear Combination of {H}amiltonian Simulation for Nonunitary Dynamics with Optimal State Preparation Cost},
    author = {An, Dong and Liu, Jin-Peng and Lin, Lin},
    journal = {Phys. Rev. Lett.},
    volume = {131},
    issue = {15},
    pages = {150603},
    numpages = {6},
    year = {2023},
    month = {Oct},
    publisher = {American Physical Society},
    doi = {10.1103/PhysRevLett.131.150603},
    url = {https://link.aps.org/doi/10.1103/PhysRevLett.131.150603}
}

@article{campbell_qdrfit_19,
    title={Random Compiler for Fast {H}amiltonian Simulation},
    volume={123},
    ISSN={1079-7114},
    url={http://dx.doi.org/10.1103/PhysRevLett.123.070503},
    DOI={10.1103/physrevlett.123.070503},
    number={7},
    journal={Phys. Rev. Lett.},
    publisher={American Physical Society (APS)},
    author={Campbell, Earl},
    year={2019},
}

@article{cstwz_trotter_comm_21,
    title = {Theory of Trotter Error with Commutator Scaling},
    author = {Childs, Andrew M. and Su, Yuan and Tran, Minh C. and Wiebe, Nathan and Zhu, Shuchen},
    journal = {Phys. Rev. X},
    volume = {11},
    issue = {1},
    pages = {011020},
    numpages = {49},
    year = {2021},
    month = {Feb},
    publisher = {American Physical Society},
    doi = {10.1103/PhysRevX.11.011020},
    url = {https://link.aps.org/doi/10.1103/PhysRevX.11.011020}
}

@article{chak_lcu_near_24,
    doi = {10.22331/q-2024-10-10-1496},
    url = {https://doi.org/10.22331/q-2024-10-10-1496},
    title = {Implementing any Linear Combination of Unitaries on Intermediate-term Quantum Computers},
    author = {Chakraborty, Shantanav},
    journal = {{Quantum}},
    issn = {2521-327X},
    publisher = {{Verein zur F{\"{o}}rderung des Open Access Publizierens in den Quantenwissenschaften}},
    volume = {8},
    pages = {1496},
    month = oct,
    year = {2024}
}

@inproceedings{attention_17,
    author = {Vaswani, Ashish and Shazeer, Noam and Parmar, Niki and Uszkoreit, Jakob and Jones, Llion and Gomez, Aidan N. and Kaiser, \L{}ukasz and Polosukhin, Illia},
    title = {Attention is all you need},
    year = {2017},
    isbn = {9781510860964},
    publisher = {Curran Associates Inc.},
    address = {Red Hook, NY, USA},
    booktitle = {Proceedings of the 31st International Conference on Neural Information Processing Systems},
    pages = {6000–6010},
    numpages = {11},
    location = {Long Beach, California, USA},
    series = {NIPS'17}
}

@article{opt_pca_18,
    title={Optimality and sub-optimality of PCA I: Spiked random matrix models},
    volume={46},
    ISSN={0090-5364},
    url={http://dx.doi.org/10.1214/17-AOS1625},
    DOI={10.1214/17-aos1625},
    number={5},
    journal={The Annals of Statistics},
    publisher={Institute of Mathematical Statistics},
    author={Perry, Amelia and Wein, Alexander S. and Bandeira, Afonso S. and Moitra, Ankur},
    year={2018},
    month=Oct 
}

%%%%%%%%%%%%%%%%%%%%%%%%%%%%%%%%%%%%%%%%%%%%%
\appendix

%%%%%%%%%%%%%%%%%%%%%%%%%%%%%%%%%%%%%%%%%%%%%
\section{The compression gadget} \label{appx:comp_gadget}

\noindent A key component of the main construction of Sec.~\ref{sec:main_const} is a method for reducing auxiliary space usage when taking high-degree matrix products of block encodings. Here we briefly review the background, main statement, and variants of these \emph{compression gadgets} (introduced in \cite{lw_comp_gadget_19}).

The basic conceit of these methods is that, while the gadget for multiplying two block encodings involves the instantiation of disjoint auxiliary qubit registers, e.g.,
    \begin{align}
        (\langle 0|_{r_1}\otimes I_{r_2})(U_{A})_{r_1 r_2}(|0\rangle_{r_1}\otimes I_{r_2}) &= A_{r_2},\\
        (\langle 0|_{r_1'}\otimes I_{r_2})(U_{B})_{r_1' r_2}(|0\rangle_{r_1'}\otimes I_{r_2}) &= B_{r_2},\\
        (\langle 0|_{r_1'r_1}\otimes I_{r_2})(U_{A})_{r_1 r_2} (U_{B})_{r_1' r_2}(| 0\rangle_{r_1'r_1}\otimes I_{r_2}) &= (AB)_{r_2},
    \end{align}
when we are multiplying many matrices at once we can instantiate a smaller auxiliary register which records only vital information about the successful application of each constituent block encoding. This reduces the auxiliary space usage from linear (i.e., copying $r_1, r_1', r_1'', \dots$ according to the example above) to logarithmic in the degree of the product. This improvement is captured by the formal statement of the compression gadget below (depicted in Fig.~\ref{fig:comp_gadget}).

\begin{lemma}[Compression gadget; Lemma 13 in \cite{lw_comp_gadget_19}, generalized in \cite{vg_reducing_space_25}] \label{lem:comp_gadget}
    Let $U_k$ for $k \in \{0, 1, \dots, K - 1\}$ be a series of unitaries respectively block encoding (square) $A_k$,
        \begin{equation}
            (\langle 0|_{a}\otimes I_s) U_k (|0\rangle_{a}\otimes I_s) = A_k,
            \quad
            \lVert A_k \rVert \leq 1,
        \end{equation}
    where each $A_k$ acts on $n_s$ qubits and each block encoding uses $n_a$ additional qubits of space. Then there exists a quantum circuit $V$ with form
        \begin{equation}
            (\langle 0|_{ac}\otimes I_s) V (|0\rangle_{ac}\otimes I_s)
            =
            |0\rangle\langle 0|_{b}\otimes I_s +
            \sum_{k = 0}^{K-1} |k\rangle\langle k|_b \otimes \Big[ \prod_{j = 0}^{k} A_j\Big]_{s},
        \end{equation}
    where $n_b \in \mathcal{O}(n_c) = \mathcal{O}(\log{K})$, and the cost of $V$ is one query to controlled-controlled-$U_k$ and $\mathcal{O}(K(n_a + \log{K}))$ additional single qubit gates. That is, $V$ is a $(1, n_a + n_c, 0)$ block encoding of the block diagonal matrix
        \begin{equation}
            W = I_{s} \oplus \bigoplus_{k = 0}^{K-1} \Big[\prod_{j = 0}^{k} A_j\Big]_{s},
        \end{equation}
    which can be converted to a block encoding of the product of all $A_k$ by selecting only the $k = K-1$ block, i.e., $V$ could also be seen as a $(1, n_a + n_b + n_c, 0) = (1, n_a + \mathcal{O}(\log{K}), 0)$ block encoding the product of all inputs.
\end{lemma}

We should clarify that the form of the block encoding in Lem.~\ref{lem:comp_gadget} has additional structure, namely, post-selecting on registers $a$, $c$ being in the all zeros state, we apply a \emph{block-diagonal} block encoding to $b$, $s$ whose $\ell$-th block is the product of the first $\ell$ of the $A_k$'s.\footnote{This is achieved by the inclusion of \emph{both} controlled modular additions and subtractions, for which we refer interested readers to the original paper.} We require, as mentioned, only a simpler block encoding of the product of all $K$ of the $A_k$'s. In this case we can eliminate the register $b$, while $c$ acts as the target of simple controlled additions from the $a$ register (i.e., if all block encodings are successful, then $a$ contains the all zeros state).\footnote{We've given the original statement from \cite{lw_comp_gadget_19} in Lem.~\ref{lem:comp_gadget} which uses, e.g., $a$ for register names and $n_a$ for their size, while we use $r_k$ for register names and Latin letters for their size.}
We depict the circuit of this simplified and sufficient version in Fig.~\ref{fig:comp_gadget} (nearly identical to the construction in \cite{vg_reducing_space_25}).

While we give only asymptotic space and query complexities in Lem.~\ref{lem:comp_gadget}, quick examination of the circuit in Fig.~\ref{fig:comp_gadget} reveals that the pre-factors involved are quite reasonable, and further optimizations are readily made (e.g., dropping the last projection-controlled addition, given the definition of a block encoding). Indeed, as we discuss a little below, one can define a whole class of related circuits which use coherently applied mid-circuit gates to transduce information about the successful application\slash transformation of intermediate block encoding unitaries. In the case that the input block encodings are approximate, the error propagation is handled in precisely the same way as the usual, linear-space product method, which we reproduce in Appx.~\ref{appx:error_prop} (Rem.~\ref{rem:error_prop_prod}).

\begin{figure}
    \centering
    \includegraphics[width=0.9\linewidth]{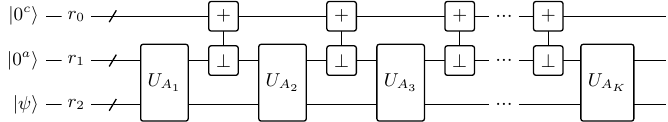}
    \caption{Circuit form of the standard compression gadget (Lem.~\ref{lem:comp_gadget}) introduced and generalized in \cite{lw_comp_gadget_19} and \cite{vg_reducing_space_25} respectively. Here we have suppressed multiple qubit wires with a slash, and compute the block encoding of the matrix product of $\{A_1, \dots, A_K\}$, given access to their block encoding unitaries, which use $r_1$ as their auxiliary register. We also compress notation for our projection-controlled addition (modulo $2^c$) operations, whose controls here are notated by $(\perp)$ and whose targets are notated by $(+)$ (i.e., if $r_1$ is not in the all zeros state, then a modular addition is performed on $r_0$). If $r_0$ and $r_1$ are measured at the end in the all zeroes state, then the desired matrix product is applied to $|\psi\rangle$.}
    \label{fig:comp_gadget}
\end{figure}

Since its introduction the compression gadget has been widely employed \cite{lw_comp_gadget_19, flt_time_marching_23} as well as extended to a general class of techniques for space-saving when manipulating many block encodings simultaneously. These extensions, which may offer interesting avenues to extend techniques of this work, pivot on the following common assumptions. Concretely, in \cite{vg_reducing_space_25} they consider the possibility of moving from an exact compression gadget to an approximate one, allowing for even further reduction in space use. They enumerate the conditions under which this further compression succeeds and fails, and identify key assumptions (e.g., that the inputs are near-identity) where this approximate scheme is practically relevant. Analogously, when the input model is not merely block encodings but access to free Hamiltonian evolution for arbitrary (polynomial) time \cite{ks_mat_arith_26}, there exist further approximate protocols for matrix arithmetic that are space-efficient compared to naïve methods for multiplying\slash adding. These methods, in addition to hybrid techniques involving Trotterization and product formulas, as well as the incorporation of randomness, offer exciting avenues for future work further optimizing our construction in special cases.

%%%%%%%%%%%%%%%%%%%%%%%%%%%%%%%%%%%%%%%%%%%%
\section{Properties of the element-wise product} \label{appx:ewt_prop}

\noindent Here we briefly collect a few known, useful properties of the element-wise product of matrices obeying certain constraints. While our main results do not depend explicitly on nor take advantage of the more involved of these properties, we briefly gloss the most important as there exist numerous other block encoding algorithms (e.g., linear combination of Hamiltonian simulation: LCHS \cite{acl_lchs_near_opt_23, ls_opt_lchs_25, lw_mat_eq_25}) whose resource complexity depends strongly on the spectral properties of their input, and thus these results may have bearing on the efficiency of block encoding methods applied post-QEWT. We also expect (especially through some of the non-trivial spectral lower bounds of Prop.~\ref{prop:ewt_spectral}) that one might be able to prove useful lower bounds on success probability and amplitude amplification requirements for specific classes of block encodings.

\begin{theorem}[Schur product theorem; Thm.~5.2.1 in \cite{hj_topics_91}]
    If $A, B \in M_n$ are positive semi-definite (PSD), then so is their element-wise product $A \circ B$. Moreover, if $B$ is positive definite (PD), and $A$ has no diagonal entry equal to zero, then $A \circ B$ is also PD. Furthermore, if $A$ and $B$ are both PD, then so is $A \circ B$. As a corollary, note that element-wise polynomials of PSD\slash PD matrices with positive real coefficients are also PSD\slash PD.
\end{theorem}

\begin{proposition}[Submultiplicativity of element-wise product; Thm.~5.5.1 of \cite{hj_topics_91}] 
\label{prop:submult_ewt}
    For any $A, B \in M_{m, n}$ ($m\times n$ matrices with complex coefficients), $\lVert (A \circ B) \rVert \leq \lVert A \rVert \, \lVert B \rVert$, where $\lVert A\rVert$ is the spectral norm, or largest singular value of $A$. We also have the following chain of inequalities:
        \begin{equation}
            \lVert (A \circ B) \rVert \leq r(A)\,c(B) \leq \lVert A \rVert \, \lVert B \rVert,
        \end{equation}
    where $r(A)$, $c(A)$ are the maximum $\ell_2$ norms among rows and columns of $A$ or $B$. The commutativity of the element-wise product establishes the statement with row and column maximum $\ell_2$-norms flipped. Moreover, the (element-wise product) submultiplicativty of the spectral norm given here can be extended to all \emph{unitarily-invariant} norms.
\end{proposition}

\begin{proposition}[Spectral properties of the element-wise product; variously \cite{hj_matrix_analysis_85, hj_topics_91}] \label{prop:ewt_spectral}
    Let $A, B \in M_{n}$ be positive definite. Then $AB = A\circ B$ iff both $A$ and $B$ are diagonal. Moreover, $\det{(A \circ B)} \geq \det{B}\prod_{k} A_{kk}$ (this is called \emph{Oppenheim's inequality}). For $A, B$ PSD, we thus also have $\det{(A \circ B)} \geq (\det{A})(\det{B})$. Finally, as a relatively weak statement, recalling that the element-wise product is a principal submatrix of the Kronecker product, we have the following spectral bounds:
        \begin{equation}
            \lambda_{\text{min}}(A \circ B) \geq \lambda_{\text{min}}(A)\,\lambda_{\text{min}}(B),\quad\quad
            \lambda_{\text{max}}(A \circ B) \leq \lambda_{\text{max}}(A)\,\lambda_{\text{max}}(B),
        \end{equation}
    whenever $A, B$ are PSD. A stronger statement can be made (Thm.~5.3.1 of \cite{hj_topics_91}) for PSD Hermitian $A, B$:
        \begin{equation}
             \lambda_{\text{min}}(A \circ B) \geq  \lambda_{\text{min}}(AB^T),\quad\quad
             \lambda_{\text{min}}(A \circ B) \geq  \lambda_{\text{min}}(AB).
        \end{equation}
\end{proposition}

In the interest of concision we stop here, but refer the interested reader to the in-depth treatment of the element-wise or Hadamard product in \cite{hj_matrix_analysis_85, hj_topics_91} (specifically Ch.~5 of the latter and its references), which provides a huge variety of useful statements, including in the block-matrix case, the study of the relative gain array, stronger constraints than PSD on the input matrices, and the interaction between the element-wise product and so-called `ordinary functions' of matrices, which generalize spectral mappings like QSVT.

%%%%%%%%%%%%%%%%%%%%%%%%%%%%%%%%%%%%%%%%%%%%%
\section{LCU for non-standard matrix products} \label{appx:lcu_nonstandard}

\noindent In the main text we stated a lemma (Lem.~\ref{lem:ew_select}) confirming that the \textsc{select} unitary for the element-wise product up to a degree $d$ could be built, just as for the usual matrix product, using $d$ copies of a controlled block encoding of the desired matrix, as well as $\mathcal{O}(\log{d})$ additional space and $\mathcal{O}(d\log{d})$ single- and two-qubit gates. In this section we both give a constructive proof of this Lem.~\ref{lem:ew_select}, and discuss the insufficiency of a brief argument given in prior work for constructing the same unitary.\footnote{We emphasize that, as our construction occurs above the element-wise power subroutine, we remedy and leave (mostly) unchanged the quoted complexities in the prior work \cite{gychanar_had_prod_25}.}

We provide a brief summary of the oversights in prior work on this topic (Rem.~\ref{rem:subtle_nonstandard}), which can be roughly grouped into questions of access model, properties of the element-wise product, and block encoding properties. We also provide an auxiliary figure (Fig.~\ref{fig:ewt_controlled_form}) which illustrates, at a minimal size, the modifications that can recover the desired construction. Following these discussions we provide the statement and accompanying proof of the main text's Lem.~\ref{lem:ew_select}.

\begin{remark}[Subtleties non-standard \textsc{select} unitaries] \label{rem:subtle_nonstandard}
    The proof of Thm.~S6 in \cite{gychanar_had_prod_25} makes only passing reference to a construction which takes as input block-encoded element-wise powers of a given matrix, and produces the \textsc{select} unitary which coherently applies these powers according to the quantum state of an auxiliary register. The core argument of this work is that the usual construction of \textsc{select} for matrix products (reproduced below in Lem.~\ref{lem:gate_cost_standard}) suffices. We point out three complications in this statement.
    \begin{enumerate}[label=(\alph*)]
        \item The construction of \textsc{select}, toward applying LCU among block encodings, requires \emph{controlled access} to the block encoding unitaries. If a circuit description of these block encodings is known this controllization can be straightforward, while for black-box unitaries controllization can be extremely expensive. Regardless, the controlled access assumption should be noted explicitly where relevant.

        \item The application of Lem.~\ref{lem:gate_cost_standard} (e.g., in the cited \cite{cks_linear_17}) involves the sequential application of $2^j$-th \emph{matrix powers} of a block encoding according to the $j$-th bit $b_j$ of the control register. This sequential application necessarily induces \emph{matrix powers} among these block encodings, leading (by the associativity of matrix multiplication) to the desired total power:
            \begin{equation}
                \prod_{j = 0}^{\log{d}} A^{b_j 2^j} 
                = A^{\sum_{j = 0}^{\log{d}} b_j 2^j}
                = A^{k} \;\text{ where }\; k \equiv b_{\lceil\log{d}\rceil} \cdots b_{1}b_{0}.
            \end{equation}
        The issue immediately encountered when trying to apply this to the element-wise product is that the matrix and element-wise products are \emph{incompatible}, meaning that the matrix multiplication of bit-wise controlled element-wise products is simply incorrect:
            \begin{equation}
                (A^{\circ b_0 2^0})(A^{\circ b_1 2^1})\cdots(A^{\circ b_{\log{d}} d}) \neq A^{\circ (b_0 2^0 + b_1 2^1 + \dots + b_{\log{d}}d)}.
            \end{equation}
        
        \item Finally, as implied in the previous point it is not even clear how we should define $A^{\circ 0}$, as the identity element for the standard matrix product and the element-wise product are not identical. Consequently, \emph{even if} the scheme of sequentially applying bit-wise controlled element-wise powers worked, we could not use standard controlled block encodings without further modification. As a small example (whose correct circuit is given in Fig.~\ref{fig:ewt_controlled_form}) one can compare the following two \textsc{select} unitaries, where the first computes the desired bit-wise-controlled element-wise product with the matrix product identity, and the second with the element-wise product identity:\footnote{Where again we assume that $r_2$ and $r_4$ are post-selected to the all zeros state, which means $r_0$ is used catalytically, and $r_1$ experiences the identity.}
            \begin{align}
                V &= |0\rangle\langle 0|_{r_5} \otimes (A \circ I)_{r_3} + |1\rangle\langle 1|_{r_5} \otimes (A^{\circ 2})_{r_3},\\
                V' &= |0\rangle\langle 0|_{r_5} \otimes (A^{\circ 1})_{r_3} + |1\rangle\langle 1|_{r_5} \otimes (A^{\circ 2})_{r_3}.
            \end{align}
        Thus, we have to make sure that in the case that a control bit for a constituent block encoding unitary is zero that we \emph{either} apply an element-wise product with the element-wise identity, or a matrix product with the matrix product identity, such that the overall effect on the product is null. It turns out (as block encoding the element-wise identity incurs a bad subnormalization) that we can do the latter comparatively simply as discussed in the proof of Lem.~\ref{lem:ew_select_proof}.
    \end{enumerate}
\end{remark}

\begin{figure}
    \centering
    \includegraphics[width=\linewidth]{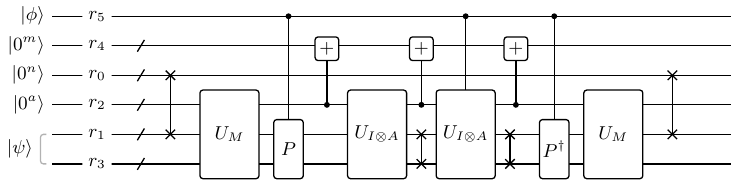}
    \caption{The quantum circuit to apply, controlled on the single-qubit register $r_5$, either $A = A^{\circ 1}$ or $(A\circ A) = A^{\circ 2}$ to the register $r_1 r_3$ (assuming successful post-selection of $r_0 r_2 r_4$ in the all-zeros state). Note that if the permutations were \emph{not controlled} that this would apply either $(A \circ I)$ or $(A \circ A)$ for $r_5 = |0\rangle, |1\rangle$ respectively, which is equivalent iff $A$ is diagonal. In the controlled version of our recursive construction these controlled permutations, at each level, would become \emph{multi-controlled permutations} applied only when \emph{not all} of the inputs $A, B$ are the identity matrix.}
    \label{fig:ewt_controlled_form}
\end{figure}

\begin{lemma}[Gate cost for controlled matrix product gadgets; Lem.~8 from \cite{cks_linear_17}] \label{lem:gate_cost_standard}
    Let $U = \sum_{k = 0}^{N} |k\rangle\langle k|\otimes Y^k$, where $Y$ is a unitary with gate complexity $G$. Let the gate complexity of $Y^{2^j}$ be $G_j \leq 2^j G$. Then the gate complexity of $U$ is $\mathcal{O}(NG)$.

    \begin{proof}
        Let $n = \lfloor \log{N} \rfloor$, and note that the unitaries $Y^{2^j}$ for $j \in \{0, \dots, n\}$, which have query complexity $G_j$, also have controllized versions with query complexity $G_j$. The desired unitary $U$ is just the sequential application of controlled-$Y^{2^j}$ to the $j$-th qubit of a register of size $n$ encoding $k \in \{0, \cdots N - 1\}$ in binary. The total gate complexity is thus
            \begin{equation}
                \mathcal{O}\Big(\sum_{j = 0}^{n} G_j\Big)
                \leq \mathcal{O}\Big(\sum_{j = 0}^{n} 2^j G\Big)
                = \mathcal{O}(NG),
            \end{equation}
        which is evidently linear in both the maximum degree and the cost of $Y$.
    \end{proof}
\end{lemma}

The morals of Lem.~\ref{lem:gate_cost_standard} are clear: the gate cost of coherently computing controlled matrix products of a desired unitary is not much larger than the gate cost of computing the highest degree constituent matrix product, and the realizing circuit needs access only to certain (sufficiently cheap) controlled powers of this matrix product. This result, while referred to in the LCU construction of Thm.~S6 of \cite{gychanar_had_prod_25}, is not sufficient to establish a similar construction in the case of \emph{element-wise} powers for the reasons discussed in Rem.~\ref{rem:subtle_nonstandard}.
Here we clarify and remedy these problems by providing a modified version of Lem.~\ref{lem:gate_cost_standard}. The key enabling factor is that our desired matrix function has a uniform circuit which, viewed as a function of (multiple copies of) the block encoding input, doesn't depend on their order. Moreover, given the known identity that the element-wise product is a submatrix of the Kronecker product (which is itself a matrix product) we have a cheap way of (coherently) ensuring that matrix-product identities rather than element-wise product identities can be applied, simply by controllizing the permutations used to realize this identity. It is an interesting open question whether other \textsc{select} unitaries for indexed matrix functions can be built in a similar way.

\begin{lemma}[Building \textsc{select} for the element-wise product] \label{lem:ew_select_proof}
    Let $U_A$ a $(1, a, 0)$ block encoding of an $n$-qubit linear operator $A$ and $d \in \mathbb{N}$. Then the \textsc{select} unitary for element-wise powers of $A$ up to degree $d$, i.e.,
    \begin{equation}
        U_\text{\textsc{select}} \equiv \sum_{k = 0}^{d} 
        |k\rangle\langle k |\otimes U_{A^{\circ k}} + \cdots,
    \end{equation}
    where the unwritten terms pad to the smallest power of two greater than $d$ (see Eq.~\ref{eq:lcu_select_def}), can be built using $d$ copies of \emph{controlled} $U_A$, $\lceil\log{d}\rceil$ qubits of additional space, and $\mathcal{O}(d\,\lceil\log{d}\rceil)$ single and two qubit gates. Note here that the notation $A^{\circ 0}$ is taken to be the \emph{usual identity matrix} on $n$ qubits.

    \begin{proof}
        The statement follows from the assumption that the circuit for the unitary-valued function we want to coherently apply $g(k, A) = A^{\circ k}$ is uniform, i.e., that it depends only on the number of copies of (block encodings of) $A$ it is provided. If this is the case, we can build our \textsc{select} from copies of controlled-$U_{A}$ by the same bit-wise control scheme with a few addenda; i.e., assuming for now that $d$ is a power of $2$,\footnote{If not, we can pad our construction to the next largest power of two and simply truncate our LCU, as is common in prior work \cite{cw_lcu_12}.} we fill the first $d/2$ slots of our uniform circuit with copies of controlled-$U_{A}$, where the control qubit is the most significant bit of the additional register ($r_5$) of size $\log{d}$, and the target is the usual registers $(r_1 r_2 r_3)$ for our block encoding. The next $d/4$ slots are then controlled by the second-most significant bit, and so on. This ensures that, coherently, when $|k\rangle$ is encountered in the control register (encoded in binary with consistent endianness), that $k$ total copies of $U_{A}$ are provided to the uniform circuit, yielding the desired $g(k, A)$. 

        What remains is then to argue that we can also simultaneously enforce the matrix product identity (which is seen by the circuit instead of $A$ when a given $U_{A}$'s control qubit is in the state $|0\rangle$) to act like the element-wise product identity. The simplest example of this was given in Fig.~\ref{fig:ewt_controlled_form}, e.g., at the lowest level of our recursion the \emph{permutation unitaries} $P, P^\dagger$ are controlled by our register $r_5$ such that, whenever \emph{either of the control qubits} for the two inputs $U_{(I\otimes A)}$ or $U_{(I\otimes B)}$ (as in Fig.~\ref{fig:full_prod}) are in the state $|0\rangle$, then neither of $P, P^\dagger$ is applied. For higher levels of recursion these $P, P^\dagger$ are turned off only when \emph{all control qubits} contributing to \emph{either input block encoding} from lower levels are zero. These are simply multi-qubit controlled permutations, which catalytically use at most an additional $\log{d}$-qubits of additional space and access to a qubit-controlled permutation (e.g., see Sec.~4.3 of the standard \cite{nc_textbook_11}).
        
        The total gate complexity of our construction is thus $\mathcal{O}(d\log{d})$, as there are $d$ total slots in our uniform circuit, and each controlled $U_{A}$ comes paired with a constant depth circuit acting over at most $\mathcal{O}(\log{d})$ qubits.
    \end{proof}
\end{lemma}

The mechanisms of Lem.~\ref{lem:ew_select_proof} can be used to amend the error in prior work raised when trying to build an element-wise \textsc{select} only from an application of Lem.~\ref{lem:gate_cost_standard} (taken from \cite{cks_linear_17}). We give this amendment in the dedicated Prop.~\ref{prop:amended_select}; note that in this statement we are building a bit-wise-controlled element-wise product under the assumption of access to block encodings of $A^{\circ 2^j}$ for varying $j$. For ease of visualization we also provide a figure of the circuit corresponding to our construction in the two bit case in Fig.~\ref{fig:ex_ewt_select}. In our main construction (Lem.~\ref{lem:ew_select_proof}) we assume only access to block encodings of $A = A^{\circ 1}$, and show that it is not strictly necessary to go through an intermediate bit-wise construction like the one below.

\begin{proposition}[Amended construction for bit-wise element-wise \textsc{select}] \label{prop:amended_select}
    Let $U_A$ an $(1, a, 0)$ block encoding of an $n$-qubit operator $A$, and $U_{A^{\circ k}}$ an $(1, a_k, 0)$ block encoding of the $k$-th element-wise product of $A$ with itself. Then given query access to $U_{A^{\circ 2^j}}$ for $j \in \{0, 1, \cdots, m\}$ we can construct the bit-wise \textsc{select} unitary for the element-wise product:
    \begin{equation} \label{eq:b_bit_select}
        \sum_{b_0, b_1, \dots b_m = 0}^{1} |b_0 b_1 \cdots b_m \rangle \langle b_0 b_1 \cdots b_m | \otimes U_{A^{\circ(2^0 b_0 + 2^1 b_1 + \cdots + 2^{m} b_m)}},
    \end{equation}
    using a single query to each \emph{controlled} $U_{A^{\circ 2^j}}$, $\mathcal{O}(n m)$ additional qubits of space, and $\mathcal{O}(n m)$ additional single and two-qubit gates. Finally here we use the notation $A^{\circ 0}$ to refer to the \emph{matrix product} identity $I_{2^n}$ and not the element-wise product identity.

    \begin{proof}
        Merely controllizing all of the unitaries in the original element-wise product construction is not sufficient to achieve our desired \textsc{select} unitary, due to the problems brought up in Rem.~\ref{rem:subtle_nonstandard}. Consequently we also need to controllize the permutations used to bring the element-wise product to the top-left block of our block encodings, \emph{as well as} the swap operation used to turn a block encoding of $(I_{r_1}\otimes A_{r_3})$ to one of $(A_{r_1}\otimes I_{r_3})$. Moreover, this controllization needs to be such that it occurs whenever \emph{either} of the input bit-wise controlled block encodings is trivial, which is accomplished by introducing an additional qubit of space and the use of Toffoli gates to achieve a bit-wise AND.\footnote{A bit-wise XNOR using two CNOTs would also suffice here.} For the two bit case this is depicted in Fig.~\ref{fig:ex_ewt_select}, from which the general circuit follows simply. We note that the additional space used\footnote{Here the additional space is logarithmic in the maximum degree $(2^{m+1} - 1)$ of the element-wise product only because we have assumed low-space block encodings for each $2^k$-th element-wise power of $A$.} scales only in the bit-length $m$ (multiplied by the size of the register containing the action of $A$, $n$), yielding the quoted $\mathcal{O}(n m)$ qubit and gate complexities.
    \end{proof}
\end{proposition}

\begin{figure}
    \centering
     \includegraphics[width=0.65\linewidth]{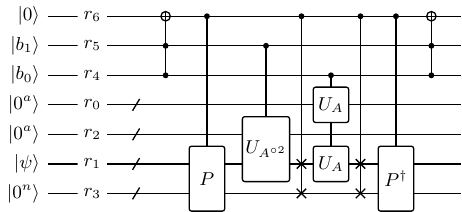}
    \caption{Circuit computing the two-bit \textsc{select} unitary (\ref{eq:b_bit_select}) for the element-wise product given controlled access to $(\ast, a, \ast)$-block encodings of $n$-qubit operators $A$ and $A^{\circ 2}$. Here $r_4, r_5$ are the LSB and MSB respectively for the \textsc{select} control, $r_1$ contains the input state, $r_0, r_2$ are the auxiliary $a$-qubit registers for the two block encodings, and $r_3$ and $r_6$ are auxiliary space used to achieve Kronecker products of block encodings and to XNOR the bits $b_0$ and $b_1$ respectively. The controlled permutation $P$ is the same as in (Def.~\ref{def:submat_permut}), and places the desired element-wise product into the top-left block (i.e., post-selecting on $|0\rangle_{r_0 r_2 r_3}$). Finally, note that $U_A$ is split into two boxes to clearly show the two registers it acts on, following a visual convention established in \cite{lw_lecture_notes_26}.}
    \label{fig:ex_ewt_select}
\end{figure}

To end this section we also include, for completeness, a separate, simpler technique not mentioned in previous literature for reducing the auxiliary space use for quantum element-wise transforms slightly from $\mathcal{O}(ad + nd + \log{d})$ to $\mathcal{O}(a + n(d - 1) + \log{d})$. We call this a \emph{half-compressed} QEWT (Thm.~\ref{thm:half_qewt}) in comparison to the main result of our work, which uses $\mathcal{O}(a + n\log{d})$ additional space. The benefits of this construction include its simplicity, close-resemblance to the usual compression gadget applied to Kronecker products, and lower depth. Just as in our main construction, the half-compressed QEWT can be slotted in to all known applications of element-wise transforms of block encodings exposited in prior work.

\begin{lemma}[Block encoding Kronecker products] \label{lem:be_kronecker_prod}
    Let $U_A$ an $(\alpha, a, \varepsilon)$ block encoding of an $n$-qubit operator $A$. Then we can construct an $(\alpha^d, a + \log{d}, \delta)$ block encoding of $A^{\otimes d}$ (the $d$-th order Kronecker product of $A$ with itself) using $d$ queries to $U_A$ and total space $nd + a + \log{d}$. Here $\delta$ propagates the same as for matrix products.

    \begin{proof}
        The construction follows simply from recognition that $A^{\otimes d}$ is the matrix product of $B_{k} \equiv I_{0}\otimes \cdots \otimes A_{k}\otimes  \cdots \otimes I_{d - 1}$, where $A$ appears in the $k$-th register and the identity matrices are on $n$-qubits. Moreover, block encoding the Kronecker product of an operator and the identity matrix is trivial, and requires only allocating additional unused qubits on which the identity matrix `acts.' Once we have reduced to the matrix product setting, the usual compression gadget (Lem.~\ref{lem:comp_gadget}) applies, reducing auxiliary space usage from the naïve $ad$ to $a + \log{d}$ qubits.
    \end{proof}
\end{lemma}

\begin{theorem}[Half-compressed QEWT] \label{thm:half_qewt}
    Let $U_A$ an $(\alpha, a, \varepsilon)$ block encoding of an $n$-qubit operator $A$. Then we can construct an $(\alpha^d, a + n(d-1) + \log{d}, \delta)$ block encoding of $A^{\circ d}$ (the $d$-th order element-wise product of $A$ with itself) using $d$ queries to $U_A$.

    \begin{proof}
        For element-wise powers, the result follows from Lem.~\ref{lem:be_kronecker_prod}, to which we apply the $\mathcal{O}(d)$ pairs of $n$-qubit qubit-wise CNOT gates to locate the copy of $A^{\circ d}$ inside the $d$-th order Kronecker product. Absorbing the unused $n(d-1)$ qubits into the auxiliary space of the final block encoding completes the construction. Proof for the general QEWT then follows by direct application of Thm.~\ref{thm:qewt} to our half-compressed element-wise product gadgets (built into a \textsc{select} unitary by the methods of Lem.~\ref{lem:ew_select_proof}). This immediately improves the space complexity in $d$, albeit not exponentially as in our fully compressed construction. 
    \end{proof}
\end{theorem}

%%%%%%%%%%%%%%%%%%%%%%%%%%%%%%%%%%%%%%%%%%%%%
\section{Error propagation in block encodings} \label{appx:error_prop}

\noindent This appendix has been reserved to offload repetitive proofs in the main paper regarding the handling of errors in input block encodings. We take some time to discuss how handling this error propagation changes with respect to the matrix functions we consider, and rectify some common oversights in prior work. Where appropriate we have referred results quoted in the main theorems to their respective proofs here.

\begin{remark}[Error propagation in various block encoding products] \label{rem:error_prop_prod}
    We take a brief detour to summarize (and in some places rectify) standard arguments used to propagate error from the input block encodings to a quantum algorithm and its output block encoding. 
    
    Let $U_A$ and $U_B$ respectively $(\alpha, a, \varepsilon)$ and $(\beta, b, \delta)$ block encodings of matrices $A$ and $B$ square and of the same dimension. Recall that this means, wlog for $U_A$, that the unitary contains an $\alpha \geq 1$ subnormalized $\varepsilon$-approximation (in the operator norm\footnote{We discuss elsewhere in the text that, in the context of element-wise products it may sometimes make more sense to use element-wise max norms, or variant Frobenius norms.}) to $A$, i.e.
    \begin{equation}
        \lVert A - \alpha(\langle 0^a|\otimes I) U_A(|0^a\rangle\otimes I)\rVert \leq \varepsilon.
    \end{equation}
    It is well known how to propagate error with respect to the equipped associative algebra (here sums and products)\footnote{We note here that the auxiliary space use for the sum operation involves padding the smaller register among $a, b$, and that for multiplication we can reduce $a + b$ to $\max(a,b) + 1$ using the compression gadget.}
    \begin{alignat}{2}
        (\alpha, a, \varepsilon), (\beta, b, \delta) \;\;&\mapsto\;\; (\alpha + \beta, \max{(a, b)} + 1, \delta + \varepsilon) \quad&&\text{(sum, $A + B$)}\\
        (\alpha, a, \varepsilon), (\beta, b, \delta) \;\;&\mapsto\;\; (\alpha\beta, a + b, \alpha\delta + \beta\varepsilon + \varepsilon\delta) \quad&&\text{(matrix product, $AB$)}
    \end{alignat}
    It turns out that these two rules (and their iterated variants, captured by LCU and the compression gadget respectively) tell us most of what we need to know, as the other common matrix products we consider are mainly sub-matrices, block-versions, or basis-transformed variants of the above manipulations. We recall a subset of the products covered in this work:
    \begin{alignat}{3}
        &\otimes &\quad &\text{\texttt{\textbackslash otimes}} &\quad &\text{(Kronecker/tensor product)},\\
        &\circ &\quad &\text{\texttt{\textbackslash circ}} &\quad &\text{(Element-wise product)},\\
        &\boxtimes &\quad &\text{\texttt{\textbackslash boxtimes}} &\quad &\text{(Tracy--Singh product; Def.~\ref{def:ts_prod})},\\
        &\odot &\quad &\text{\texttt{\textbackslash odot}} &\quad &\text{(Khatri--Rao product; Def.~\ref{def:kr_prod})},\\
        &\ast &\quad &\text{\texttt{\textbackslash ast}} &\quad &\text{(Convolution)}.
    \end{alignat}
    Note that the Kronecker product $(A \otimes B)$ is simply the matrix product of $(A\otimes I)$ and $(I\otimes B)$, which can be built from $(\alpha, a, \varepsilon)$ and $(\beta, b, \delta)$ block encodings of $A$ and $B$ respectively as $\lVert M \otimes N\rVert \leq \lVert M\rVert\,\lVert N\rVert$, and thus propagates error in the same way as the matrix product. Likewise, given that the element-wise product is a submatrix of the Kronecker product, its accrued propagation can be no-worse, and thus we can also use the usual matrix product error propagation formula. Moreover, as the TS product is a row- and column-permutation of the Kronecker product, and the KR product is a submatrix of the TS product, these also obey the matrix product error propagation formula. Finally, the 2D circular convolution we consider is a basis transformation of the element-wise product, and thus does not change the operator norm.

    To belabor one more point, we explicitly reproduce the error propagation formula in the case of multiplying together $K$ block encoded matrices at once. This treatment has been the subject of a small error since early publications, and spans both of \cite{gslw_19, gychanar_had_prod_25} (separated by at least six years). Here we consider a collection of $K$ block encoded matrices $A_0, A_1, \dots, A_{K - 1}$, assumed conformable to matrix multiplication, and with respective block encoding parameters $(\alpha_k, a_k, \varepsilon_k)$; then we can achieve a block encoding of their product $A_0A_1\cdots A_{K - 1}$ with the following parameters:
    \begin{align}
        (\alpha_k, a_k, \varepsilon_k)_{k \in \{0, \dots, K - 1\}} &\;\;\mapsto\;\; \left(\prod_{k = 0}^{K- 1} \alpha_k, [\ast], \prod_{k = 0}^{K - 1} (\alpha_k + \varepsilon_k) - \prod_{k = 0}^{K - 1} \alpha_k \right),
    \end{align}
    where the auxiliary space usage is denoted $[\ast]$ as it depends on whether a compression scheme is used. A full derivation of this result is contained in Prop.~9.12 of the developing, public lecture notes \cite{lw_lecture_notes_26}. Assuming that $\varepsilon_k \ll \alpha_j$ for all $j, k$, the leading order terms of the approximation error of this expression are what is supplied by a large fraction of extant literature.
\end{remark}

In the interest of a self-contained appendix, we calculate the propagated error in our element-wise product subroutine (Thm.~\ref{thm:iterated_prod}) in the following Prop.~\ref{prop:iter_ewt}, and then the error in our full quantum element-wise transform (Thm.~\ref{thm:qewt}) in the following Prop.~\ref{prop:error_lcu_ewt}. These are different from the error quoted in the non-compressed Thm.~S6 of \cite{gychanar_had_prod_25} for the reasons quoted in Rem.~\ref{rem:error_prop_prod}.

\begin{proposition}[Error in iterated element-wise product of block encodings] \label{prop:iter_ewt}
    Let $A \in \mathbb{C}^{N\times N}$ for $N = 2^n$ and $U_A$ a $(\alpha, a, \varepsilon)$ block encoding of $A$. Then we can construct a $(\alpha^{d}, da + (d - 1)n, \gamma)$ block encoding of $A^{\circ d}$, where $\gamma = (\alpha + \varepsilon)^d - \alpha^d$ (approximately $d\alpha\varepsilon$ in the case of large $\alpha$ and small $\varepsilon$).
    \begin{proof}
        The linear-space result is given in \cite{gychanar_had_prod_25}, and we reproduce their brief argument here. Let $A' = \alpha \langle 0^a | U_A |0^a\rangle$ and $E \equiv A - A'$, which by definition must satisfy $\lVert E \rVert \leq \varepsilon$. Then the error $\delta$ between the ideal element-wise self-product and the achieved one is bounded above according to (in the $d = 2$ case)
            \begin{align}
                \delta
                &=
                \lVert A \circ A - \langle 0^n|P(A'\otimes A')P^\dagger| 0^n\rangle\rVert,\\
                &\leq
                \lVert A \circ A - \langle 0^n|P(A'\otimes A - A'\otimes A + A'\otimes A')P^\dagger| 0^n\rangle\rVert,\\
                &\leq
                \lVert A \circ A - \langle 0^n|P(A'\otimes A)P^\dagger| 0^n\rangle\rVert + \lVert\langle 0^n|P(A'\otimes A - A'\otimes A')P^\dagger| 0^n\rangle\rVert,\\
                &\leq
                \lVert\langle 0^n|P(E \otimes A)P^\dagger| 0^n\rangle\rVert + \lVert\langle 0^n|P(A' \otimes E)P^\dagger| 0^n\rangle\rVert,\\
                &\leq (\alpha + \varepsilon)\varepsilon + \alpha\varepsilon,\\
                &= 2\alpha \varepsilon + \varepsilon^2,
            \end{align}
        where we have just applied the triangle inequality and the definition of $E$. Note carefully\footnote{This is an oversight that appears in the exposition of error propagation for block encoding multiplication in both \cite{gslw_19, gychanar_had_prod_25}, and is carefully amended in the recent (and in-progress) lecture notes of Lin and Wiebe \cite{lw_lecture_notes_26}.} that we have $\lVert A \rVert \leq \alpha + \varepsilon$, and \emph{not} $\lVert A \rVert \leq \alpha$, as our target operator is allowed to violate our block encoding subnormalization constraint by an additive factor of $\varepsilon$. The case for general $d$ follows immediately by an inductive argument, identical to that of Prop.~9.12 in \cite{lw_lecture_notes_26}, or reproduced in Rem.~\ref{rem:error_prop_prod}.
    \end{proof}
\end{proposition}

\begin{proposition}[Error in LCU of element-wise products of block encodings] \label{prop:error_lcu_ewt}
    Let $U_A$ an $(\alpha, a, \varepsilon)$ block encoding of a square matrix $A$ on $n$ qubits. Then the LCU according to a polynomial function $f$ with complex coefficients $c_k$, for $k \in \{1, 2, \dots d\}$ applied to element-wise powers of $A$ yields a block encoding unitary $U_{f^\circ(A)}$ with the following transformed parameters:
    \begin{align}
        &(\alpha, a, \varepsilon)  \;\mapsto\; \Biggl(\gamma, a + n\log{d}, \sum_{k = 1}^{d} \varepsilon_{k} \Biggr),\\
        &\gamma \equiv \sum_{k = 1}^{d} |c_k|, \quad\quad
        \varepsilon_{k} \equiv |c_k| \Big[(\alpha + \varepsilon)^k - \alpha^k\Big].
    \end{align}
    \begin{proof}
        The sub-normalization according to the $\ell_1$-norm $\gamma$ of the coefficients of $f$ is a basic consequence of the LCU construction and the unitarity of $U_{f^{\circ}(A)}$. For each individual element-wise power we have computed $\varepsilon_k$ directly from Prop.~\ref{prop:iter_ewt} and the observations in Rem.~\ref{rem:error_prop_prod}; these errors accrete linearly in the LCU, leading to the resulting total error. The auxiliary space usage follows from our recursive construction in the main Thm.~\ref{thm:qewt}.
    \end{proof}
\end{proposition}

\end{document}